\documentclass[12pt]{article}
\pdfoutput=1

\usepackage{color}
\usepackage{graphicx}
\usepackage{amsmath}
\usepackage{amssymb}
\usepackage{xspace}
\usepackage[small]{subfigure}
\usepackage[numbers,compress]{natbib}
\usepackage[hyperfootnotes=false]{hyperref}

\usepackage{mciteplus} 
\usepackage{dcolumn}
\usepackage{bm}
\usepackage{verbatim}
\usepackage{amscd}
\usepackage{amsfonts}
\usepackage{setspace}
\usepackage{amsthm}
\usepackage{enumerate}
\usepackage{mathtools}

\newlength{\abstractwidth}
\renewcommand{\title}[1]{\vbox{\center\bf{\Large{#1}}}\vspace{5mm}}
\renewcommand{\author}[1]{\vbox{\center#1}\vspace{5mm}}
\newcommand{\address}[1]{\vbox{\center\em#1}}
\newcommand{\email}[1]{\vbox{\center\tt#1}\vspace{5mm}}
\newcommand{\sbreak}{
    \begin{center}
        $\blacklozenge$$\blacklozenge$$\blacklozenge$$\blacklozenge$
    \end{center}
}
\newcommand{\be}{\begin{equation}}
\newcommand{\ee}{\end{equation}}
\newcommand{\bi}{\begin{itemize}}
\newcommand{\ei}{\end{itemize}}

\newcommand{\tb}[1]{\textbf{#1}}
\newcommand{\mc}[1]{\mathcal{#1}}
\newcommand{\ket}[1]{| #1 \rangle }
\newcommand{\bra}[1]{\langle #1 |}
\newcommand{\trr}[1]{\text{tr}\,\{ #1 \}}

\newcommand{\tr}{\text{tr}}

\newcommand{\T}[1]{\langle #1 \rangle}
\newcommand{\Fk}{F_{\mathcal{E}}^{(k)}}

\theoremstyle{plain}
\newtheorem{theorem}{Theorem}
\theoremstyle{plain}
\newtheorem{lemma}{Lemma}
\theoremstyle{plain}
 
\theoremstyle{plain}

\theoremstyle{remark}

\newtheorem{conjecture}{Conjecture}

\theoremstyle{definition}

\theoremstyle{definition}

\theoremstyle{definition}

\theoremstyle{definition}

\begin{document}

\begin{titlepage}

\begin{center}
\hfill \\
\hfill \\
\vskip .5cm

\title{Chaos and complexity by design}
\author{Daniel A. Roberts${}^{a,b}$ and Beni Yoshida${}^{c}$}
\address{$^{a}$ Center for Theoretical Physics {\it and} \\  Department of Physics, Massachusetts Institute of Technology \\ Cambridge, Massachusetts 02139, USA

\vspace{10pt}

$^{b}$ School of Natural Sciences, Institute for Advanced Study,\\ Princeton, NJ 08540, USA

\vspace{10pt}

$^{c}$ Perimeter Institute for Theoretical Physics,\\ Waterloo, Ontario N2L 2Y5, Canada
}

\email{roberts@ias.edu, byoshida@perimeterinstitute.ca}

\end{center}

\begin{abstract}
We study the relationship between quantum chaos and pseudorandomness by developing probes of unitary design. A natural probe of randomness is the ``frame potential,'' which is minimized by unitary $k$-designs and measures the $2$-norm distance between the Haar random unitary ensemble and another ensemble. A natural probe of quantum chaos is out-of-time-order (OTO) four-point correlation functions. We show that the norm squared of a generalization of out-of-time-order $2k$-point correlators is proportional to the $k$th frame potential, providing a quantitative connection between chaos and pseudorandomness. Additionally, we prove that these $2k$-point correlators for Pauli operators completely determine the $k$-fold channel of an ensemble of unitary operators. Finally, we use a counting argument to obtain a lower bound on the quantum circuit complexity in terms of the frame potential. This provides a direct link between chaos, complexity, and randomness. 
\end{abstract}

\end{titlepage}


\tableofcontents
\baselineskip=17.63pt

\section{Introduction}

Random unitary operators have often been used to approximate chaotic dynamics. Notably, in the context of black holes Hayden and Preskill used a model of random dynamics to show such systems can be efficient information scramblers; initially localized information will quickly thermalize and become thoroughly mixed across the entire system \cite{Hayden07}.  It was later conjectured \cite{Sekino08} and then proven \cite{Shenker:2013pqa,Maldacena:2015waa} that black holes are the fastest scramblers in nature, suggesting their dynamics must have much in common with random unitary evolution. Such ``scrambling'' is a byproduct of strongly-coupled chaotic dynamics \cite{Lashkari13,Almheiri13b,Hosur:2015ylk}, and so the work of \cite{Hayden07} suggests there should be strong quantitative connection between such chaos and pseudorandomness.\footnote{
    Scrambling has a close connection to the decoupling theorem, see e.g.~\cite{Berta, Brown15}. Also, see \cite{Chamon14,Hayden:2016cfa,Nahum:2016muy} for studies connecting randomness to entanglement.
}

The connection between pseudorandomness and chaotic dynamics can also be understood at the level of operators. For example, consider $W$ a local operator of low weight (e.g. a Pauli operator acting on a single spin). With a chaotic Hamiltonian $H$, the operator $W(t)=e^{iHt} W e^{-iHt}$ will be a complicated nonlocal operator that has an expansion as a sum of products of many local operators with an exponential number of terms each with a pseudorandom coefficient \cite{Roberts:2014isa}. We can gain intuition for this by considering the Baker-Campbell-Hausdorff expansion of $W(t)$
\be
   W(t) = \sum_{j=0}^\infty \frac{(it)^j}{j!}\underbrace{[H, \dots[H}_j,W\underbrace{] \dots]}_j. \label{eq:BCH}
\ee
If $H$ is $q$-local and sufficiently ``generic'' to contain all possible $q$-local interactions, then the $j$th term will consist of a sum of roughly $\sim (n/q)^{qj}$ terms each of weight ranging from $1$ to $\sim j(q-1)$, where we assume the system consists of $n$ spins so that the Hilbert space is of dimension $d=2^n$: 
\begin{itemize}
    \item At roughly $j \sim n / (q-1)$,  there will be many terms in the sum of weight $n$. These terms are delocalized over the entire system. For a system without spatial locality, the relationship between time $t$ and when the $j$th term becomes $O(1)$ is roughly $t \sim \log j$. The timescale $t \sim O(\log n)$ for the operator to cover the entire system is indicative of fast-scrambling behavior.
    \item At around $j \sim 2n / q \log(n/q)$, the total number of terms will reach $2^{2n}$, equal to the total number of orthogonal linear operators acting on the Hilbert space.\footnote{The number of terms is actually not a well defined quantity, since it can change under local rotations of spin, and should possibly instead consider the minimum number of terms under all local changes of basis. For chaotic systems, the distinction should not be important, since there will not exist a local change of basis where the number of terms drastically simplifies. We thank Douglas Stanford who thanks Juan Maldacena for raising this point.} Even after the operator covers the entire system, it continues to spread over the unitary group (though possibly only until a time roughly $O(\log n) +$ a constant). 
\end{itemize}
Furthermore, the coefficient of any given term will be incredibly complicated, depending on the details of the path through the interaction graph and the time. 
Over time, $W(t)$ should cover the entire unitary group (possibly quotiented by a group set by the symmetries of the Hamiltonian). At sufficiently large $t$, one might even suspect that for many purposes $W(t)$ can be approximated by a random operator $\tilde{W}\equiv U^\dagger W U$, with $U$ sampled randomly from the unitary group.\footnote{Of course, for a less generic and sparser Hamiltonian e.g. a free system, the expansion of $W(t)$ will organize itself in a way such that the commutator in Eq.~\eqref{eq:BCH} does not produce an exponential number of terms. Said another way, the terms in the expansion of $W(t)$ will only cover a small subset of the unitary group, and the assumption of uniform randomness will not hold.} If this is true, then we would say that $W(t)$ behaves pseudorandomly.

\subsubsection*{Chaos}

This pattern of growth of $W(t)$ can be measured by a second local operator $V$. For example, the group commutator of $W(t)$ with $V$, given by $W(t)^\dagger\, V^\dagger \, W(t)\, V$, measures the effect of the small perturbation $V$ on a later measurement of $W$. In other words, it is a measure of the butterfly effect and the strength of chaos:
\begin{itemize}
    \item If $W(t)$ is of low weight and few terms, then $W(t)$ and $V$ approximately commute $[W(t),V]\approx 0$, and the operator $W(t)^\dagger\, V^\dagger \, W(t)\, V$ is close to the identity. 
    \item If instead the dynamics are strongly chaotic, $W(t)$ will grow to eventually have a large commutator with all other local operators in the system (in fact, just about all other operators), and so $W(t)^\dagger\, V^\dagger \, W(t)\, V$ will be nearly random and have a small expectation in most states.
\end{itemize}
Thus, the decay of out-of-time-order (OTO) four-point functions of the form
\be
    \langle W(t)^\dagger\, V^\dagger \, W(t)\, V \rangle = \langle U(t)^\dagger W^\dagger U(t) ~ V^\dagger ~ U(t)^\dagger W U(t) ~ V \rangle, \label{oto-4pt-def}
\ee 
can act as a simple diagnostic of quantum chaos, where 
$U(t)=e^{-iHt}$ is the unitary time evolution operator, and the correlator is usually evaluated on the thermal state $\langle \cdot \rangle \equiv \tr \, \{ e^{-\beta H} \, \cdot \, \} / \tr \, e^{-\beta H}$ \cite{Larkin:1969abc,Almheiri13b,Shenker:2013pqa,Kitaev:2014t1,Maldacena:2015waa}.\footnote{Usually $W, V$ are taken to be Hermitian, but here we will more generally allow them to be unitary.} For further discussion, please see a selection (but by all means not a complete set) of recent work on out-of-time-order four-point functions and chaos \cite{Shenker:2013pqa,Shenker:2013yza,Roberts:2014isa,Roberts:2014ifa,Kitaev:2014t1,Shenker:2014cwa,Maldacena:2015waa,Hosur:2015ylk,Stanford:2015owe,Fitzpatrick:2016thx,Gu:2016hoy,Caputa:2016tgt,Swingle:2016var,Perlmutter:2016pkf,Blake:2016wvh,Roberts:2016wdl,Blake:2016sud,Swingle:2016jdj,Huang:2016knw,fan2016out,Halpern:2016zcm}.

For sufficiently chaotic systems and sufficiently large times, the correlators Eq.~\eqref{oto-4pt-def} will reach a floor value equivalent to the substitution, $W(t) \to U^\dagger W U$ with $U$ chosen randomly. 
Furthermore, it can be shown \cite{Hosur:2015ylk} that the decay of correlators Eq.~\eqref{oto-4pt-def} implies the sort of information-theoretic scrambling studied by Hayden and Preskill in \cite{Hayden07}. This explains why the random dynamics model of \cite{Hayden07} was such a good approximation for strongly-chaotic systems, such as black holes. However, are out-of-time-order four-point functions Eq.~\eqref{oto-4pt-def} actually a sufficient diagnostic of chaos?

In \cite{Hayden07} the authors did not actually require the dynamics to be a uniformly random unitary operator sampled from the Haar measure on the unitary group. Instead, it would have been sufficient to sample from a simpler ensemble of operators that could reproduce only a few moments of the larger Haar ensemble.\footnote{Inspection of Eq.~\eqref{oto-4pt-def} suggests that only two moments are required since there are only two copies of $U$ and two copies of $U^\dagger$ in the correlator. We will explain this more carefully in the rest of the paper.} Of course, there may be other finer-grained information theoretic properties of a system that are dependent on higher moments and would require a larger ensemble to replicate the statistics of the Haar random dynamics. If random dynamics is a valid approximation for computing some, but not all, of these finer-grained quantities then they can represent a measure of the degree of pseudorandomness of the underlying dynamics. In this paper, we will make some progress in developing some of these finer-grained quantities and therefore connect measures of chaos to measures of randomness.

\subsubsection*{Unitary design}

The extent to which an ensemble of operators behaves like the uniform distribution can be quantified by the notion of unitary $k$-designs \cite{DiVincenzo02, Emerson05, Ambainis07, Gross07, Dankert09}.\footnote{N.B. in the literature, these are often referred to as unitary $t$-designs. Here, we will reserve $t$ for time. For recent work on unitary designs, see~\cite{Emerson03, Emerson05, Brown10, Brown15, Hayden07, Harrow09, Knill08, Brandao12, Kueng15, Webb15, Low_thesis, Scott08, Zhu15, Collins16,Nakata:2016blv}.} A unitary $k$-design is a subset of the unitary group that replicates the statistics of at least $k$ moments of the distribution. Consider a finite-dimensional Hilbert space $\mathcal{H}^{\otimes k}=(\mathbb{C}^{d})^{\otimes k}$ consisting of $k$ copies of $\mathcal{H}=\mathbb{C}^d$. Given an ensemble of unitary operators $\mathcal{E}=\{p_{j},U_{j}\}$ acting on $\mathcal{H}$ with probability distribution $p_{j}$, the ensemble $\mathcal{E}$ is a unitary $k$-design if and only if
    \begin{align}
        \sum_{j}p_{j} (\underbrace{U_{j}\otimes \cdots \otimes U_{j}}_{k}) \rho (\underbrace{U_{j}^{\dagger}\otimes \cdots \otimes U_{j}^{\dagger}}_{k}) = \int_{\text{Haar}}dU(\underbrace{U\otimes \cdots \otimes U}_{k}) \rho (\underbrace{U^{\dagger}\otimes \cdots \otimes U^{\dagger}}_{k}), 
    \end{align}
for all quantum states $\rho$ in $\mathcal{H}^{\otimes k}$. Therefore, whether or not a given ensemble $\mathcal{E}$ forms at most a $k$-design is a fine-grained notion of the randomness of the ensemble.\footnote{With this definition, a $k+1$-design is automatically a $k$-design.}

Inspection of Eq.~\eqref{oto-4pt-def} suggests that a $2$-design is sufficient for OTO four-point functions to decay. This was also the requirement for a system to appear scrambled \cite{Page93,Hayden07,Sekino08}. However, the four-point functions Eq.~\eqref{oto-4pt-def} are not the only observables related to chaos and the butterfly effect. In an attempt to understand the geometry and interior of a typical microstate of a holographic black hole, Shenker and Stanford constructed a series of black hole geometries by making multiple perturbations of an essentially maximally entangled state \cite{Shenker:2013yza}. While they were unfortunately unsuccessful at constructing a generic microstate, studying correlation functions in these geometries leads one to a family of $2k$-point out-of-time-order correlation functions of the form
    \be
        \T{\mc{W}^\dagger \, V(0)^\dagger \, \mc{W} \, V(0) } = \T{W_1(t_1)^\dagger \cdots W_{k-1}(t_{k-1})^\dagger ~ V(0)^\dagger ~ W_{k-1}(t_{k-1}) \cdots W_1(t_1) ~ V(0) }, \label{multiple-shocks-correlator}
    \ee
where $V(0)$ and the $W_j(0)$ are local operators, and $\mc{W} \equiv W_{k-1}(t_{k-1}) \cdots W_1(t_1)$ is a composite operator that lets us understand how the correlator is organized. The existence of these correlators begs the question: do these contain any additional information about the system or are they redundant with the four-point functions? The fact that the correlators Eq.~\eqref{multiple-shocks-correlator} involve many more copies of unitary time evolution, suggests a relationship between higher-point functions and unitary design. In fact, we will show that a generalization of the OTO higher-point correlators Eq.~\eqref{multiple-shocks-correlator} are probes of unitary design.

\subsubsection*{Complexity}

Probes of chaos and unitary design can only indirectly address a sharp question of recent interest in holography: the existence and growth of the black hole interior via the time evolution of the state. To be precise, one cannot make use of unitary designs since evolution by a time-independent Hamiltonian does not define any ensemble. However, even in this setting we still expect that higher-point OTO should be able to see fine-grained properties of time evolved states that are not captured by four-point OTO correlators.

One fine-grained quantity is the quantum circuit complexity of a quantum state $|\psi(t)\rangle = U(t)|\psi(0)\rangle$. Consider a simple initial state, such as the product state $|\psi(0)\rangle= |0\rangle^{\otimes n}$, undergoing chaotic time evolution. After a short time thermalization time of $O(1)$, the system will evolve to an equilibrium in which local quantities will reach their thermodynamic values. Next, after the scrambling time of $O(\log n)$, the initial state $|\psi(0)\rangle$ will be forgotten: the information will be distributed in such a way that measurements of even a large number of collections of local quantities of $|\psi(t)\rangle$ will not reveal the initial state. However, even after the scrambling time, the quantum circuit complexity of the time-evolved state $|\psi(t)\rangle$, as quantified by the number of elementary quantum gates necessary to reach it from the initial product state, will continue to evolve. In fact, it is expected to keep growing linearly in time until it saturates at a time exponential in the system size $e^{O(n)}$ \cite{Knill95,Susskind:2015toa}.

\sbreak

We hope this presents an intuitive picture that chaos, pseudorandomness, and quantum circuit complexity should be related.  
To that end, having first established a connection between higher-point out-of-time-order correlators and pseudorandomness, we will next connect the randomness of an ensemble to computational complexity. Finally, we will use correlation functions to probe pseudorandomness by comparing different random averages to expectations from time evolution.

Below, we will summarize our main results, deferring the technical statements to the body and appendices.

\subsection*{Main results}

We will focus on a particular form of $2k$-point correlation functions evaluated on the maximally mixed state $\rho=\frac{1}{d}I$
    \be
        \langle A_{1} \, U^{\dagger}B_1U \cdots A_{k} \, U^{\dagger}B_kU  \rangle := \frac{1}{2^n}\tr \, \{ A_{1}\, U^{\dagger}B_1U \cdots A_{k}\, U^{\dagger}B_kU \},  \label{OCO-correlator}
    \ee
where any of the $A_{j}, B_{j}$ may be a product Pauli operators that act on a single spin.\footnote{To be clear, this means that the operators we are correlating are not necessarily simple or local.} Note that each of the $B_j$ is conjugated by the same unitary $U$ (which is similar to picking all the time arguments in Eq.~\eqref{multiple-shocks-correlator} to be either $0$ or $t$). Furthermore, $U$ will not necessarily represent Hamiltonian time evolution, and instead we will let $U$ be sampled from some ensemble.\footnote{As a result, these correlators are not really out-of-\emph{time}-order, since there may not be a notion of time. Instead, they are probably more accurately called \emph{out-of-complexity-order} (OCO) since we might generalize the notion of time ordering to complexity ordering, where we put unitaries of smaller complexity to the right of unitaries of larger complexity. In order to (hopefully) avoid confusion, we will continue to call the $2k$-point functions in Eq.~\eqref{OCO-correlator} out-of-time-order, despite there not necessarily being a notion of time.} From this point forward, we will use the notation
\be
\tilde{B} \equiv U^{\dagger}BU,
\ee
to simplify expressions involving unitary conjugation. Therefore, we can represent the ensemble average of OTO $2k$-point correlation functions as
\begin{align}
\langle A_{1}\tilde{B_{1}}\cdots A_{k}\tilde{B_{k}} \rangle_{\mathcal{E}} := \int_{\mathcal{E}} dU \langle A_{1}\tilde{B_{1}}\cdots A_{k}\tilde{B_{k}} \rangle,
\end{align}
where the integral is with respect to the probability distribution in an ensemble of unitary operators $\mathcal{E}=\{p_{j},U_{j}\}$. Finally, the $k$-fold channel over the ensemble $\mathcal{E}$ is 
\begin{align}
\Phi_{\mathcal{E}}^{(k)}(\cdot) = \int_{\mathcal{E}} dU (U_{j}\otimes \cdots \otimes U_{j}) (\cdot )(U_{j}^{\dagger}\otimes \cdots \otimes U_{j}^{\dagger}),
\end{align}
which is a superoperator. 

\subsubsection*{Chaos and $k$-designs}

First,  we will prove a theorem stating that a particular set of $2k$-point OTO correlators, averaged over an ensemble $\mathcal{E}$, is in a one-to-one correspondence with the $k$-fold channel~$\Phi_{\mathcal{E}}^{(k)}$ 
\begin{align}
\text{$2k$-point OTO correlators}\ \leftrightarrow \ \text{$k$-fold channel $\Phi_{\mathcal{E}}^{(k)}$},
\end{align}
and we provide a simple formula to convert from one to the other. Such an explicit relation between OTO correlators and the $k$-fold channel may have practical and experimental applications such as statistical testing (e.g. a quantum analog of the $\chi^2$-test) and randomized bench marking~\cite{Knill08}.

Next, we prove that generic ``smallness'' of $2k$-point OTO correlators implies that the ensemble $\mathcal{E}$ is close to $k$-design. We will make this statement precise by relating OTO correlators to a useful quantity known as the frame potential 
\begin{align}
F_{\mathcal{E}}^{(k)}=\frac{1}{|\mathcal{E}|^{2}}\sum_{U,V\in\mathcal{E}}\big|\tr\, \{ U^{\dagger}V \} \big|^{2k}.
\end{align}
This quantity, first introduced in \cite{Scott08}, measures the ($2$-norm) distance between $\Phi_{\mathcal{E}}^{(k)}(\cdot)$ and $\Phi_{\text{Haar}}^{(k)}(\cdot)$, and has been shown to be minimized if and only if the ensemble $\mathcal{E}$ is $k$-design. We will derive the following formula:
\begin{align}
\text{Average of $|\text{$2k$-point OTO correlator}|^{2}$}\ \propto \ \text{$k$th frame potential $F_{\mathcal{E}}^{(k)}$},
\end{align}
which shows that $2k$-point OTO correlators $\langle A_{1}\tilde{B_{1}}\cdots A_{k}\tilde{B_{k}} \rangle_{\mathcal{E}}$ are measures of whether an ensemble $\mathcal{E}$ is a unitary $k$-design.
Thus, the decay of OTO correlators can be used to quantify an increase in pseudorandomness. 

\subsubsection*{Chaos, randomness, and complexity}

We prove a lower bound on quantum circuit complexity needed to generate an ensemble of unitary operators $\mathcal{E}$
\begin{align}
\text{Complexity of $\mathcal{E}$}\ \geq \ \frac{2kn \log(2) - \log F_{\mathcal{E}}^{(k)} }{\log( \text{choices})}.\label{intro-complexity-lower-bound}
\end{align}
This bound is actually given by a rather simple counting argument. The denominator should be thought of as (the log of) the number of choices made at each step in generating the circuit.  For instance, if we have a set $\mc{G}$ of cardinality $g$ of $q$-qubit quantum gates available and at each step randomly select $q$ qubits out of $n$ total and select one of the gates in $\mc{G}$ to apply, then we would make $g\binom{n}{q}$ choices at each step. Recalling our result relating OTO correlators and the frame potential, this result implies that generic smallness of OTO correlators leads to higher quantum circuit complexity. This is a direct and quantitative link between chaos and complexity. 

However, we caution the reader that in many cases Eq.~\eqref{intro-complexity-lower-bound} may not be a very tight lower bound. We will provide some discussion of this point as well as a few examples, however further work is most likely required to better understand the utility of this bound.

\subsubsection*{Haar vs. simpler ensemble averages}
Finally, we present calculations of the Haar average of some higher-point OTO correlators and compare them to averages in simpler ensembles. These results suggest that the floor value of OTO correlators of local operators might be good diagnostics of pseudorandomness.

For $4$-point OTO correlators, we find
\begin{equation}
\begin{split}
&\langle A\tilde{B}C\tilde{D} \rangle_{\text{Haar}} = \langle AC \rangle \langle B \rangle\langle D \rangle + \langle A \rangle\langle C \rangle \langle BD \rangle -  \langle A\rangle \langle C \rangle \langle B \rangle\langle D \rangle - \frac{1}{d^2-1}\langle\!\langle AC\rangle\!\rangle \langle\!\langle BD\rangle\!\rangle,
\end{split} \label{intro-4p-haar-average}
\end{equation}
where $\langle\!\langle AC\rangle\!\rangle$ represents a connected correlator and $d=2^{n}$ is the total number of states. In contrast, if the $U$ are averaged over the Pauli operators (which form a $1$-design but not a $2$-design), we find
\begin{align}
\langle A\tilde{B}C\tilde{D} \rangle_{\text{Pauli}}&= 
\langle AC \rangle \langle B D \rangle. \label{intro-4p-pauli-average}
\end{align}
We will present intuitive explanations on this difference from the viewpoint of local thermal dissipations vs. global thermalization or scrambling. 

For $8$-point OTO correlators with Pauli operators $A,B,C, D$, we compute averages over the unitary group and the Clifford group (which form a $3$-design on qubits but not a $4$-design)
\begin{equation}
\begin{split}
&\langle A\tilde{B}C\tilde{D}A^{\dagger}  \tilde{D}^{\dagger} C^{\dagger}\tilde{B}^{\dagger}\rangle_{\text{Haar}} \sim \frac{1}{d^4},\\
&\langle A\tilde{B}C\tilde{D}A^{\dagger}  \tilde{D}^{\dagger} C^{\dagger}\tilde{B}^{\dagger}\rangle_{\text{Clifford}} \sim \frac{1}{d^2}.
\end{split} \label{intro-8p-averages}
\end{equation}
This suggests that forming a higher $k$-design leads to a lower value of the correlator.

The results Eq.~\eqref{intro-4p-haar-average}-\eqref{intro-8p-averages} are exact for any choice of operators. (The extended results for the correlators in Eq.~\eqref{intro-8p-averages} is presented in Appendix~\ref{sec:8-pt-functions}.) However, for a particular ordering of OTO $4m$-point functions where we average over choices of operators, we will also show that a Haar averaging over the unitary group scales as
\be
\textrm{OTO}^{(4m)} \sim \frac{1}{d^{2m}}.\label{intro-4m-averages}
\ee
This result hints that these correlators continue to be probes of increasing pseudorandomnes.

\subsection*{Organization of the paper}
For convenience of the reader, we included a(n almost) self-contained introduction to Haar random unitary operators and unitary design in \S\ref{sec:review}. In \S\ref{sec:OTO_channel}, we establish a formal connection between chaos and unitary design by proving the theorems mentioned above. In \S\ref{sec:complexity-bound}, we connect complexity to unitary design by proving the complexity lower bound \eqref{intro-complexity-lower-bound}.
In \S\ref{sec:haar-averages}, we include the explicit calculations of $2$-point and $4$-point functions averaged over different ensembles and discuss how these averages relate to expectations from time evolution with chaotic Hamiltonians. We also discuss results and expectations for higher-point functions.
We conclude in \S\ref{sec:discussion} with an extended discussion of these results, their relevance for physical situations, and outline some future research directions.

Despite page counting to the contrary, this is actually a short paper. A knowledgeable reader may learn our results by simple reading \S\ref{sec:OTO_channel}, \S\ref{sec:complexity-bound} and \S\ref{sec:haar-averages}. On the other hand, a large number of extended calculations and digressions are relegated to the Appendices: 
\begin{itemize}
\item In \S\ref{app:proof}, we collect some proofs that we felt interrupted the flow of the main discussion. 
\item In \S\ref{sec:appendix:orthogonal}, we discuss the number of nearly orthogonal states in large-dimensional Hilbert spaces.
\item In \S\ref{sec:complexity-appendix}, we extend our complexity lower bound to minimum circuit depth by considering gates that may be applied in parallel. We also derive a bound on early-time complexity for evolution with an ensemble of Hamiltonians.
\item In \S\ref{sec:appendix-averages}, we hide the details of our $8$-point functions Haar averages and also derive the ${}\sim d^{-2m}$ scaling of Haar averages of certain $4m$-point functions.
\item In \S\ref{sec:sub-space-randomization}, we provide a generalization of the frame potential that equals an average of the square of OTO correlators for arbitrary states rather than just the maximally mixed state.
\item Finally, in \S\ref{sec:more-chaos} we prove some extended results relating to our earlier work \cite{Hosur:2015ylk} that are somewhat outside the main focus of the current paper.
\end{itemize}

\section{Measures of Haar}\label{sec:review}

The goal of this section is to provide a review of the theory of Haar random unitary operators and unitary design in a self-contained manner. The presentation of this section owes a lot to a recent paper by Webb~\cite{Webb15} as well as a course note by Kitaev \cite{Kitaev-Haar}.\footnote{We also would like to highlight a Master's thesis by Yinzheng Gu, which considers a certain matrix multiplication problem with conjugations by Haar random unitary operators. Although not directly relevant to the work, these are actually out-of-time-order correlation functions in a disguised form~\cite{gu2013moments}.}

\subsection*{Haar random unitaries}

\subsubsection*{Schur-Weyl duality}
Consider a finite-dimensional Hilbert space $\mathcal{H}^{\otimes k}=(\mathbb{C}^{d})^{\otimes k}$ consisting of $k$ copies of $\mathcal{H}=\mathbb{C}^d$. A permutation operator $W_{\pi}$ with a permutation $\pi=\pi(1)\ldots \pi (k)$ is defined as follows
\begin{align}
W_{\pi} |a_{1},\ldots,a_{k}  \rangle = |a_{\pi(1)},\ldots,a_{\pi(k)}  \rangle,
\end{align}
and $W_{\pi}(A_{1}\otimes \cdots \otimes A_{k})W_{\pi}^{-1} = A_{\pi(1)}\otimes \cdots \otimes A_{\pi(k)}$. If $\pi$ is a cyclic permutation ($\pi(j)=j+1$), then $W_{\text{cyc}}$ acts as follows
\begin{align}
\includegraphics[width=0.35\linewidth]{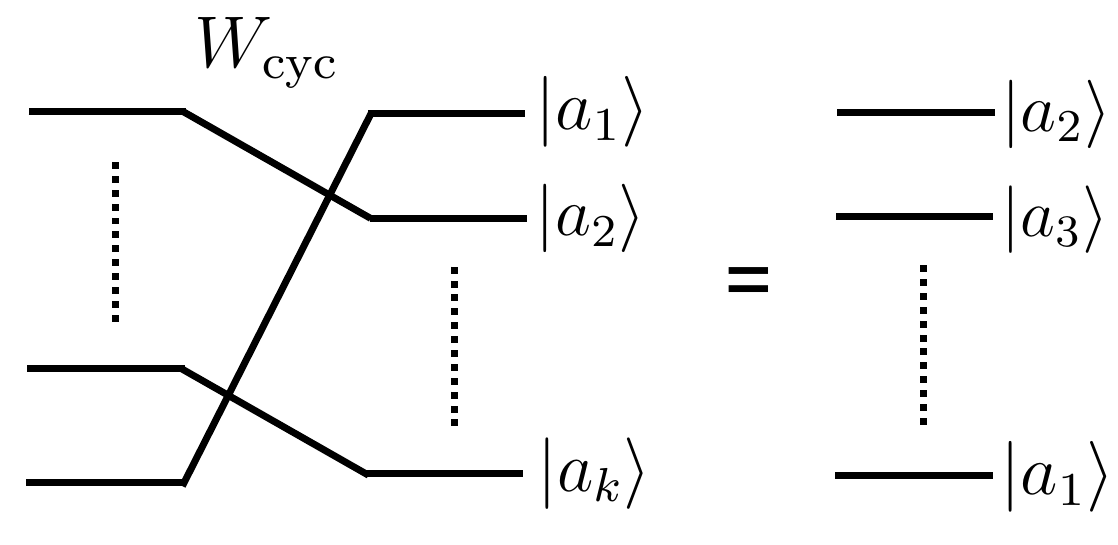}.
\end{align}

\begin{theorem}\emph{\tb{[Schur-Weyl duality]}}
Let $L(\mathcal{H}^{\otimes k})$ be the algebra of all the operators acting on $\mathcal{H}^{\otimes k}$. Let $U(\mathcal{H})$ be the unitary group on $\mathcal{H}$. An operator $A \in L(\mathcal{H}^{\otimes k})$ commutes with all operators $V^{\otimes k}$ with $V\in U(\mathcal{H})$ if and only if $A$ is a linear combination of permutation operators $W_{\pi}$
\begin{align}
[A,V^{\otimes k}]=0, ~ \forall V \ \Leftrightarrow \ A = \sum_{\pi\in S_k} c_{\pi}\cdot W_{\pi}.\label{Schur-Weyl-duality} 
\end{align}
\end{theorem}

When an operator $A$ is a linear combination of permutation operators, it is clear that $A$ commutes with $V^{\otimes k}$ ($\Leftarrow$). A difficult part is to prove the converse ($\Rightarrow$), which relies on Von Neumann's double commutant theorem. 

\subsubsection*{Pauli operators}
Pauli operators for $\mathbb{C}^d$ (i.e. $d$-state spins or qudits) are defined by 
\begin{align}
X|j\rangle = |j+1\rangle, \qquad Z|j\rangle = \omega^j|j\rangle, \label{eq:pauli-def}
\end{align}
where $\omega \equiv e^{2\pi i/d}$. We note that Eq.~\eqref{eq:pauli-def} implies $ZX = \omega XZ$ and $X^d = Z^d = I$, and that for $d>2$, the Pauli operators are unitary and traceless, but not Hermitian.

The Pauli group is $\tilde{\mathcal{P}}=\langle \tilde{\omega}I,X,Z \rangle$, where $\tilde{\omega}=\omega$ for odd $d$, and $\tilde{\omega}=e^{\pi/d}$ for even $d$. Since we are usually uninterested in global phases, we will consider the quotient of the group
\begin{align}
\mathcal{P} = \tilde{\mathcal{P}} \backslash \langle \tilde{\omega} I \rangle.
\end{align}
There are $d^2$ (representative) Pauli operators in $\mathcal{P}$. When the Hilbert space is built up from the space of $n$ qubits, we will denote the Pauli group by $\mathcal{P}_{n}$. For such systems, (the representatives of) $\mathcal{P}_{n}$ consist of tensor products of qubit Pauli operators, such as $X \otimes Y \otimes I \otimes Z \otimes \cdots $ without any global phases.

The Pauli operators provide a basis for the space of linear operators acting on the Hilbert space. They are orthogonal, $\tr \, \{P_{i}^{\dagger}P_{j} \}=d\delta_{ij}$ for $P_{i},P_{j}\in \mathcal{P}$, and therefore we can expand any operator $A$ acting on $\mathcal{H}$ as
\begin{align}
A = \sum_{j} a_{j} P_{j},\qquad a_{j}=\frac{1}{d}\tr \, \{P_{j}^{\dagger}A\}.
\end{align}
With this property, the cyclic permutation operator $W_{\text{cyc}}$ on $\mathcal{H}^{\otimes k}$ can be decomposed as
\begin{align}
W_{\text{cyc}} = \frac{1}{d^{k-1}}\sum_{P_{1},\ldots,P_{k-1}\in \mathcal{P}} P_{1}\otimes P_{2} \otimes \cdots P_{k-1}\otimes Q^{\dagger}, 
\qquad Q = P_{1}P_{2}\cdots P_{k-1}, \label{eq:pauli-decomposition-of-cyc}
\end{align}
where the sum is over $k-1$ copies of $\mathcal{P}$.

The case of $k=2$ is particularly important, giving an operator that swaps two subsystems. Explicitly, we have $\text{SWAP}=\frac{1}{d}\sum_{P} P\otimes P^{\dagger}$, or graphically (up to a multiplicative factor)
\begin{align}
\includegraphics[width=0.35\linewidth]{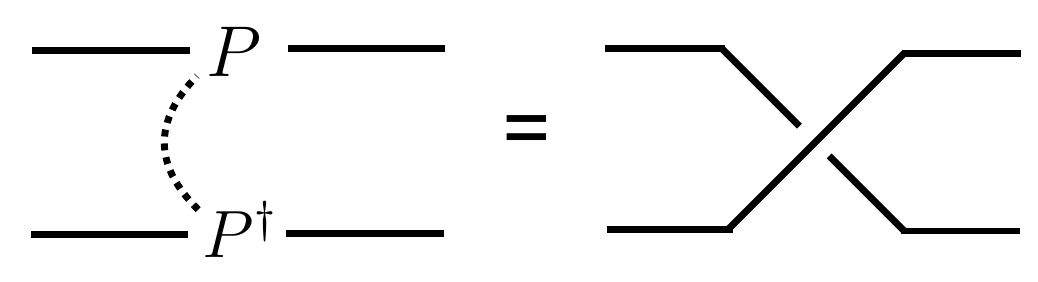}.\label{eq:fig_SWAP}
\end{align}
Here and in what follows, a dotted line represents an average over all Pauli operators. For example, let us consider the Pauli channel
\begin{align}
\frac{1}{d^2}\sum_{P \in \mathcal{P}}P^{\dagger} A P = \frac{1}{d}\, \tr  \{A \},\label{eq:Pauli_twirl}
\end{align}
where $A$ is any operator on the system.
This equation can be derived graphically by applying Eq.~(\ref{eq:fig_SWAP})
\begin{align}
\includegraphics[width=0.50\linewidth]{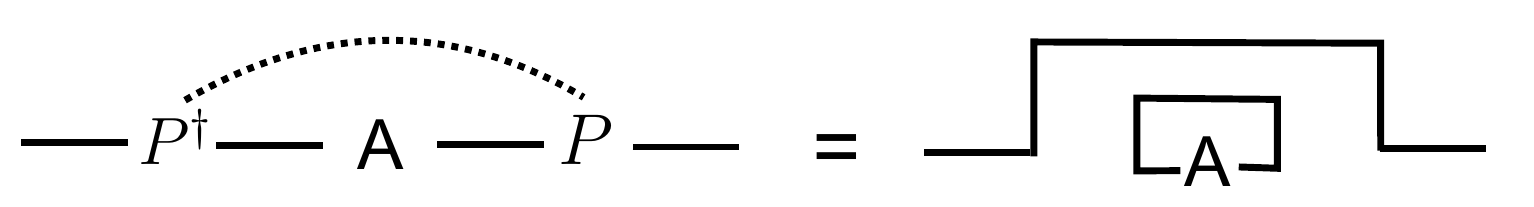}.
\end{align}

\subsubsection*{$k$-fold channel}
Let $A$ be an operator acting on $\mathcal{H}^{\otimes k}$. The $k$-fold channel of $A$ with respect to the unitary group is defined as
\begin{align}
\Phi_{\text{Haar}}^{(k)}(A) := \int_{\text{Haar}} (U^{\otimes k})^{\dagger} A \, U^{\otimes k}  dU,
\end{align}
where the integral is taken over the Haar measure. Note, this is sometimes referred to as the $k$-fold twirl of $A$. The Haar measure is the unique probability measure on the unitary group that is both left-invariant and right-invariant~\cite{Watrous}
\begin{align}
\int_{\text{Haar}} dU = 1,\qquad \int_{\text{Haar}} f(VU)dU=\int_{\text{Haar}} f(UV)dU = \int_{\text{Haar}} f(U)dU, \label{eq:def_Haar}
\end{align} 
for all $V\in U(\mathcal{H})$, where $f$ is an arbitrary function. If we take $f(U)= (U^{\otimes k})^{\dagger}A\,U^{\otimes k}$, then we can show that the twirl of $A$ is invariant under $k$-fold unitary conjugation, 
\begin{align}
(V^{\otimes k})^{\dagger}\big(\Phi^{(k)}_{\text{Haar}}(A)\big) V^{\otimes k}  &= \int_{\text{Haar}} f(UV)dU = \Phi^{(k)}_{\text{Haar}}(A),\label{eq:duality1}
\end{align}
and that twirl of the $k$-fold unitary conjugation of $A$ equals the twirl of $A$
\begin{align}
\Phi^{(k)}_{\text{Haar}}\big((V^{\otimes k})^{\dagger}A V^{\otimes k}\big)&=\int_{\text{Haar}} f(VU)dU = \Phi^{(k)}_{\text{Haar}}(A),\label{eq:duality2}
\end{align}
where for Eq.~\eqref{eq:duality1} we used the right-invariance property, and for Eq.~\eqref{eq:duality2} we used left-invariance.

\subsubsection*{Weingarten function}
The content of Eq.~(\ref{eq:duality1}) is that $\Phi^{(k)}_{\text{Haar}}(A)$ commutes with all operators $V^{\otimes k}$. Thus, we may use the Schur-Weyl duality Eq.~\eqref{Schur-Weyl-duality} to rewrite it as
\begin{align}
\Phi^{(k)}_{\text{Haar}}(A) = \sum_{\pi \in S_k} W_{\pi} \cdot u_{\pi}(A).
\end{align}
Here, $S_k$ is the permutation group, and $u_{\pi}(A)$ is some linear function of $A$. Since $u_{\pi}(A)$ is a linear function, it can be written as
\begin{align}
u_{\pi}(A) = \tr \{ C_{\pi} A\},
\end{align}
for some operators $C_{\pi}$. From Eq.~(\ref{eq:duality2}), we find that $C_{\pi}$ commutes with all operators $V^{\otimes k}$. Then again, by the Schur-Weyl duality, we have
\begin{align}
\Phi^{(k)}_{\text{Haar}}(A) =\sum_{\pi,\sigma\in S_k} c_{\pi,\sigma} W_{\pi}\, \tr \{ W_{\sigma}A \}.
\end{align}
The coefficients $c_{\pi,\sigma}$ are called the Weingarten matrix~\cite{Collins03}. Since $\Phi^{(k)}_{\text{Haar}}(W_{\lambda})=W_{\lambda}$, we have $W_{\lambda} = \sum_{\pi,\sigma} c_{\pi,\sigma}W_{\pi} \, \tr \{W_{\sigma}W_{\lambda} \}$. Recalling that $\tr \{ W_{\sigma}W_{\lambda}\}=d^{\# \text{cycles}(\sigma\lambda)}$, we have
\begin{align}
\delta_{\pi,\lambda} = \sum_{\sigma\in S_k}c_{\pi,\sigma} Q_{\sigma,\lambda},\qquad Q_{\sigma,\lambda}:=d^{\# \text{cycles}(\sigma\lambda)}.
\end{align}
So finally, we find
\begin{align}
\Phi^{(k)}_{\text{Haar}}(A) = \sum_{\pi,\sigma\in S_k} (Q^{-1})_{\pi,\sigma}  W_{\pi} \cdot \tr \{W_{\sigma}A\}. \label{eq:Weingarten_decomposition}
\end{align}
Here, we assumed the presence of the inverse $Q^{-1}$ which is guaranteed for $k \leq d$.

\subsection*{Examples}
For $k=1$, $Q_{I,I}=d$, so one has
\begin{align}
\Phi^{(1)}_{\text{Haar}}(A) = \frac{1}{d}  \tr\{A\}.
\end{align}
For $k=2$, so one has
\begin{align}
Q = \left(
\begin{array}{cc}
Q_{I,I}, Q_{I,S} \\
Q_{S,I} , Q_{S,S}
\end{array}
\right) =
\left(
\begin{array}{cc}
d^2,&d \\
d , &d^2
\end{array}
\right) \quad 
C = \left(
\begin{array}{cc}
C_{I,I}, C_{I,S} \\
C_{S,I} , C_{S,S}
\end{array}
\right) =
\left(
\begin{array}{cc}
\frac{1}{d^2-1} ,&  \frac{-1}{d(d^2-1)} \\
\frac{-1}{d(d^2-1)} , &\frac{1}{d^2-1}
\end{array}
\right).
\end{align}
Explicitly, we have
\begin{align}
\Phi^{(2)}_{\text{Haar}}(A) = \frac{1}{d^2-1} \left( I \ \tr \{A\} + S\ \tr \{ SA\} - \frac{1}{d} S\ \tr \{A\} - \frac{1}{d} I\ \tr \{SA\} \right),
\end{align}
where $S$ is the SWAP operator.

\subsubsection*{Haar random states}
We can also consider the $k$-fold average of a Haar random state. Define a random state by $|\psi\rangle = U|0\rangle$, where $U$ is sampled uniformly from the unitary group. Then, we have
\begin{align}
\int_{\text{Haar}} (|\psi\rangle \langle \psi| )^{\otimes k} d\psi := \Phi^{(k)}_{\text{Haar}}\big((|0\rangle\langle 0 |)^{\otimes k}\big) = \sum_{\pi}c_{\pi}W_{\pi},\label{eq:Haar_state}
\end{align}
for some coefficients $c_{\pi}$.
Since $\int_{\text{Haar}} (|\psi\rangle \langle \psi| )^{\otimes k}$ commutes with $V^{\otimes k}$, we again decomposed it with permutation operators. Furthermore, one has
\begin{align}
W_{\pi}\int_{\text{Haar}} (|\psi\rangle \langle \psi| )^{\otimes k}=\int_{\text{Haar}} (|\psi\rangle \langle \psi| )^{\otimes k}, \qquad \forall \pi \in S_{k}
\end{align}
which implies $c_{\pi}=c$ for all $\pi$. By taking the trace in Eq.~(\ref{eq:Haar_state}), we have
\begin{align}
c_{\pi}= \frac{1}{\sum_{\pi\in S_k} d^{\text{\#cycles($\pi$)}}} = \frac{k!}{
\Big(\begin{array}{c}
k+d-1 \\
k
\end{array}\Big)
}.
\end{align}
Defining the projector onto the symmetric subspace as $\Pi_{\text{sym}} = \frac{1}{k!}\sum_{\pi\in S_k}W_{\pi}$, one has
\begin{align}
\int_{\text{Haar}} (|\psi\rangle \langle \psi| )^{\otimes k} d\psi = \frac{\Pi_{\text{sym}}}{
\Big(\begin{array}{c}
k+d-1 \\
k
\end{array}\Big) \label{haar-random-states-equation}
}.
\end{align}

\subsubsection*{Frame potential} In this paper, we will be particularly interested in the following quantity
\begin{align}
F_{\text{Haar}}^{(k)}=\iint dU dV\big|\tr\, \{ U^{\dagger}V \} \big|^{2k}.
\end{align}
Using Eq.~\eqref{eq:Weingarten_decomposition}, we can rewrite this
\begin{equation}
\begin{split}
F_{\text{Haar}}^{(k)} &= \sum_{\pi_{1},\pi_{2},\pi_{3},\pi_{4}\in S_{k}} \tr \{ W_{\pi_{1}}W_{\pi_{2}} \} (Q^{-1})_{\pi_{2},\pi_{3}} \tr \{ W_{\pi_{3}}W_{\pi_{4}} \} (Q^{-1})_{\pi_{4},\pi_{1}},  \\
&= \sum_{\pi_{1},\pi_{2},\pi_{3},\pi_{4}\in S_{k}}Q_{\pi_{1},\pi_{2}} (Q^{-1})_{\pi_{2},\pi_{3}}Q_{\pi_{3},\pi_{4}}(Q^{-1})_{\pi_{4},\pi_{1}} = k!,
\end{split}
\end{equation}
for $k\leq d$, where we used $C=Q^{-1}$ and $|S_{k}|=k!$.

\subsection*{Unitary design}

Consider an ensemble of unitary operators $\mathcal{E}=\{p_{j},U_{j}\}$ where $p_{j}$ are probability distributions such that $\sum_{j}p_{j}=1$, and $U_{j}$ are unitary operators. The action of the $k$-fold channel with respect to the ensemble $\mathcal{E}$ is given by 
\begin{align}
\Phi^{(k)}_{\mathcal{E}}(A) := \sum_{j} p_{j}   (U_{j}^{\otimes k})^{\dagger}A (U_{j})^{\otimes k },
\end{align}
or for continuous distributions
\begin{align}
\Phi^{(k)}_{\mathcal{E}}(A) := \int_{\mathcal{E}} dU   (U_{j}^{\otimes k})^{\dagger}A (U_{j})^{\otimes k }.
\end{align}
The ensemble $\mathcal{E}$ is a unitary $k$-design if and only if $ \Phi^{(k)}_{\mathcal{E}}(A)=\Phi^{(k)}_{\text{Haar}}(A)$ for all $A$. Intuitively, a unitary $k$-design is as random as the Haar ensemble up to the $k$th moment. That is, a unitary $k$-design is an ensemble which satisfies the definition of the Haar measure in Eq.~(\ref{eq:def_Haar}) when $f(U)$ contains up to $k$th powers of $U$ and $U^{\dagger}$ (i.e. balanced monomials of degree at most $k$). By this definition, if an ensemble is $k$-design, then it is also $k-1$-design. However, the converse is not true in general.

It is convenient to write the above definition of $k$-design in terms of Pauli operators. An ensemble $\mathcal{E}$ is $k$-design if and only if $ \Phi^{(k)}_{\mathcal{E}}(P)=\Phi^{(k)}_{\text{Haar}}(P)$ for all Pauli operators $P\in (\mathcal{P})^{\otimes k}$, since the Pauli operators are the basis of the operator space. Furthermore, for an arbitrary ensemble $\mathcal{E}$
\begin{align}
\Phi_{\text{Haar}}\left( \Phi_{\mathcal{E}}\big( A \big) \right) =\Phi_{\text{Haar}}\left( A\right), \qquad \forall \mathcal{E},\label{eq:ensemble-followed-by-Haar}
\end{align}
due to the left/right-invariance of the Haar measure. By using Eq.~\eqref{eq:ensemble-followed-by-Haar}, we can derive the following useful criteria for $k$-designs~\cite{Webb15}
\begin{align}
\text{$\mathcal{E}$ is $k$-design} \quad \Leftrightarrow \quad \text{$\Phi_{\mathcal{E}}(P)$ is a linear combination of $W_{\pi}$ for all $P\in (\mathcal{P})^{\otimes k}$.} \label{eq:criteria}
\end{align}
To make use of this, we will look at some illustrative examples. 

\subsubsection*{Pauli is a $1$-design}
The Pauli operators form a unitary $1$-design. In fact, we have already shown this in with Eq.~\eqref{eq:Pauli_twirl}. We have shown that an average over Pauli operators gives $\frac{1}{d^2}\sum_{P\in \mathcal{P}}P^{\dagger} A P = \frac{1}{d}\, \tr\{A\}$, and the Haar random channel gives $\Phi^{(1)}_{\text{Haar}}(A)=\frac{1}{d}\, \tr\{A\}$, so therefore
\begin{align}
\frac{1}{d^2}\sum_{P\in \mathcal{P}}P^{\dagger} \rho P=\Phi^{(1)}_{\text{Haar}}(\rho),
\end{align}
for all $\rho$. For example, if $d=2$ and $A=X$ (Pauli $X$ operator), then we have
\begin{align}
\frac{1}{4}(IXI + XXX + YXY + ZXZ) = X + X - X - X =0,
\end{align}
which is consistent with $\tr \{X \}=0$. 

Thus, an average over Pauli operators is equivalent to talking a trace. Since Pauli operators can be written in a tensor product form: $P_{1}\otimes P_{2}\otimes \ldots \otimes P_{n}$ for a system of $n$ qubits, they do not create entanglement between different qubits and do not scramble. Instead, they can only mix quantum information locally. This implies some kind of relationship between $1$-designs and local thermalization that we will revisit in the discussion \S\ref{sec:discussion}.

\subsubsection*{Clifford is a $2$-design}
The Clifford operators form a unitary $2$-design. The Clifford group $\mathcal{C}_{n}$ is a group of unitary operators acting on a system of $n$ qubits that transform a Pauli operator into another Pauli operator
\begin{align}
C^{\dagger} P C = Q, \qquad P,Q\in \tilde{\mathcal{P}}, \quad C\in \mathcal{C}_{n}.
\end{align}
Clearly, Pauli operators are Clifford operators, since $PQP = e^{i\theta} Q$ for any pairs of Pauli operators $P,Q$: Pauli operators transform a Pauli operator to itself up to a global phase. However, non-trivial Clifford operators are those which transform a Pauli operator into a different Pauli operator. An example of such an operator is the Control-$Z$ gate
\begin{align}
\text{C$Z$}|a,b\rangle = (-1)^{ab}|a,b\rangle, \qquad a,b =0,1,
\end{align}
where summation is modulo $2$. The conjugation of a Pauli operator with a Control-$Z$ gate is as follows
\begin{align}
X\otimes I \rightarrow X \otimes Z, \qquad Z\otimes I \rightarrow Z \otimes I, 
\qquad
I\otimes X \rightarrow Z \otimes X, \qquad I\otimes Z \rightarrow I \otimes Z. \label{eq:control-z}
\end{align}

Let us prove that the Clifford group is $2$-design by using Eq.~(\ref{eq:criteria})~\cite{Webb15}. For qubit Pauli operators of the form $P\otimes P$ ($P\not=I$), the action of the Clifford $2$-fold channel is
\begin{align}
\sum_{C\in \mathcal{C}_{n} } (C^{\dagger}\otimes C^{\dagger}) P\otimes P (C\otimes C) \propto \sum_{Q\in \mathcal{P}_{n}, Q\not=I }Q\otimes Q,
\end{align}
because a random Clifford operator will transform $P$ into some other non-identity Pauli operator. Recalling the definition of the swap operator $\text{SWAP}=\frac{1}{d}\sum_{P\in \mathcal{P}_{n}} P\otimes P$, the RHS is a linear combination of $I\otimes I$ and $\text{SWAP}$. On the other hand, for other Pauli operators $P\otimes Q$ with $P\not=Q$, the action of the channel is
\begin{align}
\sum_{C\in \mathcal{C}_{n} } (C^{\dagger}\otimes C^{\dagger}) P\otimes Q (C\otimes C) =0.
\end{align}
This can be seen by rewriting the sum as 
\be
\frac{1}{2}\sum_{C\in \mathcal{C}_{n} } (C^{\dagger}\otimes C^{\dagger}) P\otimes Q (C\otimes C) + \frac{1}{2}\sum_{RC\in \mathcal{C}_{n} } ( C^\dagger R^{\dagger}\otimes C^\dagger R^{\dagger}) P\otimes Q (RC\otimes RC), 
\ee
since the Clifford group is invariant under element-wise multiplication by a Pauli $R$. Since we assumed the Pauli operators are different we have $PQ \neq I$, and thus we can pick $R$ such that it anti-commutes with $PQ$. This implies $R^\dagger P R \otimes R^\dagger Q R = - P \otimes Q$, and therefore the two terms cancel.
Finally, $I\otimes I$ remains invariant. Thus, since the action of the channel on Pauli operators gives a linear combination of permutation operators, by Eq.~\eqref{eq:criteria} the Clifford group is a unitary $2$-design. 

Note that Clifford operators do not have a tensor product form, in general. This means that unlike evolution restricted to Pauli operators, they can change the sizes of an operator as seen in the Control-$Z$ gate example Eq.~\eqref{eq:control-z}. In other words they can grow local operators into a global operators, indicative of the butterfly effect. 

In fact, Clifford operators can prepare a large class of interesting quantum states called the stabilizer states. Let $|\psi_{0}\rangle=|0\rangle^{\otimes n}$ be an initial product state. This state satisfies $Z_{j}|\psi_{0}\rangle = |\psi_{0}\rangle$. Let $U$ be an arbitrary Clifford operator and consider $|\psi\rangle = U |\psi_{0}\rangle$. This state $|\psi\rangle$ satisfies the following
\begin{align}
S_{j}|\psi\rangle = |\psi\rangle, \qquad S_{j}=UZ_{j}U^{\dagger}.
\end{align}
By definition, $S_{j}$ are Pauli operators and will commute with each other. A quantum state that can be represented by a set of commuting Pauli operators $S_{j}$ is called a stabilizer state. Examples of stabilizer states include ground states of the toric code  and the perfect tensors used in construction of holographic quantum error-correcting codes \cite{Pastawski15b}. The upshot is that Clifford operators can create a global entanglement and can scramble quantum information. We will return to this point again in the discussion \S\ref{sec:discussion}.

\subsubsection*{? is a higher-design}
Currently there is no known method of constructing an ensemble which forms an exact $k$-design for $k\geq 4$ in a way which generalizes to large $d$. Instead, there are several constructions for preparing approximate $k$-design in an efficient manner~\cite{Brandao12,Nakata:2016blv}.

\section{Measures of chaos and design}\label{sec:OTO_channel}

In this section, we show that $2k$-point OTO correlators are probes of $k$-unitary designs. We will focus on a Hilbert space $\mathcal{H}=\mathbb{C}^d$ with $2k$-point correlators of the following form
\begin{align}
\big\langle A_{1}\tilde{B_{1}}\cdots A_{k}\tilde{B_{k}} \big\rangle := \frac{1}{d}\tr\, \{A_{1}\tilde{B_{1}}\cdots A_{k}\tilde{B_{k}}\},\label{eq-form-of-correlator}
\end{align}
where $\tilde{B_{j}}=U^{\dagger}B_{j}U$. We can think of this as a correlator evaluated in a maximally mixed or  infinite temperature state $\rho = \frac{1}{d}I$. The trace can be rewritten as
\begin{align}
\tr \,\big\{ (A_{1} \otimes \cdots \otimes A_{k}) (U^{\dagger}\otimes \cdots\otimes U^{\dagger})
(B_{1} \otimes \cdots \otimes B_{k})(U\otimes \cdots\otimes U)
 W_{\pi_{\text{cyc}}} \big\},\label{eq:cyclic}
\end{align}
by considering an enlarged Hilbert space $\mathcal{H}^{\otimes k}$ that consists of $k$ copies of the original Hilbert space $\mathcal{H}$, where $W_{\pi_{\text{cyc}}}$ represents a cyclic permutation operator on $\mathcal{H}^{\otimes k}$. The action of $W_{\text{cyc}}$ is to send the $j$th Hilbert space to the $(j+1)$th Hilbert space (modulo $k$), see Fig.~\ref{fig_k-fold_twirl} for a graphical representation.\footnote{N.B. this trick of using cyclic permutation operators is similar to the method used for computing R\'enyi-$k$ entanglement entropies in quantum field theories via the insertion of twist operators.} 

\begin{figure}[htb!]
\centering
\includegraphics[width=0.40\linewidth]{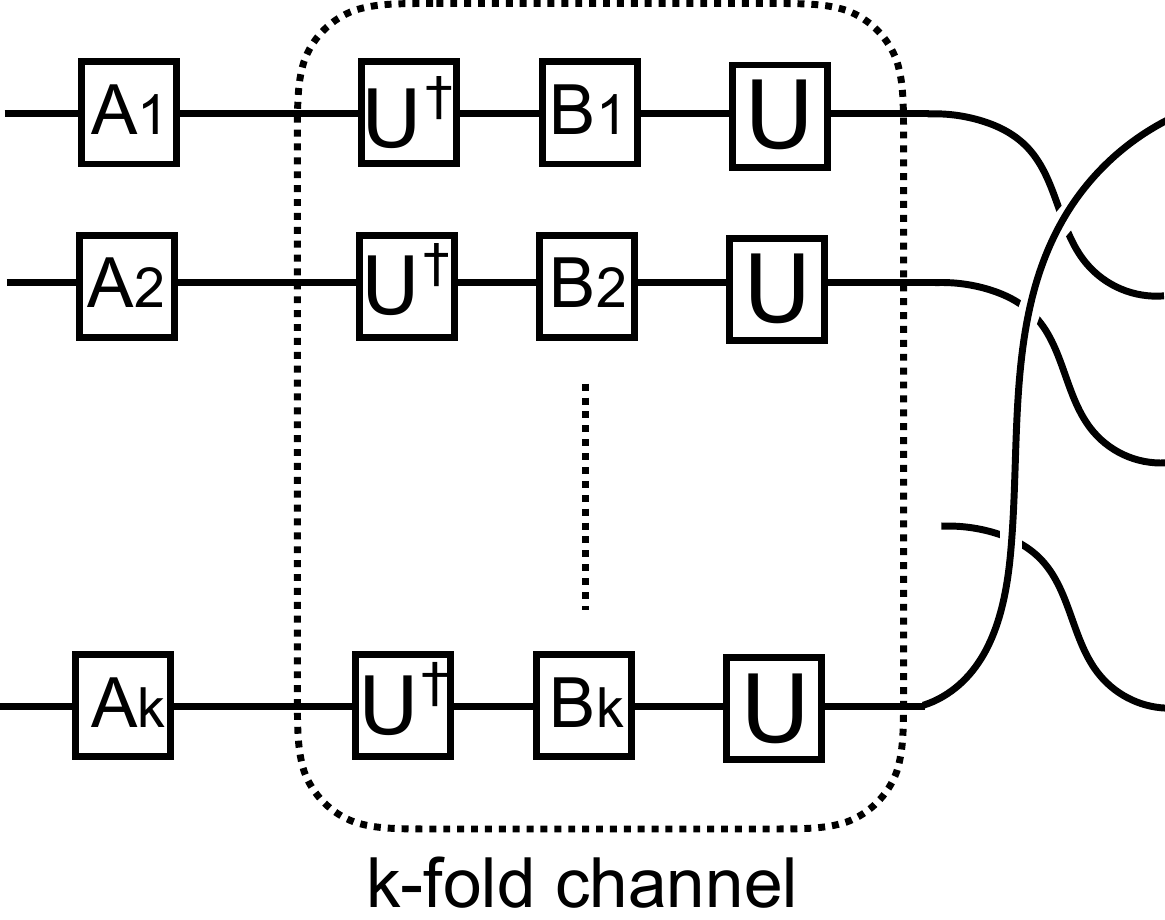}
\caption{Schematic form of the $2k$-point OTO correlation functions Eq.~\eqref{eq-form-of-correlator}, interpreted as a correlation function on the enlarged $k$-copied system. The dotted line diagram surrounds the $k$-fold channel $\Phi_{\mathcal{E}}(B_{1} \otimes \cdots \otimes B_{k})$, which is probed by $A_{1}\otimes \ldots \otimes A_{k}$. (Periodic boundary conditions are implied to take the trace.)
} 
\label{fig_k-fold_twirl}
\end{figure}

Observe that Eq.~(\ref{eq:cyclic}) contains a $k$-fold unitary action 
\be
(U\otimes \cdots\otimes U){}
(B_{1} \otimes \cdots \otimes B_{k})
(U^{\dagger}\otimes \cdots\otimes U^{\dagger}), \notag
\ee 
suggesting correlators of the form Eq.~\eqref{eq-form-of-correlator} have the potential to be sensitive to whether an ensemble is or is not a $k$-design.
Unitary $k$-designs concern the \emph{$k$-fold channel} of an ensemble of unitary operators
\begin{align}
\Phi_{\mathcal{E}}(B_{1} \otimes \cdots \otimes B_{k}) = \int_{\mathcal{E}} dU (U\otimes \cdots\otimes U)
(B_{1} \otimes \cdots \otimes B_{k})
(U^{\dagger}\otimes \cdots\otimes U^{\dagger}),
\end{align} 
where $\mathcal{E}$ is an ensemble of unitary operators.  To further this point, let us consider an average of these correlators over an ensemble of unitary operators $\mathcal{E}$
\begin{align}
\left|\big\langle A_{1}\tilde{B_{1}}\cdots A_{k}\tilde{B_{k}} \big\rangle\right|_{\mathcal{E}} := \frac{1}{d}\int_{\mathcal{E}} dU\, \tr\, \{ A_{1}\tilde{B_{1}}\cdots A_{k}\tilde{B_{k}}\}. \label{eq-form-of-average-oto}
\end{align}
Looking back at Eq.~\eqref{eq:cyclic}, the idea is that $A_{1},\ldots,A_{k}$ operators probe the outcome of $k$-fold channel $\Phi_{\mathcal{E}}(B_{1} \otimes \cdots \otimes B_{k})$. Indeed, the part of Fig.~\ref{fig_k-fold_twirl} surrounded by a dotted line, is $\Phi_{\mathcal{E}}(B_{1} \otimes \cdots \otimes B_{k})$. Below, we will make this intuition precise by proving that this set of OTO correlators Eq.~\eqref{eq-form-of-correlator} completely determine the $k$-fold channel $\Phi^{(k)}_{\mathcal{E}}$.

\subsection{Chaos and $k$-designs}

In this subsection, we prove that $2k$-point OTO correlators completely determines the $k$-fold channel of an ensemble $\mathcal{E}$, denoted by $\Phi_{\mathcal{E}}^{(k)}$. 
To recap, we write the $k$-fold channel of an ensemble $\mathcal{E}=\{p_{j},U_{j}\}$ by 
\begin{align}
\Phi^{(k)}_{\mathcal{E}}(\rho) = \sum_{j} p_{j}  (U_{j}^{\dagger}\otimes \cdots \otimes U_{j}^{\dagger})\rho (U_{j}\otimes \cdots \otimes U_{j}),
\end{align}
where $\rho$ is defined over $k$ copies of the system, $\mathcal{H}^{\otimes k}$. The map is linear, and completely-positive and trace-preserving (CPTP), i.e. a quantum channel. For simplicity of discussion, we assume that the system $\mathcal{H}$ is made of $n$ qubits so that $\mathcal{H}=\mathbb{C}^{d}$ with $d=2^n$. The input density matrix $\rho$ can be expanded by Pauli operators 
\begin{align}
\rho_{\text{in}} = \sum_{B_{1},\ldots,B_{k}}\beta_{B_{1},\ldots,B_{k}}(B_{1}\otimes \cdots \otimes B_{k}),
\end{align}
where $\beta_{B_{1},\ldots,B_{k}}= d^{-k} \, \tr \, \big\{(B_{1}^{\dagger}\otimes \cdots \otimes B_{k}^{\dagger}) \rho\big\}$. The output density matrix is given by
\begin{align}
\rho_{\text{out}} = \sum_{B_{1},\ldots,B_{k}}\beta_{B_{1},\ldots,B_{k}}\Phi_{\mathcal{E}}(B_{1}\otimes \cdots \otimes B_{k}).
\end{align}
For a given Pauli operator $B_{1},\ldots,B_{k}$, we would like to examine $\Phi_{\mathcal{E}}(B_{1}\otimes \cdots \otimes B_{k})$. 

Let us fix Pauli operator $B_{1},\ldots,B_{k}$ for the rest of the argument. Note that the output $\Phi_{\mathcal{E}}(B_{1}\otimes \cdots \otimes B_{n})$ can be also expanded by Pauli operators
\begin{align}
\Phi_{\mathcal{E}}(B_{1}\otimes \cdots \otimes B_{k}) = \sum_{C_{1},\ldots,C_{k}}\gamma_{C_{1},\ldots,C_{k}}(C_{1}\otimes \cdots \otimes C_{k}).\label{eq:expansion}
\end{align}
Since we have fixed $B_{1}\otimes \cdots \otimes B_{k}$, for notational simplicity we have not included $B_{1}, \cdots, B_{k}$ indices from the tensor $\gamma$. In order to characterize the $k$-fold channel, we need to know values of $\gamma_{C_{1},\ldots,C_{n}}$ for a given $B_{1}, \cdots, B_{k}$. We would like to show that we can determine the values of $\gamma_{C_{1},\ldots,C_{n}}$ by knowing a certain collection of OTO correlators. Consider a $2k$-point OTO correlator labeled by the set of $A$ operators, averaged over an ensemble $\mathcal{E}$
\begin{align}
\alpha_{A_{1},\ldots,A_{k}} = \left|\big\langle A_{1} \tilde{B_{1}} \cdots A_{k}\tilde{B_{k}} \big\rangle \right|_{\mathcal{E}}\label{eq:definition},
\end{align}
where as always $\tilde{B_{j}}=U^{\dagger}B_{j}U$ and $A_{1},\ldots,A_{k}$ are Pauli operators. As before, for simplicity of notation we have not included $B_{1},\ldots,B_{k}$ indices on $\alpha$. Now that the notation is setup, the main question is whether one can determine the coefficients $\gamma_{C_{1},\ldots,C_{k}}$ from the numbers $\alpha_{A_{1},\ldots,A_{k}}$.

 Substituting Eq.~\eqref{eq:expansion} into Eq.~\eqref{eq:definition}, we see
\begin{align}
\alpha_{A_{1},\ldots,A_{k}} = \frac{1}{d}\cdot M^{C_{1},\ldots,C_{k}}_{A_{1},\ldots,A_{k}} \gamma_{C_{1},\ldots,C_{k}},
\qquad
M^{C_{1},\ldots,C_{k}}_{A_{1},\ldots,A_{k}} =
 \tr \, \{ A_{1}C_{1}\cdots A_{k}C_{k} \}, \label{eq:normal}
\end{align}
where tensor contractions are implicit following the Einstein summation convention.
This shows we can compute OTO correlators $\alpha_{A_{1},\ldots,A_{k}}$ the coefficients defining the $k$-fold channel $\gamma_{C_{1},\ldots,C_{k}}$. To establish the converse, we must prove that the tensor $M$ is invertible.

\begin{theorem}\label{theorem:OTO_chaos}
Consider the tensor $M^{C_{1},\ldots,C_{k}}_{A_{1},\ldots,A_{k}}$ 
and its conjugate transpose ${M^{\dagger}}_{C_{1},\ldots,C_{k}}^{A_{1},\ldots,A_{k}}$. Then
\begin{align}
\sum_{C_{1},\ldots,C_{k}}{M^{\dagger}}_{C_{1},\ldots,C_{k}}^{A_{1}',\ldots,A_{k}'}M^{C_{1},\ldots,C_{k}}_{A_{1},\ldots,A_{k}} = d^{2k}\cdot  
\delta^{A_{1}'}_{A_{1}}\cdots \delta^{A_{k}'}_{A_{k}},\label{eq:tensor_M}
\end{align}
where $\delta^{P}_{Q}$ is the delta function for Pauli operators $P,Q$
\begin{align}
\delta^{P}_{Q} = 1, \quad (P=Q), \qquad \delta^{P}_{Q} = 0, \quad (P\not=Q).
\end{align}
\end{theorem}

The proof of this theorem is sort of technical and has been relegated to Appendix~\ref{sec:proof:oto-channnel}. Thus, from the OTO correlators $\alpha_{C_{1},\ldots,C_{k}}$, we can completely determine the $k$-fold channel $\gamma_{A_{1},\ldots,A_{k}}$
\begin{align}
\gamma_{C_{1},\ldots,C_{k}}= \frac{1}{d^{2k-1}} {M^{\dagger}}_{C_{1},\ldots,C_{k}}^{A_{1},\ldots,A_{k}}\alpha_{A_{1},\ldots,A_{k}}.
\end{align}
As an obvious corollary, this means that $2k$-point OTO correlators can measure whether or not an ensemble forms $k$-design.

\subsection{Frame potentials}

In this section we introduce the \emph{frame potential}, a single quantity that can measure whether an ensemble is a $k$-design. Furthermore, we show how the frame potential may be computed from OTO correlators.

Given an ensemble of unitary operators $\mathcal{E}$, the $k$th frame potential is defined by the following double sum \cite{Scott08}
\begin{align}
F_{\mathcal{E}}^{(k)} := \frac{1}{|\mathcal{E}|^{2}}\sum_{U,V\in \mathcal{E}} \left| \tr \{U^{\dagger}V\} \right|^{2k},
\end{align}
where $|\mathcal{E}|$ denotes the cardinality of $\mathcal{E}$.Denote the frame potential for the Haar ensemble as $F^{(k)}_{\text{Haar}}$. Then, the following theorem holds.

\begin{theorem}
For any ensemble $\mathcal{E}$ of unitary operators, 
\begin{align}
F^{(k)}_{\mathcal{E}} \geq F^{(k)}_{\mathrm{Haar}},
\end{align}
with equality if and only if $\mathcal{E}$ is $k$-design.
\end{theorem}

The proof of this theorem is very insightful and beautiful, which we reprint from~\cite{Scott08}. 

\begin{proof}
Letting $S=\int_{\mathcal{E}}(U^{\dagger})^{\otimes k}\otimes U^{\otimes k}  - \int_{\text{Haar}}(U^{\dagger})^{\otimes k}\otimes U^{\otimes k}$, we have
\begin{equation}
\begin{split}
0 \leq \tr \{ S^{\dagger}S\} =  &\int_{U\in\mathcal{E}}\int_{V\in\mathcal{E}}dUdV\, |\tr \{U^{\dagger}V \} |^{2k} - 2 \int_{U\in\mathcal{E}}\int_{V\in\text{Haar}}dUdV \, |\tr \{U^{\dagger}V \} |^{2k} \\
&+ \iint_{U,V\in\text{Haar}}dUdV \, |\tr \{ U^{\dagger}V \} |^{2k}.
\end{split}
\end{equation}
The first term is $F^{(k)}_{\mathcal{E}}$, and the third term is $F^{(k)}_{\text{Haar}}$ by definition. The second term is equal to  $- 2F^{(k)}_{\text{Haar}}$ by the left/right-invariance of the Haar measure.  Thus, we see that
\begin{align}
0 \leq F^{(k)}_{\mathcal{E}} - 2F^{(k)}_{\text{Haar}} + F^{(k)}_{\text{Haar}}
=F^{(k)}_{\mathcal{E}} - F^{(k)}_{\text{Haar}},
\end{align}
with equality if and only if $\mathcal{E}$ is $k$-design. 
\end{proof}

Note that we derived the minimal value of the frame potential in \S\ref{sec:review},
\be
F_{\text{Haar}}^{(k)}=k!,
\ee
 which holds for $k \le d$. The frame potential quantifies the $2$-norm distance between the Haar ensemble and the $k$-fold $\mathcal{E}$-channel.\footnote{The distance between two quantum channels is typically measured by the diamond norm, while the frame potential measures the $2$-norm (or more precisely the Frobenius norm) distance between a given ensemble and the Haar ensemble. Importantly, the $2$-norm is weaker than the diamond norm.  For precise statements regarding the different norms and bounds relating them, see~\cite{Low_thesis}.}  Here we show that the frame potential can be expressed as a certain average of OTO correlation functions. 

\begin{theorem}
For any ensemble $\mathcal{E}$ of unitary operators,
\begin{align}
\frac{1}{d^{4k}}\sum_{A_{1},\cdots,B_{1},\cdots}\left|\big\langle A_{1}\tilde{B_{1}}\cdots A_{k}\tilde{B_{k}} \big\rangle_{\mathcal{E}}\right|^{2} = \frac{1}{d^{2(k+1)}}\cdot F^{(k)}_{\mathcal{E}}\label{eq:OTO_frame},
\end{align}
where summations are over all possible Pauli operators.
\end{theorem} 

The LHS of the equation is the operator average of the $2$-norm of OTO correlators, and the RHS is the $k$th frame potential up to a constant factor. There are $d^{4k}$ Pauli operators $A_{1},\ldots, A_k, B_{1},\ldots, B_k$, which leads to $1/d^{4k}$. The theorem implies that the quantitative effect of random unitary evolution is to decrease the frame potential, which is equivalent to the decay of  OTO correlators.

\begin{proof}
We take the averages over $A_{1},\cdots,A_k, B_{1},\cdots, B_k$ first. Expanding the LHS gives 
\begin{align}
\frac{1}{d^{4k}}\frac{1}{|\mathcal{E}|^{2}d^2}\sum_{U,V\in \mathcal{E}}\sum_{A_{1},\cdots,B_{1},\cdots}\tr \{A_{1}U^{\dagger}B_{1}U\cdots A_{k}U^{\dagger}B_{k}U \} \cdot\tr \{ V^{\dagger}B_{k}^{\dagger}VA_{k}^{\dagger}\cdots V^{\dagger}B_{1}^{\dagger}VA_{1}^{\dagger}\}.
\end{align}
For $k=2$, this can be depicted graphically as
\begin{align}
\includegraphics[width=0.40\linewidth]{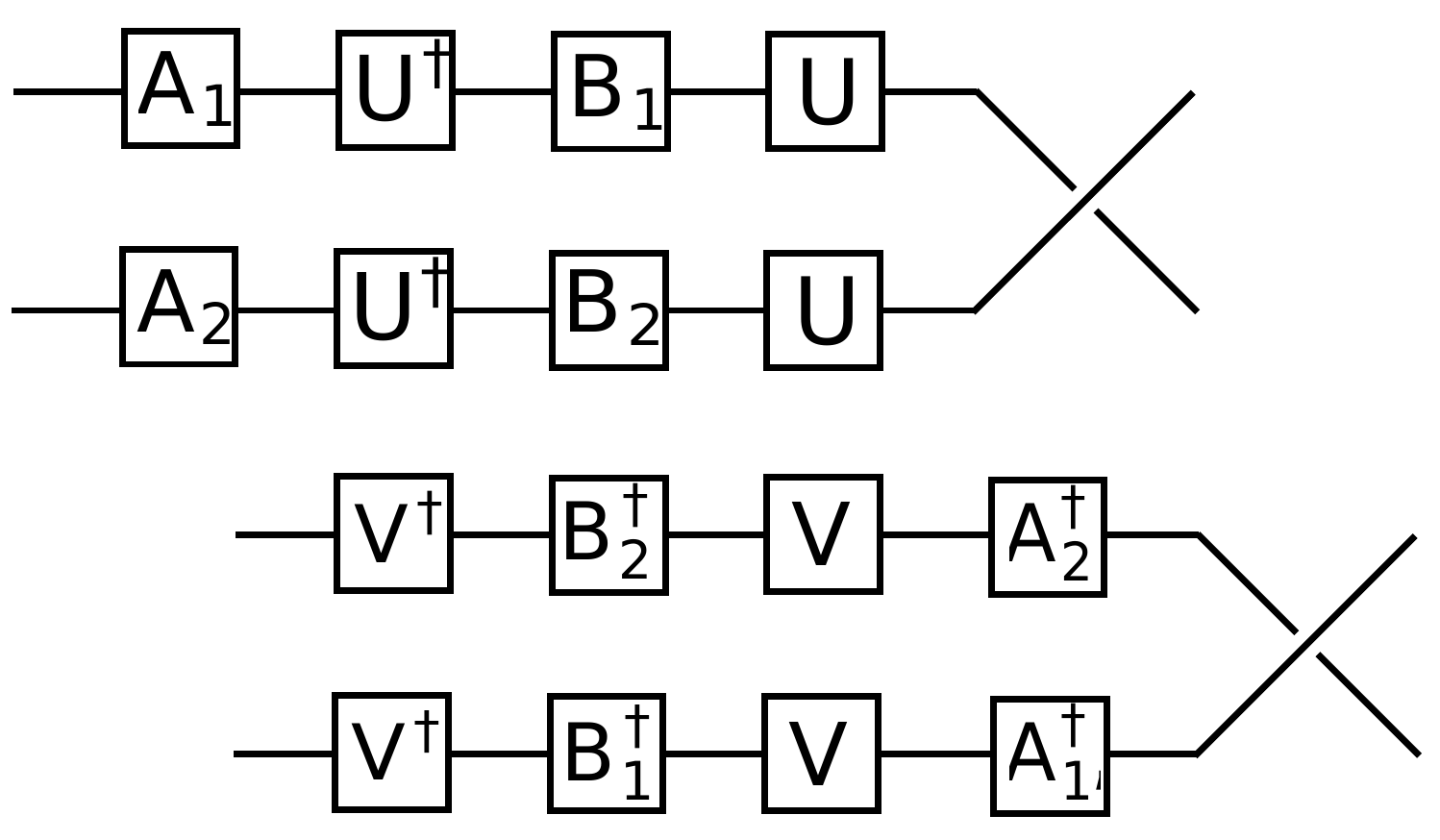}.
\end{align}
Recall that a SWAP operator is given by $\text{SWAP} = \frac{1}{d}\sum_{P} P \otimes P^{\dagger}$
\begin{align}
\includegraphics[width=0.35\linewidth]{fig_SWAP}.
\end{align}
Thus, we replace each average of Pauli operators $A_{1},\cdots,A_k, B_{1},\cdots, B_k$ by SWAP operators. There are $2k$ loops, where $k$ of them contribute to $\tr \{UV^{\dagger}\}$ and the remaining $k$ loops contribute to  $\tr \{ VU^{\dagger} \}$. Keeping track of the number of factors of $d$, we find 
\begin{align}
\frac{1}{d^{2(k+1)}|\mathcal{E}|^2}\sum_{U,V\in\mathcal{E}}\tr \{ UV^{\dagger} \} ^k\cdot \tr \{ VU^{\dagger}\}^k = \frac{1}{d^{2(k+1)}} \cdot F^{(k)}_{\mathcal{E}},
\end{align}
which is the desired result.
\end{proof}

Lastly, we note that we can use the frame potential to lower bound the size of the ensemble. Since all the terms in the sum are positive, by taking the diagonal part of the double sum, we find
\begin{align}
F_{\mathcal{E}}^{(k)} \geq \frac{1}{|\mathcal{E}|^2}\sum_{U,V\in \mathcal{E},\, U=V} \big|\tr \{I \}\big|^{2k}=
 \frac{1}{|\mathcal{E}|}d^{2k}, \label{frame-potential-vs-size}
\end{align}
meaning the cardinality is lower bounded as $|\mathcal{E}| \ge d^{2k} / \Fk$. Using the fact that $F^{(k)}_\mathrm{Haar} = k!$, we can simply bound the size of a $k$-design\footnote{Actually, one can slightly improve this lower bound for a $k$-design, see~\cite{Roy:2009aa, Brandao12}.}
\be
\mathcal{|E|} \ge \frac{d^{2k}}{k!}, \qquad \text{($k$-design)}. \label{eq-cardinality-of-k-design}
\ee

\subsubsection*{Large $k$}

In Appendix~\ref{sec:appendix:orthogonal}, we provided an intuitive way of counting the number of nearly orthogonal states in a high-dimensional vector space $\mathbb{C}^d$ by introducing a precision tolerance $\epsilon$. A similar argument holds for the number of operators; if we can only distinguish operators up to some precision $\epsilon$, then the total number of unitary operators is given by the volume of the unitary group measured in balls of size $\epsilon$, which roughly goes like $\sim \epsilon^{-2^{2n}}$.

The key point is that if we have a tolerance $\epsilon$, then the number of operators is finite even for a continuous ensemble. This means that there is a maximum size to any ensemble that is a subset of the unitary group. Looking at our bound Eq.~\eqref{eq-cardinality-of-k-design}, we see that for $k\sim d$ we begin to reach the maximum size for some $\epsilon$. 

As a corollary, this means that for large $k \gtrsim d$, being a $k$-design implies that the ensemble is an approximate $k+1$-design. In this sense, there are really only $O(d)$ nontrivial moments of the Haar ensemble.

As an explicit example, we will consider the case of $d=2$. Here, for Haar we can compute the frame potential exactly for all $k$ \cite{Scott08}
\begin{align}
F^{(k)}_{\text{Haar}} = \frac{(2k)!}{k ! (k+1)!},
\end{align}
and thus we have the following relationship
\begin{align}
\frac{F_{\text{Haar}}^{(k+1)}}{F_{\text{Haar}}^{(k)}} = 4 \frac{k+1/2}{k+2}.
\end{align}
On the other hand, for any ensemble the frame potential satisfies
\begin{align}
F^{(k+1)} \leq d^2 F^{(k)}=4F^{(k)}.
\end{align}
Therefore, for $d=2$ we see that  if an ensemble is a $k$-design, for $k \gtrsim 2$ it will automatically be close to being a  $k+1$-design.

\subsection{Ensemble vs. time averages}\label{sec:F-time-average}

In this subsection, we consider the time average of the frame potential
to ask whether the ensemble specified by sampling $U(t)=e^{-iHt}$ at different times $t$ will ever form a $k$-design.  In the classical statistical mechanics and chaos literature, this is the question of whether a system is ergodic---whether the average over phase space equals an average over the time evolution of some initial state.

Consider the one-parameter ensemble of unitary matrices defined by evolving with a fixed time-independent Hamiltonian $H$, 
\be
\mathcal{E} = \big\{ e^{-iHt} \big\}_{t=0}^{\infty}.
\ee 
We can compute the frame potential as
\begin{align}
F_{\mathcal{E}}^{(k)} &= \lim_{T\to\infty} \frac{1}{T^2}\int_0^T dt_1 dt_2 \, \Big|\trr{e^{-iH(t_1 - t_2)}}\Big|^{2k}, \\
               &= \lim_{T\to\infty} \frac{1}{T^2}\int_0^T dt_1 dt_2 \, \sum_{\ell_{1}\dots \ell_{k}} \sum_{m_{1}\dots m_{k}} \exp \Big\{ -i (t_2 - t_1) \sum_{i=1}^k E_{\ell_{i}} + i(t_2-t_1)\sum_{j=1}^k E_{m_{j}}  \Big\}. \notag
\end{align}
If the spectrum is chaotic/generic, then the energy levels are all incommensurate. The terms in the sum evaluate to zero unless we have $E_{\ell_{i}} = E_{m_{j}}$ for all possible pairings $(i,j)$
\begin{align}
F_{\mathcal{E}}^{(k)} &= \sum_{\ell_{1}\dots \ell_{k}} \sum_{m_{1}\dots m_{k}} \sum_{ij} \delta_{\ell_{i}m_{j}}, \\
               &= k! \Big(\sum_{\ell} I \Big)^k = k!\,d^k. \notag
\end{align}
This is larger than $F_{\text{Haar}}^{(k)}$ by a factor of $d^k$; the ensemble $\mathcal{E}=\{e^{-iHt}\}$ does not form a $k$-design!
Recall that $F_{\mathcal{E}}=d^{2k}$ for a trivial ensemble $\mathcal{E}=\{I\}$ while $F_{\mathcal{E}}=k!$ for a Haar ensemble. The time average ensemble sits in the middle of these ensembles.

This ``half-randomness''can be understood in the following way. Let us write the Hamiltonian as
\begin{align}
H = \sum_{j} E_{j} |\psi_{j}\rangle \langle \psi_{j}|.
\end{align}
Rotating $H$ by a unitary operator does not affect the frame potential, so we can consider a classical Hamiltonian 
\begin{align}
H' = \sum_{j} E_{j} |j\rangle \langle j|,
\end{align}
with the same frame potential. Even though $H'$ is classical, it has an ability to generate entanglement. Namely, if an initial state is $|+\rangle=\sum_{j}|j\rangle$, then time-evolution will create entanglement. In terms of the original Hamiltonian $H$, we see that the system keeps evolving in a non-trivial manner as long as the initial state is ``generic'' and is different from eigenstates of the Hamiltonian $H$. This arguments suggests that the frame potential, in a time-average sense, can only see the spectrum distribution.

This fact, that Hamiltonian time evolution can never lead to random unitaries, can also be understood in terms of the distribution of the spacing of the eigenvalues.\footnote{This argument has been made by Michael Berry, who related it to Steve Shenker, who related it to Douglas Stanford, who related it to us. As far as we know, it does not otherwise appear in the literature.} The phases of the Haar random unitary operators have Wigner-Dyson statistics; they are repulsed from being degenerate. The eigenvalues of a typical Hamiltonian $H$ also have this property. However, these eigenvalues live on the real line, while the phases of $e^{-iHt}$ live on the circle. In mapping the line to the circle, the eigenvalues of $H$ wrap many times. This means that the difference of neighboring phases of $e^{-iHt}$ will not be repulsed; $e^{-iHt}$ will have Poisson statistics! In this sense, the ensemble formed by sampling $e^{-iHt}$ over time may never become Haar random.\footnote{Relatedly, the space of unitaries with a fixed Hamiltonian is $d$-dimensional, while the space of unitaries is $d^2$-dimensional.} The calculation in the beginning of this subsection provides an explicit calculation of the $2$-norm distance between such an ensemble and $k$-designs and quantifies the degree to which the ensemble average is ``too random'' as compared to the time average.

\section{Measures of complexity}\label{sec:complexity-bound}
For an operator or state, the computational complexity is a measure of the minimum number of basic elementary gates necessary to generate the operator or the state from a simple reference operator (e.g. the identity operator) or reference state (e.g. the product state). With the assumption that the elementary gates are easy to apply, the complexity is a measure of how difficult it is to create the state or operator.

On the other hand, an ensemble contains many different operators with different weights or probabilities. In that case, the computational complexity of the ensemble should be understood as the number of steps it takes to generate the ensemble by probabilistic applications of elementary gates. For instance, to generate the ensemble of Pauli operators, we randomly choose with probability $1/4$ a Pauli operator $I,X,Y,Z$ to apply to a qubit, and then repeat this procedure for all the qubits.

The complexity of an ensemble is related to the complexity of an operator in the following way. If an ensemble can be prepared in $\mathcal{C}$ steps, then all the operators in the ensemble can be generated by applications of at most $\mathcal{C}$ elementary gates. On the other hand, if an ensemble cannot be prepared (or approximated) in $\mathcal{C}$ steps, then---for the sorts of ensembles we are interested in---most of the operators  cannot be generated by applications of $\mathcal{C}$ elementary gates. 
For example, generating the Haar ensemble will take exponential complexity since, on average, individual elements have exponential complexity.

The complexity of the ensemble can be lower bounded in terms of the number of elements or cardinality of the ensemble $|\mathcal{E}|$. If all the elements are represented equally (with uniform probabilities), then clearly at least $\mathcal{E}$ circuits need to be generated from probabilistic applications of the elementary gates. Making use of a fact introduced in the previous section that $\Fk$ provides a lower bound on $|\mathcal{E}|$, here we show that the frame potential provides a lower bound on a circuit complexity of generating $\mathcal{E}$. We will also explain how this bound applies to ensembles that depend continuously on some parameters and thus have a divergent number of elements. 

Two additional bounds that are somewhat outside the scope of the main presentation, one on circuit depth and one on the early-time complexity growth with a disordered ensemble of Hamiltonians, are relegated to Appendix~\ref{sec:complexity-appendix}.

\subsection{Discrete ensembles}\label{sec:complexity-discrete}

Consider a system of $n$ qubits. Let $\mathcal{G}$ denote an elementary gate set that consists of a finite number of two-qubit quantum gates. We denote the number of elementary two-qubit gates by $g:=|\mathcal{G}|$. At each time step we assume that we can implement any of the gates from $\mathcal{G}$. One typically chooses $\mathcal{G}$ so that gates in $\mathcal{G}$ enables a universal quantum computation. A well-known example is
\begin{align}
\mathcal{G} = \{\text{$2$-qubit Clifford}, T\},  
\end{align}
where $T$ is the $\pi/4$ phase shift operator; $T=\text{diag}(1,e^{i\pi/4})$. (Of course, this is not the only choice of elementary gate sets.) 

Our goal is to generate an ensemble of unitary operators $\mathcal{E}$ by sequentially implementing quantum gates from $\mathcal{G}$. Let us denote the necessary number of steps (i.e. the circuit complexity) to generate $\mathcal{E}$ by $\mathcal{C}(\mathcal{E})$. Then one has the following complexity lower bound. 
\begin{theorem}
Let $g$ be the number of distinct two-qubit gates from the elementary gate set. Then the circuit complexity $\mathcal{C}(\mathcal{E})$ to generate an ensemble $\mathcal{E}$ is lower bounded by 
\begin{align}
\mathcal{C}(\mathcal{E}) \geq \frac{\log |\mathcal{E}| }{\log(g n^2)}.\label{eq-for-theorem-complexity-ensemble-size-bound}
\end{align}
\label{theorem-complexity-ensemble-size-bound}
\end{theorem}

The proof relies on an elementary counting argument. Arguments along this line of thought have been used commonly in the literature.

\begin{proof}
 At each step, we randomly pick a pair of qubits. Since there are $g$ implementable quantum gates and $\binom{n}{2}$ qubit pairs, there are in total $\simeq g n^2$ choices at each step. 
If this procedure runs for $\mathcal{C}$ steps, the number of unique circuits this procedure can implement is upper bounded by
\begin{align}
\text{$\#$ of circuits}\leq (g n^2)^{\mathcal{C}}.
\end{align}
Since there are $|\mathcal{E}|$ unitary operators in an ensemble $\mathcal{E}$, we must have
\begin{align}
(g n^2)^{\mathcal{C}} \geq |\mathcal{E}|, 
\end{align}
which implies
\begin{align}
\mathcal{C}(\mathcal{E}) \geq \frac{\log |\mathcal{E}| }{\log(g n^2)}. 
\end{align} 
\end{proof}

In Appendix~\ref{sec:appendix:orthogonal}, we provide an intuitive way of counting the number of states in a $d=2^n$-dimensional Hilbert space (which was well known for a long time from \cite{Knill95}). 
As a sanity check on Theorem~\ref{theorem-complexity-ensemble-size-bound}, if we substitute $|\mathcal{E}| \gtrsim 2^{2^n}$ into Eq.~\eqref{eq-for-theorem-complexity-ensemble-size-bound} we see that the complexity 
of most states is exponential in the number of qubits $\mathcal{C} > 2^n \log 2$.

Finally, let us examine the relation between the frame potential and the circuit complexity. Using Eq.~\eqref{frame-potential-vs-size} to trade $|\mathcal{E}|$ for $\Fk$, we find
\begin{align}
\mathcal{C}(\mathcal{E}) \geq \frac{2kn \log(2) - \log F_{\mathcal{E}}^{(k)} }{\log(g n^2)}.
\end{align}
Of course, this bound obviously depends on the choice of basic elements $g$ and the fact that we are using two-qubit gates. If we had considered $q$-body gates, we would have found a denominator $\log(g\binom{n}{q}) \approx \log(gn^q)$ for $n\gg q$.
This is no more than a choice of ``units'' with which to measure complexity. Thus, we state our key result as:
\begin{theorem}\label{theorem-bound}
For an ensemble $\mathcal{E}$ with the $k$th frame potential $F_{\mathcal{E}}^{(k)}$, the circuit complexity is lower bounded by
\begin{align}
\mathcal{C}(\mathcal{E}) \geq \frac{2kn \log(2) - \log F_{\mathcal{E}}^{(k)} }{\log( \mathrm{choices})},\label{eq:lower-bound}
\end{align}
\end{theorem}
In this context, $\log( \mathrm{choices})$ simply indicates the logarithm of the number of decisions that are made at each step. If we imagine we have some kind of decision tree for determining which gate to apply where we make a binary decision at each step (and use $\log_2$), then we may set the denominator to unity and measure complexity in bits rather than gates, i.e. $\mathcal{C}(\mathcal{E}) \geq 2kn  - \log_2 F_{\mathcal{E}}^{(k)}$.

In the above discussion, we glossed over a subtlety. Here, we considered the quantum circuit complexity to prepare or approximate an entire ensemble $\mathcal{E}$. A closely related but different question concerns the quantum circuit complexity required to implement a typical unitary operator from the ensemble $\mathcal{E}$. Nevertheless, in ordinary settings the typical operator complexity and the ensemble complexity are roughly of the same order. While establishing a rigorous result in this direction is beyond the scope of this paper, see~\cite{Knill95} for some basic proof techniques that are useful in establishing this connection. 

An important consequence of Theorem~\ref{theorem-bound} is that the smallness of the frame potential (i.e. generic smallness of OTO correlators) implies increases in quantum circuit complexity of generating the ensemble $\mathcal{E}$. As a corollary, we can simply rewrite Eq.~\eqref{eq:lower-bound} as
\begin{align}
\mathcal{C}(\mathcal{E}) \geq(2k-1)2n   - \log_2  \sum_{A_{1},\cdots,B_{1},\cdots}\left|\big\langle A_{1}\tilde{B_{1}}\cdots A_{k}\tilde{B_{k}} \big\rangle_{\mathcal{E}}\right|^{2}.\label{eq:lower-bound-oto}   
\end{align}
In this sense, we see how the decay of OTO correlators is directly related to an increase in the (lower-bound) of the complexity.

Next, recall that the frame potential for a $k$-design is given by $F^{(k)}_{\text{Haar}}=k!$, which does not grow with $n$. The complexity of a $k$-design is thus lower bounded as
\be
\mathcal{C}(k\text{-design}) \ge 2kn  - k\log_2 (k) + \log_2 k + \dots, \label{eq:complexity-k-design}
\ee
which for large $n$ grows roughly linearly in $k$ and $n$ (at least). In \S\ref{sec:complexity:depth}, we show that the minimum depth circuit to make a $k$-design also growths linearly in $k$ (at least).

Finally, we offer an additional information-theoretic interpretation for our lower bound and generalize it for ensembles with non-uniform probability distributions. Consider an ensemble $\mathcal{E}=\{ p_{j}, U_{j} \}$ with probability distribution $\{p_{j} \}$ such that $\sum_{j}p_{j}=1$. The second R\`{e}nyi entropy of the distribution $\{p_{j} \}$ is defined as $S^{(2)}= - \log \big( \sum_{j} p_{j}^2 \big)$. In this more general situation, we can still bound the frame potential by considering the diagonal part of the sum
\begin{align}
F^{(k)}_{\mathcal{E}} = \sum_{i,j} p_{i}p_{j} \, \big|\tr \{ U_{i}U_{j}^{\dagger} \}\big|^{2k} \geq \sum_{i} p_{i}^2 \, \big|\tr \{ I \} \big|^{2k} = e^{-S^{(2)}}d^{2k}. 
\end{align}
Since the von Neumann entropy $S^{(1)} = - \sum_{j} p_j \log(p_j)$ is always greater than the second R\`enyi entropy, $S^{(1)} \geq S^{(2)}$, we can bound the von Neumann entropy as
\begin{align}
S^{(1)} \geq 2kn - \log_2 F^{(k)}_{\mathcal{E}}.
\end{align}
The entropy of the ensemble is a notion of complexity measured in bits.\footnote{However, this is not to be confused with the entanglement entropy. The entropy of the ensemble is essentially the logarithm of the number of different operators, and therefore can be exponential in the size of the system. Instead, the entanglement entropy (as a measure of entanglement) can only be as large as (half) the size of the system.}

\subsection{Continuous ensembles}\label{sec:continuous}

Many interesting ensembles of unitary operators are defined by continuous parameters, e.g. a disordered system has a time evolution that may be modeled by an ensemble of Hamiltonians.\footnote{A notable example of recent interest to the holography and condensed matter community is the ensemble implied by time evolving with the Sachdev-Ye-Kitaev Hamiltonian \cite{Sachdev:1992fk,Kitaev:2014t2,Maldacena:2016hyu}.} While the counting argument in \S\ref{sec:complexity-discrete} is not directly applicable to these systems with continuous parameters, the complexity lower bound generalizes to such systems by allowing approximations of unitary operators. To be concrete, imagine that we wish to create some unitary operator $U_{0}$ by combining gates from the elementary gate set. In practice, we do not need to create an exact unitary operator $U_{0}$. Instead, we may be fine with preparing some $U$ that faithfully approximates $U_{0}$ to within a trace distance
\begin{align}
||U_{0}-U||_{1}\leq \epsilon,
\end{align}
where the notation $||\cdot||_p \equiv (\tr \, \{ |\cdot|^p \} )^{1/p}$ specifies the $p$-norm of an operator.

Now, let us derive a complexity lower bound up to an $\epsilon$-tolerance. We begin by taking $N_s$ samples from the ensemble and use them to estimate the frame potential of the continuous distribution
\be
\Fk \approx \frac{1}{N_s^2}\sum_{i,j} \big|\tr\, \{ U_i^{\dagger}V_j \} \big|^{2k},\label{sampled-frame-potential}
\ee
where the each of the two sums runs over all $N_s$ samples. We can lower bound Eq.~\eqref{sampled-frame-potential} as follows
\be
\Fk \ge \frac{1}{N_s^2}\sum_{i} \sum_{\substack{j, \\ \mathllap{||}U_i\mathrlap{ - V_j||_1 < \epsilon} } } \big|\tr\, \{ U_i^{\dagger}V_j \} \big|^{2k},\label{epsilon-frame-potential}
\ee
where the sum over $i$ runs from $1$ to $N_s$. The sum over $j$ includes only a smaller subset $N_\epsilon(U_i)$, which is the number of operators within a trace distance $\epsilon$ of a particular $U_i$.

To continue, let's bound the summand. First, note that
\be
\big| \tr \{ U^\dagger V \} \big| > \text{Re} \big\{ \, \tr \, \{  U^\dagger V \} \big\} = d - \frac{1}{2}||U-V||_2^2.
\ee
The $2$-norm is upper bounded by the $1$-norm as $|| \mathcal{O} ||_2 \le \sqrt{d}  || \mathcal{O}||_1$, which let's us rewrite this as
\be
\big| \tr\, \{ U^{\dagger}V \} \big|^{2k} > d^{2k}\bigg(1 - \frac{1}{2}||U-V||_1^2\bigg)^{2k}.\label{summand-bound}
\ee
For this formula to be sensible, this approximation requires that $\epsilon < \sqrt{2}$. Substituting Eq.~\eqref{summand-bound} into Eq.~\eqref{epsilon-frame-potential}, we can bound the frame potential
\be
\Fk \ge  \frac{1}{N_s} d^{2k}\bigg(1 - \frac{\epsilon^2}{2}\bigg)^{2k} \bigg[ \frac{1}{N_s}  \sum_{i=1}^{N_s} N{_\epsilon}(U_i) \bigg] = d^{2k}\bigg(\frac{ \overline{N_\epsilon} }{N_s} \bigg) \bigg(1 - \frac{\epsilon^2}{2}\bigg)^{2k}. 
\ee
The term in brackets in the middle expression is the average number of operators within a trace distance $\epsilon$ of an operator in our sample set. In the final expression this is represented by the symbol $\overline{N_\epsilon}$.

Now, let's run the counting argument again. If we want to make $N_s$ circuits exactly, then in $\mathcal{C}$ steps we must have
\be
(\text{choices})^\mathcal{C} > N_s,
\ee
where as before $(\text{choices})$ summarizes the information about our choice of gate set, etc. Instead, if you only care about making circuits to within an $\epsilon$-accuracy, then in $\mathcal{C}_\epsilon$ steps $N_s$ instead satisfies
\be
(\text{choices})^{\mathcal{C}_\epsilon} ~  \overline{N_\epsilon}  > N_s.
\ee
This lets us lower bound the complexity of the ensemble at precision $\epsilon$ as
\be
\mathcal{C}_\epsilon(\mathcal{E}) > \frac{2k \log(d) - k \epsilon^2 - \log \Fk}{\log (\text{choices}) }.  \label{complexity-boudnd-eps}
\ee
We then take the continuum limit by taking $N_s \to \infty$. The number of operators within an $\epsilon$-ball of a given sample will also diverge, but the ratio $ N_s/ \overline{N_\epsilon}$ should remain finite and converge to some value, roughly the volume of the ensemble as measured in balls of $\epsilon$-radius.\footnote{Note that $\log(\text{choices})$ diverges if the elementary gate set is continuous. This problem may be also fixed by employing the $\epsilon$-tolerance.}

Finally, in $\S\ref{sec:complexity:early}$, we further extend this notion of bounding complexity for continuous ensembles to show that the initial early-time growth of complexity for evolution with an ensemble of Hamiltonians grows initially as $t^2$ for a time $t < 1/\sqrt{\log(d)}$.

\section{Measures of correlators}\label{sec:haar-averages}

While much of the focus of this paper has been on the behavior of ensembles, we were originally motivated by the following question: When is a random unitary operator an appropriate approximation to the underlying dynamics? In this section, we will attempt to return the focus to this question by computing random averages over correlation functions and comparing them to expectations for chaotic time evolution in physical systems.

\subsection{Haar random averages}

In this subsection, we will explicitly compute some ensemble averages of OTO correlators for different choices of ensembles. A particular goal will be to understand the asymptotic behavior of these averages in the limit of a large number of degrees of freedom $d\to \infty$. We present explicit calculations for $2$-point and $4$-point functions here and provide results for  $6$-point and $8$-point functions.  Additional calculations may be found in Appendix~\ref{sec:appendix-averages}.

\subsubsection*{$2$-point functions}

Consider a $2$-point correlator, averaged over Haar random unitary operators
\begin{align}
\langle A  \tilde{B} \rangle_{\text{Haar}}=\frac{1}{d}\int_{\text{Haar}} dU \, \tr \, \{ A \, U^{\dagger}BU\} .
\end{align}
Since $U$ and $U^\dagger$ each only appear in the expression once, we will obtain the same answer if the average is instead performed over a 1-design: $\langle A  \tilde{B} \rangle_{\text{Haar}} = \langle A \tilde{B} \rangle_{\text{1-design}}$. By using a formula from \S\ref{sec:review}, we can derive the following expression
\begin{align}
\langle A  \tilde{B} \rangle_{\text{Haar}} = \langle A\rangle\langle B \rangle. \label{2-pt-haar}
\end{align}
Graphically, the calculation goes as follows
\begin{align}
\includegraphics[width=0.65\linewidth]{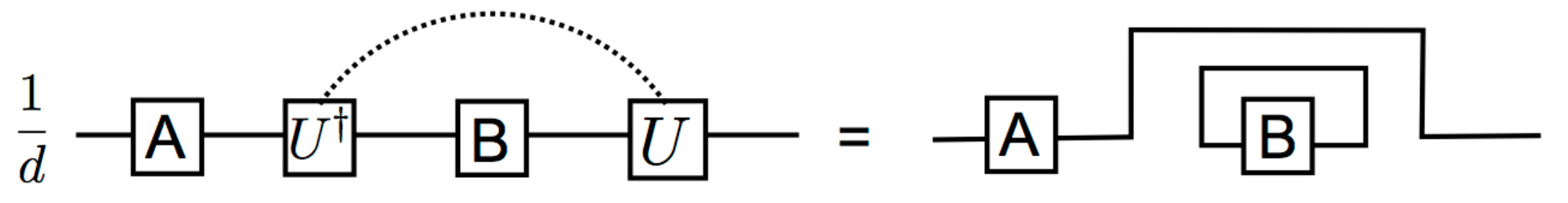}.
\end{align}
It is often convenient to consider physical observables with zero mean by shifting $A \rightarrow A -\langle A\rangle$. Then, we see that these $2$-point correlation function vanish. Of course, this always holds for Pauli operators
\begin{align}
\langle A  \tilde{B} \rangle_{\text{Haar}} = 0, \qquad A,B\in \mathcal{P},\quad A,B\not=I.
\end{align}

Next, let us consider the norm squared of a $2$-point correlator averaged over Haar random unitary operators $|\langle A  \tilde{B} \rangle|^2_{\text{Haar}}=\frac{1}{d^2}\int_{\text{Haar}} dU\, \tr \{ A\,U^{\dagger}BU \} \, \tr \{ U^{\dagger}B^\dagger U \, A^{\dagger}\}$. Note that we take the Haar average after squaring the correlator. Since there are two pairs of $U$ and $U^{\dagger}$ appearing, we can perform the average over a $2$-design: $| \langle A  \tilde{B} \rangle|^{2}_{\text{Haar}} = |\langle A  \tilde{B} \rangle|^{2}_{\text{2-design}}$. Let us assume that $A,B$ are Pauli operators so that we can neglect contributions from $\langle A\rangle$ and $\langle B\rangle$. There are four terms, but only one term survives because the trace of non-identity Pauli operators is zero. We depict the calculation graphically as
\begin{align}
\includegraphics[width=0.6\linewidth]{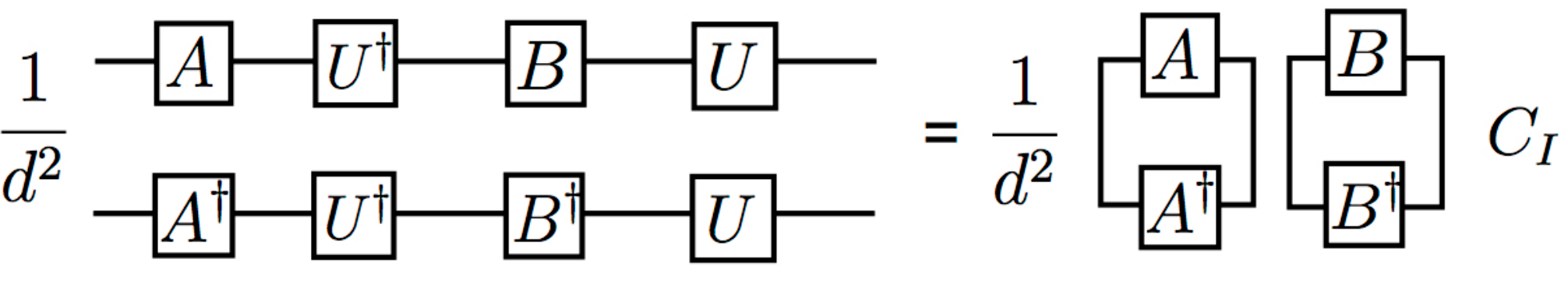}.
\end{align}
where $C_{I}=1/(d^2-1)$ comes from the Weingarten function as shown in \S\ref{sec:review}.  The final result is
\begin{align}
|\langle A  \tilde{B}\rangle|^2_{\text{Haar}}=\frac{1}{d^2-1},\qquad A,B\in \mathcal{P},\quad A,B\not=I.\label{eq:2-pt-variance}
\end{align}
Thus, the variance of the Haar averaged $2$-point function is exponentially small in the number of qubits. 

\subsubsection*{4-point functions}

Next, consider a $4$-point OTO correlator averaged over Haar random unitary operators: 
\begin{align}
\langle A \tilde{B} C \tilde{D} \rangle_{\text{Haar}} = \frac{1}{d}\int_{\text{Haar}} dU \tr \{ A\, U^\dagger B U \,C \, U^\dagger D U \}. 
\end{align}
As has already been explained, we will obtain the same answer if the average is performed with a 2-design: $\langle A \tilde{B} C \tilde{D} \rangle_{\text{Haar}} = \langle A \tilde{B} C \tilde{D} \rangle_{\text{2-design}}$. By using formulas from \S\ref{sec:review} we can derive the following expression\footnote{Note: this formula was independently obtained by Kitaev.}
\begin{equation}
\langle A \tilde{B} C \tilde{D}  \rangle_{\text{Haar}} = \langle AC \rangle \langle B \rangle\langle D \rangle + \langle A \rangle\langle C \rangle \langle BD \rangle -  \langle A\rangle \langle C \rangle \langle B \rangle\langle D \rangle - \frac{1}{d^2-1}\langle\!\langle AC\rangle\!\rangle \langle\!\langle BD\rangle\!\rangle, \label{eq:OTO_formula}
\end{equation}
where $d=2^n$ and $\langle \! \langle AC\rangle \! \rangle\equiv\langle AC\rangle-\langle A\rangle\langle C\rangle$. In particular, for Pauli operators $A,B,C,D\not=I$, one has
\begin{equation}
\begin{split}
\langle A \tilde{B} C \tilde{D} \rangle_{\text{Haar}} &= - \frac{1}{d^2-1} \qquad (A=C^{\dagger},\ B=D^{\dagger} ), \\
&= 0 \qquad \quad \qquad (\text{otherwise}).
\end{split}
\end{equation}
When nonzero, the result is exponential small in the number of qubits $n$. The derivation of the aforementioned formula can be understood graphically as follows
\begin{equation}
\includegraphics[width=0.88\linewidth]{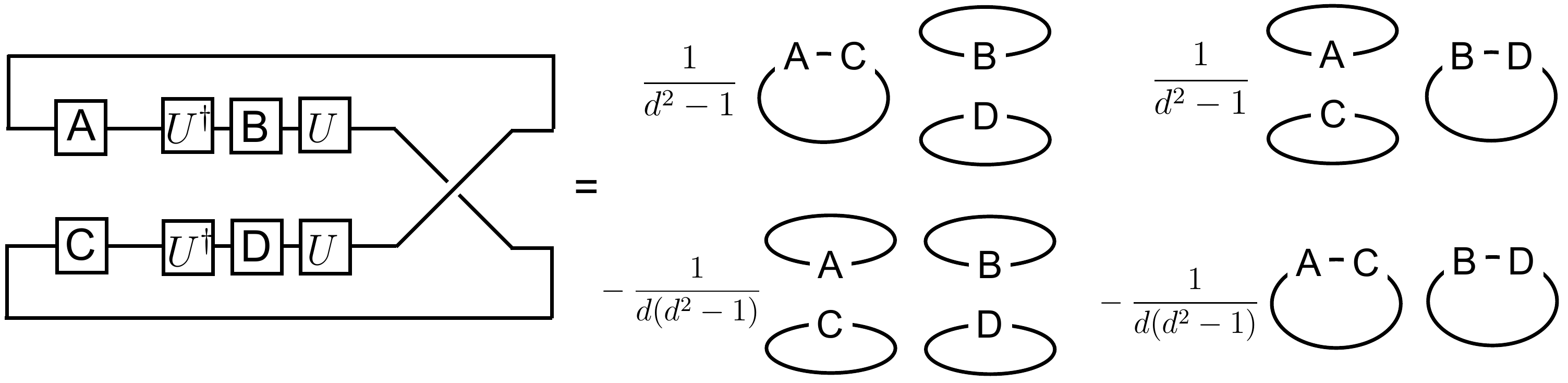}.
\end{equation}
By rewriting the expression in terms of connected correlators, we obtain the formula Eq.~\eqref{eq:OTO_formula}.

Of course, we can obtain the same result by instead averaging over the Clifford group. When $A\not=C^{\dagger}$ or $B\not=D^{\dagger}$, it is easy to show $\langle A\tilde{B}C\tilde{D} \rangle_{\text{Clifford}}=0$.  When $A=C^{\dagger}$ and $B=D^{\dagger}$, the value of the correlator depends on the commutation relations between $A$ and~$B$
\begin{equation}
\begin{split}
\langle A\tilde{B}A^{\dagger}\tilde{B}^{\dagger} \rangle&=1, \quad \qquad [A,B]=0,\\
&=-1, \qquad \{A,B\}=0.
\end{split}
\end{equation}
Since by definition Clifford operators transform a Pauli operators to another Pauli operator, $\tilde{B}$ is a random Pauli operator with $\tilde{B}\not=I$. There are $d^2-1$ non-identity Pauli operators. Among them, $d^2/2-1$ Pauli operators commute with $A$, and $d^2/2$ anti-commute with $A$. Therefore, we find
\begin{align}
\langle A\tilde{B}A^{\dagger}\tilde{B}^{\dagger} \rangle_{\text{Clifford}}
= \frac{d^2/2-1}{d^2-1} - \frac{d^2/2}{d^2-1} = \frac{-1}{d^2-1}.
\end{align}
As expected, the $4$-point OTO correlator ensemble average over the Clifford group equals the Haar average. Recalling our result from \S\ref{sec:OTO_channel} that $4$-point OTO values completely determine the $2$-fold channel, this explicit calculation gives an alternative proof that the Clifford group forms a unitary $2$-design.

Finally, we will present one additional way of computing the Haar average of $4$-point OTO correlators. 
For convenience, we introduce the following notation
\begin{align}
\text{OTO}^{(4)}(A,B) :=\langle A \tilde{B} A^{\dagger}\tilde{B}^{\dagger} \rangle_{\text{Haar}}.
\end{align}
Since $U$ is sampled uniformly over the unitary group, $\text{OTO}^{(4)}(A,B)$ does not depend on $A,B$ as long as $A,B\not=I$. Consider the average of OTO correlation functions over all Pauli operators $A$, including identity operators: $\sum_{A} \langle A \tilde{B} A^{\dagger}\tilde{B}^{\dagger} \rangle_{\text{Haar}}$. Since $\frac{1}{d}\sum_{A}A\otimes A^{\dagger}$ is the swap operator, we have
\begin{align}
\includegraphics[width=0.65\linewidth]{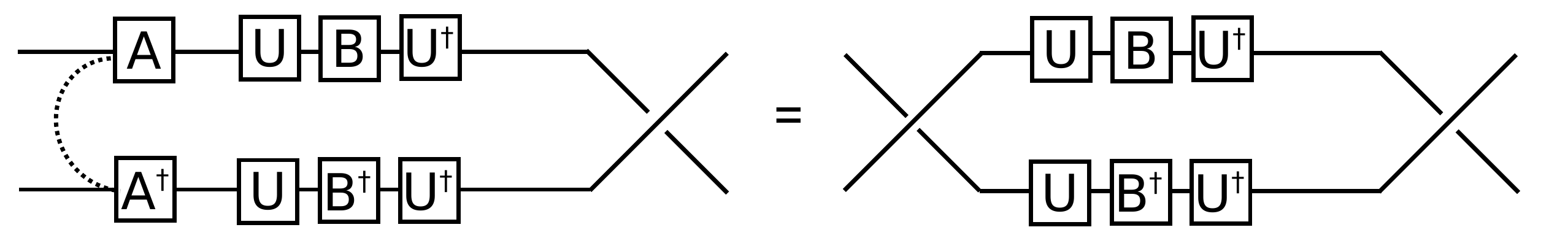},
\end{align}
where a dotted line represents an average over all the Pauli operators. This expression must be zero for $B\not=I$. Thus, we have
\begin{align}
\sum_{A}\text{OTO}^{(4)}(A,B) = 0, \qquad B\not=I. 
\end{align}
If $A=I$, we have $\text{OTO}(A,B)=1$. Since there are $d^2$ Pauli operators, we have 
\begin{align}
\text{OTO}^{(4)}(A,B) = -\frac{1}{d^2-1}, \qquad A,B\not=I. 
\end{align}
In Appendix~\ref{sec:k-point-averages}, we use this method to estimate the scaling of higher-point OTO correlators with $d$, finding $4m$-point functions of a related ordering scale as $\sim 1/d^{2m}$.

\subsubsection*{$6$-point functions}

Next, consider the Haar average of $6$-point OTO correlators, $\langle A \tilde{B} C \tilde{D} E \tilde{F} \rangle_{\text{Haar}}$. We will assume that $A,\ldots,F\not=I$ are Pauli operators. In order to have non-zero contributions, we must have $ACE \propto I$ and $BDF \propto I$. Thus, we will only consider cases with $ACE \propto I$ and $BDF \propto I$.

The results depend on commutation relations between $A,C$ and $B,D$, but always have the following scaling
\begin{align}
\langle \text{6-point} \rangle_{\text{Haar}} \sim \frac{1}{d^2}, \qquad (A C E=I, \quad B  D  F=I).
\end{align}
Explicitly, when $[A,C]=0$ and $[B,D]=0$, we find
\begin{align}
\langle \text{6-point} \rangle_{\text{Haar}}  = \frac{2d^2}{(d^2-1)(d^2-4)}.
\end{align}
The more general expression is slightly complicated though has the same scaling.
Thus, the Haar average of $6$-point OTO correlators does not reach any lower a floor value than the Haar average of $4$-point OTO correlators. 

\subsubsection*{$8$-point functions}
Finally, we will study Haar averages of $8$-point OTO correlators. In this case, there are two different types of nontrivial out-of-time ordering, which behave differently at large $d$. These computations are annoyingly technical, and so the details are hidden in Appendix~\ref{sec:8-pt-functions}.

The $8$-point OTO correlators of the first type can be written in the following manner
\begin{align}
\langle A\tilde{B}C\tilde{D}A^{\dagger}  \tilde{B}^{\dagger} C^{\dagger}\tilde{D}^{\dagger}\rangle, \qquad \text{(non-commutator type)}.\label{eq:non-commutator-type}
\end{align}
For Hermitian operators, this  essentially repeats $A\tilde{B}C\tilde{D}$ twice. For reasons that will subsequently become clear, we will call such OTO correlators ``non-commutator types.'' (However, the result does depends on the commutation relations between $A,C$ and $B,D$.) For these correlators, the scaling of the Haar average with respect to $d$ is
\begin{align}
\langle A\tilde{B}C\tilde{D}A^{\dagger}  \tilde{B}^{\dagger} C^{\dagger}\tilde{D}^{\dagger}\rangle_{\text{Haar}} \sim \frac{1}{d^2}, \qquad \text{(non-commutator type)},
\end{align}
and this scaling does not depend on any commutation relations. Similar to what we found for the Haar average of the $6$-point functions, these non-commutator type $8$-point OTO correlators have the same scaling with $d$ as the Haar average of $4$-point OTO correlators. 

The $8$-point OTO correlators of the second type take the form
\begin{align}
\langle A\tilde{B}C\tilde{D}A^{\dagger}  \tilde{D}^{\dagger} C^{\dagger}\tilde{B}^{\dagger}\rangle, \qquad \text{(commutator type)},
\end{align}
and are denoted ``commutator-type'' correlators.
These correlators have the property that they can be written in the form
\begin{align}
\langle A\tilde{B}C\tilde{D}A^{\dagger}  \tilde{D}^{\dagger} C^{\dagger}\tilde{B}^{\dagger}\rangle = \langle AKA^{\dagger} K^{\dagger}\rangle,
\qquad K = \tilde{B}C\tilde{D},
\end{align}
i.e. they are the expectation value of the group commutator of the operator $AKA^{\dagger} K^{\dagger}$.
The OTO correlators Eq.~\eqref{eq:non-commutator-type} cannot be written in this way. As with the non-commutator types, the exact Haar average depends on commutation relations between $A,C$ and $B,D$. However, the scaling with respect to $d$ does not
\begin{align}
\langle A\tilde{B}C\tilde{D}A^{\dagger}  \tilde{D}^{\dagger} C^{\dagger}\tilde{B}^{\dagger}\rangle_{\text{Haar}} \sim \frac{1}{d^4}, \qquad \text{(commutator type)}.
\end{align}
The Haar average of these correlators is much smaller than the Haar average of the non-commutator types and the $4$- and $8$-point Haar averages! This suggests they might be a useful statistic for distinguishing ensembles that form a $4$-design and ensembles that form a $2$-design but do not form a higher design.

To test this idea, we can take an average of the commutator type $8$-point OTO correlators averaged over the Clifford group.
Since we have assumed the operators $A, \dots, D$ are Pauli operators, we find
\begin{align}
\langle A\tilde{B}C\tilde{D}A^{\dagger}  \tilde{D}^{\dagger} C^{\dagger}\tilde{B}^{\dagger}\rangle_{\text{Clifford}} =  \langle A\tilde{B}\tilde{D}CA^{\dagger} C^{\dagger} \tilde{D}^{\dagger} \tilde{B}^{\dagger} \rangle_{\text{Clifford}}  = \frac{K(A,C^{\dagger})}{d}  \, \tr \{ A\tilde{B}\tilde{D}A^{\dagger} \tilde{D}^{\dagger} \tilde{B}^{\dagger}\},
\end{align}
where in the first equality we commuted $C,\tilde{D}$ and $C^{\dagger},\tilde{D}^{\dagger}$, which holds because $C, D$ are Pauli operators, and in the second equality we
have defined $K(P,Q)$ by
\be
Q^{\dagger}PQ = K(P,Q)P,
\ee 
for Pauli operators $P,Q$.  The final answer is
\begin{align}
\langle A\tilde{B}C\tilde{D}A^{\dagger}  \tilde{D}^{\dagger} C^{\dagger}\tilde{B}^{\dagger}\rangle_{\text{Clifford}} = \frac{-K(A,C^{\dagger})}{d^2-1}\sim \frac{1}{d^2}.
\end{align}
Recall that the Clifford group is a unitary $2$-design, is not a $3$-design in general (except for a system of qubits) \cite{Webb15}, and is never a $4$-design. Therefore, we see that commutator-type correlators may provide a statistical test of whether an ensemble forms a $k$-design but not a $k+1$-design. We explore this idea further in \S\ref{sec:k-point-averages}.

\subsection{Dissipation vs. scrambling}
In this subsection, we will return to the $2$- and $4$-point averages and compare them against expectations from time evolution. Furthermore, we will attempt to provide some physical intuition for the behavior of these averages over different ensembles. This will support our picture of chaotic time evolution leading to increased pseudorandomness.

For strongly coupled thermal systems, it is expected that the connected part of the $2$-point correlation functions decays exponentially within a time scale $t_{d}$ of order the inverse temperature $\beta$
\begin{align}
\langle A(0)B(t) \rangle \rightarrow \langle A\rangle \langle B \rangle + O(e^{-t/t_d}). \label{2-pt-decay}
\end{align}
This time scale is often referred to as a ``dissipation'' or ``thermalization'' time and is related to the time it takes local thermodynamic quantities reaching equilibrium.\footnote{For weakly coupled systems where the quasiparticle picture is valid (e.g. Fermi-liquids) the time scale is instead be given by $t_{d}\sim \beta^2$.} 
It is suggestive that the results Eq.~\eqref{2-pt-haar} and Eq.~\eqref{2-pt-decay} are so similar. After a short time $t_d$, for these $2$-point functions the chaotic dynamics give the same results as the Haar random dynamics.

Next, we turn to the variance of the $2$-point correlator $\langle A(0)B(t) \rangle$. For a closed system of finite number of degrees of freedom, the $2$-point function will be quasi-periodic with recurrences after a timescale $t_r \sim e^{d}$ that is exponential in the dimension and doubly exponential in the number of degrees of freedom $d=2^n$.  As such, the long-time average of $|\langle A(0)B(t) \rangle|^2$ must be nonzero. This can be estimated by performing a time average and gives a well known result \cite{Dyson:2002pf,Barbon:2003aq,Barbon:2014rma}\footnote{N.B. there is an error in this calculation as presented in \cite{Dyson:2002pf,Barbon:2003aq} and so interested readers should consult \cite{Barbon:2014rma} for the actual details.}
\begin{align}
\lim_{T\rightarrow \infty}\frac{1}{T}\int_{0}^{T} |\langle A(0)B(t)|^2 dt \sim \frac{1}{d^2}.
\end{align}
Comparing against our result for the Haar-averaged dynamics Eq.~\eqref{eq:2-pt-variance}, we see that they coincide.

Next, let's consider  $4$-point correlators in strongly-coupled theories with a large number of degrees of freedom $N$.\footnote{We thank Douglas Stanford for conversations relating to dissipative behavior of $4$-point functions.} (For example, this can be thought of a system of $N$ qubits where all the qubits interact but the interactions are at most $q$-local, with $q\ll N$ and $N\to \infty$.) First, let's consider the case of a \emph{time-ordered} correlator. Similar to the case of the $2$-point functions, two of the three Wick contractions are expected to decay exponentially within a dissipation time $t_{d}$
\begin{equation}
\begin{split}
\langle A(0)C(0)B(t)D(t) \rangle \rightarrow 
\langle AC \rangle \langle B D \rangle + O(e^{-t/t_d}),
\end{split}
\end{equation}
which for qubits is analogous to considering a $2$-point function between the composite operators $AC$ and $B D$.\footnote{Note that the exact timescale $t_{d}$ might depend on the operators being correlated and the particular contraction.} Thus, this correlator will equilibrate after a time $t_{d}$ with a late-time value that depends on the expectations $\langle AC \rangle$ and $\langle B D \rangle$.

Now, let's consider the out-of-time-order $4$-point correlator in a large $N$ strongly interacting theory. For $t \sim t_d$, this will behave similarly to the time-ordered correlator with two of the three Wick contractions decaying exponentially
\begin{equation}
\begin{split}
\langle A(0)B(t)C(0)D(t) \rangle \rightarrow 
\langle AC \rangle \langle B D \rangle + O(e^{-t/t_d}), \qquad t< t_d.
\end{split}
\end{equation}
However, for $t > t_d$ the correlator obtains a exponentially growing connected component
\be
\langle A(0)B(t)C(0)D(t) \rangle \rightarrow 
\langle AC \rangle \langle B D \rangle - O(e^{\lambda (t-t_*)}), \qquad t_d < t < t_*.
\ee
This growth occurs in the regime $t_d < t < t_*$. The time scale $t_* = \lambda^{-1} \log N$, known as the fast scrambling time, is the time at which the exponentially growing piece of the correlator compensates its $1/N$ suppression and  becomes $O(1)$. The coefficient $\lambda$ has the interpretation of a new kind of Lyapunov exponent \cite{Kitaev:2014t1} and is bounded from above by $2\pi / \beta$ \cite{Maldacena:2015waa}. Finally, for $t > t_*$, these OTO $4$-point correlators are expected to decay to a small floor value that is  exponentially small in $N$. A natural guess for this floor is
\begin{equation}
\langle A(0)B(t)C(0)D(t) \rangle \rightarrow  e^{-O(N)}\langle AC\rangle \langle BD\rangle, \qquad t > t_*,
\end{equation}
which is reproduced from Eq.~\eqref{eq:OTO_formula} with $1$-point functions assumed to be subtracted off.

As we mentioned, the $2$-point function Eq.~\eqref{2-pt-decay} reached its Haar random value after a short dissipation time $t_d$. This is very suggestive of a picture where chaotic dynamics behave as a pseudo-$1$-design after a time $t_d$. Taking this point further, let's consider the 4-point OTO correlator averaged over Pauli operators, an ensemble that forms a $1$-design, but not a $2$-design. Furthermore, we will assume that the operators $A$ and $B$ have zero overlap. Under this assumption, we can show that 
\begin{equation}
\begin{split}
\langle  A \tilde{B} C \tilde{D} \rangle_{\text{Pauli}} = 
\langle AC \rangle \langle B D \rangle. \label{eq:some_result}
\end{split}
\end{equation}
The proof of this is relegated to \S\ref{sec:proof:pauli}.
Apparently, Pauli operators capture the behavior of the dynamics around $t\sim t_d$, i.e. from after the dissipative regime until the scrambling regime, but then a $2$-design is required to capture the behavior after $t\sim t_*$, i.e. the post-scrambling regime.\footnote{Note that these observations depend on the ensemble we average over actually being the Pauli operators, and not just any ensemble that forms a $1$-design without forming a $2$-design. Furthermore, we assume that $A,B,C,D$ are simply few-body operators so that the correlator is of local operators. These choices are determined for us by the basis in which the Hamiltonian is $q$-local.}

Thus, we might say that after a time $\sim t_d$, the system becomes a pseudo-$1$-design, and then after $\sim t_*$ the system becomes a pseudo-$2$-design. (See Fig.~\ref{fig_scrambling} for a cartoon of this behavior.) However, it remains an open question whether there are any additional meaningful timescales that can be probed with correlators after $t_*$, though we are hopeful that such timescales might be hiding in higher-point OTO correlators.

\begin{figure}[htb!]
\centering
\includegraphics[width=0.70\linewidth]{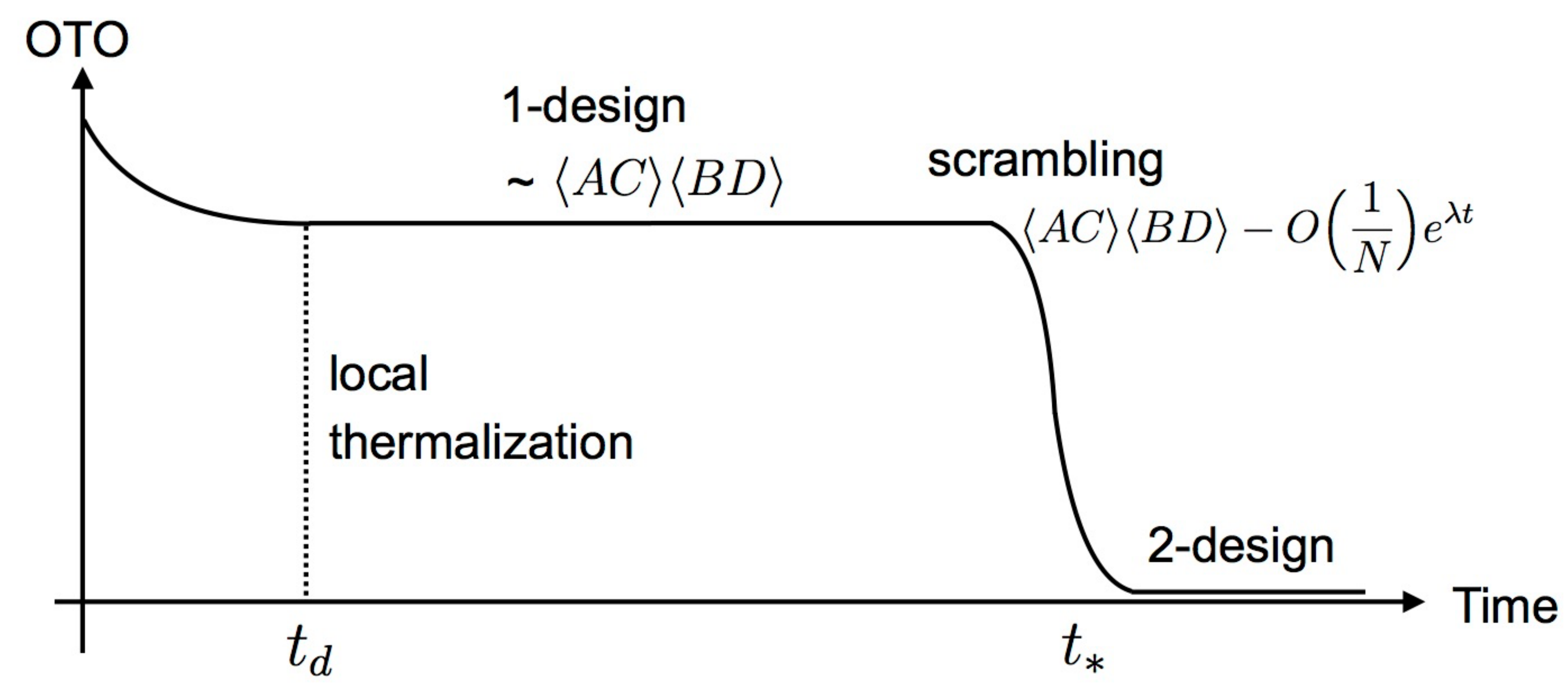}
\caption{Thermalization and scrambling in the decay of a four-point OTO correlator. These correlators typically decay to $\langle AC \rangle \langle B D \rangle$ in a thermal dissipation time $t_d$, and then they decay to a floor value $\sim d^{-2}$ at around the scrambling time $t_*$. These regimes are very well captured by replacing the dynamics with $1$-designs or $2$-designs, respectively. 
} 
\label{fig_scrambling}
\end{figure}

\section{Discussion}\label{sec:discussion}
In this paper, we have connected the related ideas of chaos and complexity to pseudorandomness and unitary design. A cartoon of these ideas is expressed nicely by Fig.~\ref{fig-universe}. Operators can be thought of as being organized by increasing complexity. Regions defined by circles of larger and larger radius can be thought of as defining designs with increasing $k$.\footnote{
    While this picture is a cartoon, the manifold for the unitary group $U(n)$ has a dimension exponentially large in $n$. However, following \cite{Dowling2006:gt}, one can find a metric in which the length of a minimal geodesic between points computes operator complexity and in which most sections are expected to be hyperbolic. Taking this further, \cite{Brown:2016wib} considered an interesting analog system on a hyperbolic plane that captures many of the expected properties of complexity.
} In the rest of the discussion, we will make some related points, tie up loose ends, and mention future work.

\begin{figure}[htb!]
\centering
\includegraphics[scale=.45]{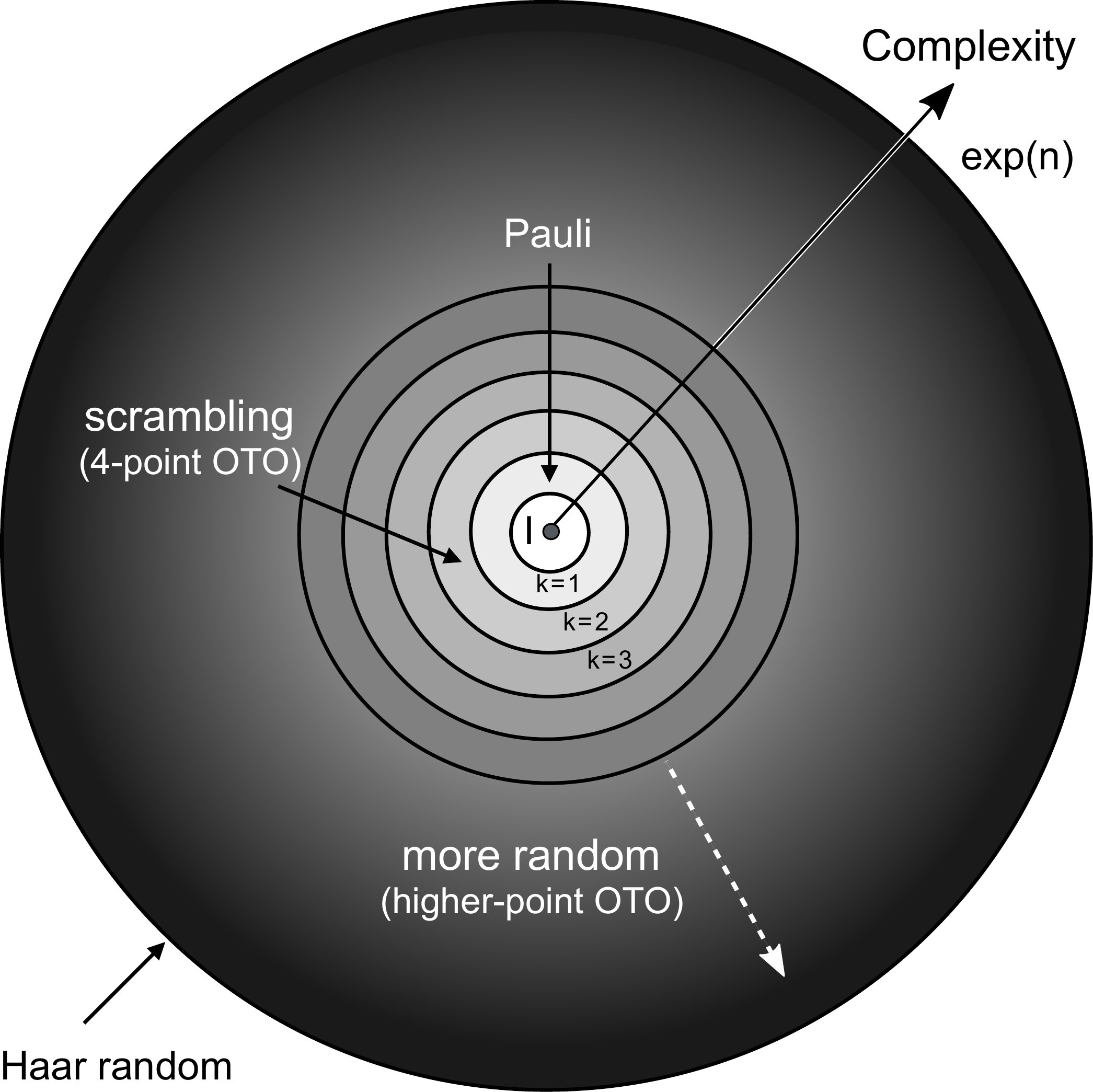}
\caption{A cartoon of the unitary group, with operators arranged by design. We pick the identity operator to be the reference operator of zero complexity and place it at the center. Typical operators have exponential complexity and live near the edge. Operators closer to the center have lower complexity, which makes them both atypical and more physically realizable in a particular computational model.
}

\label{fig-universe}
\end{figure}

\subsection*{Generalized frame potentials and designs}

In realistic physical systems, one usually does not have access to the full Hilbert space. For example, there may be some conserved quantities, such as energy or particle numbers, or the system may be at some finite temperature $\beta$. In that case, one would be interested in understanding pseudorandomness inside a subspace of the Hilbert space, i.e restricted to some state $\rho$. In Appendix~\ref{sec:sub-space-randomization}, we generalize the frame potential for an arbitrary state $\rho$ finding that the quantity
\begin{align}
\mathcal{F}^{(k)}_{\mathcal{E}}(\rho)&= \iint dU dV \ [\tr \{ \rho^{1/k} UV^{\dagger}\} \, \tr \{ \rho^{1/k} VU^{\dagger} \}]^{k}, \label{eq:discussion:generalized}
\end{align}
has all the useful properties desired of a frame potential. In particular, it is minimized by the Haar  ensemble, and it provides a lower bound on ensemble size and complexity.

However, if the state $\rho$ is the thermal density matrix, $\propto e^{-\beta H}$ and the ensemble is given by time evolution with an ensemble of Hamiltonians $\mathcal{E}= \{ e^{-iHt}\}$, then we need to take into account the fact that the state itself depends on the ensemble. Instead, we can define a \emph{thermal} frame potential
\begin{align}
\mathcal{W}^{(k)}_{\beta}(t)=  \iint dG dH \ \frac{ 
\big|\tr \, \{e^{-(\beta/2k-it )G } e^{-(\beta/2k+it )H }\}\big|^{2k}
}{\tr \, \{e^{-\beta G} \}\tr \, \{e^{-\beta H} \}}.
\end{align}
 In this case, even at $t=0$ one can derive an interesting bound on the complexity of the ensemble. We hope to return to this in the future to analyze the \emph{complexity of formation}: the computational complexity of forming the thermal state $\rho_{\beta}$ from a suitable reference state.\footnote{This is also a question of interest in holography, see e.g. \cite{CofFormation}.}

Finally, it would be similarly interesting to consider a different generalization of unitary designs where, under some physical constraints, we can only access some limited degrees of freedom in the system. In this sense, one could think of the unitary ensemble (as opposed the state) as being generated by tracing over these additional degrees of freedom. These ``subsystem designs'' would then be ``purified'' by integrating back in the original degrees of freedom.\footnote{We thank Patrick Hayden and Michael Walter for initial discussions about this idea.}
 This interesting direction is a potential subject of future work.

\subsection*{More chaos in quantum channels}

In Appendix~\ref{sec:more-chaos}, we revisit some ideas from our previous work \cite{Hosur:2015ylk}, though these ideas are also relevant to the current work. In particular, in \S\ref{sec:more:HP} we reconsider the Hayden-Preskill notion of black hole scrambling \cite{Hayden07}. In this thought experiment, Alice throws a secret quantum state into a black hole. Assuming Bob knows the initial state and dynamics of the black hole,  we show how the question of whether Bob can reconstruct Alice's secret is related to the decay of an average of a certain set of OTO four-point correlators.

 In \S\ref{sec:more:op}, we provide an operational interpretation to taking averages over OTO correlators. We show that this is related to a quantum game of ``catch'' where Alice may ``spit-on'' or otherwise perturb the ball before throwing it to Bob. The average over four-point OTO correlators gives the probability that Alice did not modify the ball. We also show that an average over higher-point OTO correlators can be interpreted as an ``iterated'' game of catch (i.e. what normally people just call ``catch'') where both Alice and Bob have the opportunity to modify the ball each turn. In this case, the OTO correlator average is related to the joint probability that neither Alice nor Bob perturb the ball.

 Finally, in \S\ref{sec:more:renyi} we show that an average over a particular ordering of $2k$-point OTO correlators can be related to the $k$th R\'enyi entropy of the operator $U$ interpreted as a state. We find that
\be
-\log( \text{a certain average of $2k$-point OTO correlators}) \propto S_{ \text{subsystem}}^{(k)}(U),\label{eq:summary-of-renyi-formula}
\ee
where $S_{\text{subsystem} }^{(k)}$ is the R\'enyi $k$-entropy of a particular subsystem of the density matrix $\rho = \ket{U}\bra{U}$.

\subsection*{Volume of unitary operators}

The argument in \S\ref{sec:continuous} led to a bound on the ratio $ N_s / \overline{N_\epsilon}$, which can be interpreted as the volume of $\mathcal{E}$ in terms of $\epsilon$-balls. An interesting application of this bound might be to think about the volume of unitary operators in $U(d)$ that can be probed in a finite time scale $T$, i.e. the volume of operators with depth $\mathcal{D} \sim T$. (See \S\ref{sec:complexity:depth} for further discussion of a lower bound on circuit depth in terms of the frame potential.)

In fact, in certain situations (such as the Brownian circuit introduced in \cite{Lashkari13} or the random circuit model), it is not difficult to show that the volume of unitary operators with depth $T$ grows at least $\sim \exp(\text{const} \cdot n \, T)$ for some small $T$ and some constant independent of $n$ by computing the $k=1$ frame potential. This implies that  the space of unitary operators, with the metric being quantum gate complexity, has hyperbolic structure with constant curvature, as discussed in e.g. \cite{Dowling2006:gt,Susskind:2014jwa,Brown:2016wib} (see also \cite{Chemissany:2016qqq}). On the other hand, one can upper bound the volume of unitary operators with circuit depth $T$ by thinking about how the depth can grow $V(T) \sim \binom{n}{2}\binom{n-2}{2}\binom{n-4}{2}  \cdots \binom{2}{2} \approx  \exp(n \log n \cdot T)$. Thus, for small $T$ and large $n$, the lower bound seems to be reasonably tight.

Once a lower bound on the volume of unitary operators in an ensemble is obtained in the unit of $\epsilon$-balls, we can also obtain a lower bound (of the same order) on the complexity of a typical operator in the ensemble. This seems possible by using the formal arguments given in \cite{Knill95} even when the elementary gate set is not discrete, e.g. all the two-qubit gates.

Finally, it's a curious fact that for systems with time-dependent Hamiltonian ensembles (such as the random circuit models or the Brownian circuit of \cite{Lashkari13}) that we get an initial linear growth of the volume with $T$. As argued in \S\ref{sec:complexity:early} (and confirmed numerically), for time independent Hamiltonian evolution---e.g. in SYK or in the Gaussian unitary ensemble (GUE)---we get a lower bound  $V(T) \sim \exp(\text{const} \cdot n\, T^2 )$, which persists for a short time $T \sim 1/\sqrt{n}$. It would be very interesting to understand the difference in this scaling.\footnote{
    One might worry that this bound saturates at a value smaller than unity. However, we expect that there may be a continuous definition of complexity sensible for small complexities, see e.g. \cite{Brown:2017jil}. 
}

\subsection*{Tightness of the complexity bound}

While the frame potential provides a rigorous lower bound on the complexity of generating an ensemble of unitary operators, there may be a cost: the  bound may not be very tight when applied to time evolution by an ensemble of Hamiltonians.\footnote{We have learned this by some numerical investigations of the frame potential.} 

Let us try to understand this better. 
To be concrete, consider the $k=2$ frame potential for a strongly coupled spin systems that scrambles in $t_* \sim \log n$ time. In such a system, for local operators $W, V$ of unit weight, OTO four-point correlators $\langle W(t)^{\dagger}V^{\dagger}W(t)V\rangle$ will begin to decay after $t\sim O(\log n)$. Since the $k=2$ frame potential is the average of four-point OTO correlators, one might expect that the frame potential will also start to decay at $t\sim O(\log n)$. 

However, this is not quite right. We expect the decay time for more general correlators of larger operators to be reduced to $\bar{t}_* \sim t_*  -  O\big(\log(\text{size}~W)\big) - O\big(\log(\text{size}~V)\big)$, where $t_*\sim O(\log n)$ is the scrambling when $V$ and $W$ are low-weight operators. If we randomly select Pauli operators $W$ and $V$, they will typically be nonlocal operators with $O(n)$ weights, and therefore the OTO decay time $\bar{t}_*$ will be reduced to $O(1)$ for $V$ and $W$ of typical sizes.

In fact, the above estimate suggests most of the correlators determining the complexity bound should begin to decay immediately. As we can see in Eq.~\eqref{eq:lower-bound-oto}, each correlator itself only makes a logarithmic contribution to the complexity and so we shouldn't expect the remaining slow decaying local correlators to be dominant. (To be sure, a further investigation of this point is required.)

One possible way to fix this problem would be to generalize the frame potential by using $p$-norm with $p \neq 2$ so that it is more sensitive to the slower decaying local correlators. We leave the study of such a generalization to the future.

\subsection*{Complexity and holography}
Finally, we will return to the question of complexity and holography discussed in the introduction.
In the context of holography, computational complexity was ``introduced'' \cite{Harlow:2013tf} as a possible resolution to the firewall paradox of \cite{Almheiri13,Almheiri13b}. A direct connection between complexity and black hole geometry was first proposed by Susskind \cite{Susskind:2014rva,Susskind:2014ira}, which culminated in proposals that the interior of the black hole geometry is holographically dual to the spatial volume \cite{Stanford:2014jda} or the spacetime action \cite{Brown:2015bva,Brown:2015lvg}. These proposals are motivated by the fact that the black hole interior continues to grow as the state evolves long past the time entropic quantities equilibrate \cite{Hartman13}. While there is nice qualitative evidence for both of these proposals \cite{Susskind:2014jwa,Roberts:2014isa,Brown:2016wib}, missing is a direct understanding of computational complexity in systems that evolve continuously in time with a time-independent Hamiltonian.

    A hint can be obtained by considering some of the motivations for these holographic complexity proposals. In particular, building on the work of \cite{Hartman13} and previous work of Swingle \cite{Swingle12,Swingle12b}, Maldacena suggested that the black hole interior could be found in the boundary  theory by a tensor network construction of the state \cite{Maldacena:2013t1}. A tensor network toy model of the evolution of the black hole interior was investigated in \cite{Hosur:2015ylk}. In this toy model, the interior of the black hole was modeled as a flat tiling of perfect tensors, see Fig.~\ref{tn-erb}. These tensors were elements of the Clifford group and acted as two-qubit unitary operators that highly entangled neighboring qubits at each time step. From the perspective of the boundary theory, this is a model for Hamiltonian time evolution.\footnote{In addition to providing a model for the black hole interior, Swingle's identification of the ground state of AdS with a tensor network \cite{Swingle12,Swingle12b} has led to numerous other quantum information insights in holography (e.g. quantum error correction and its relationship to bulk operator reconstruction \cite{Almheiri14}) as well as additional toy models that demonstrate those features (e.g. \cite{Pastawski15b,Hayden:2016cfa}).}

\begin{figure}[htb!]
\centering
\includegraphics[width=1\linewidth]{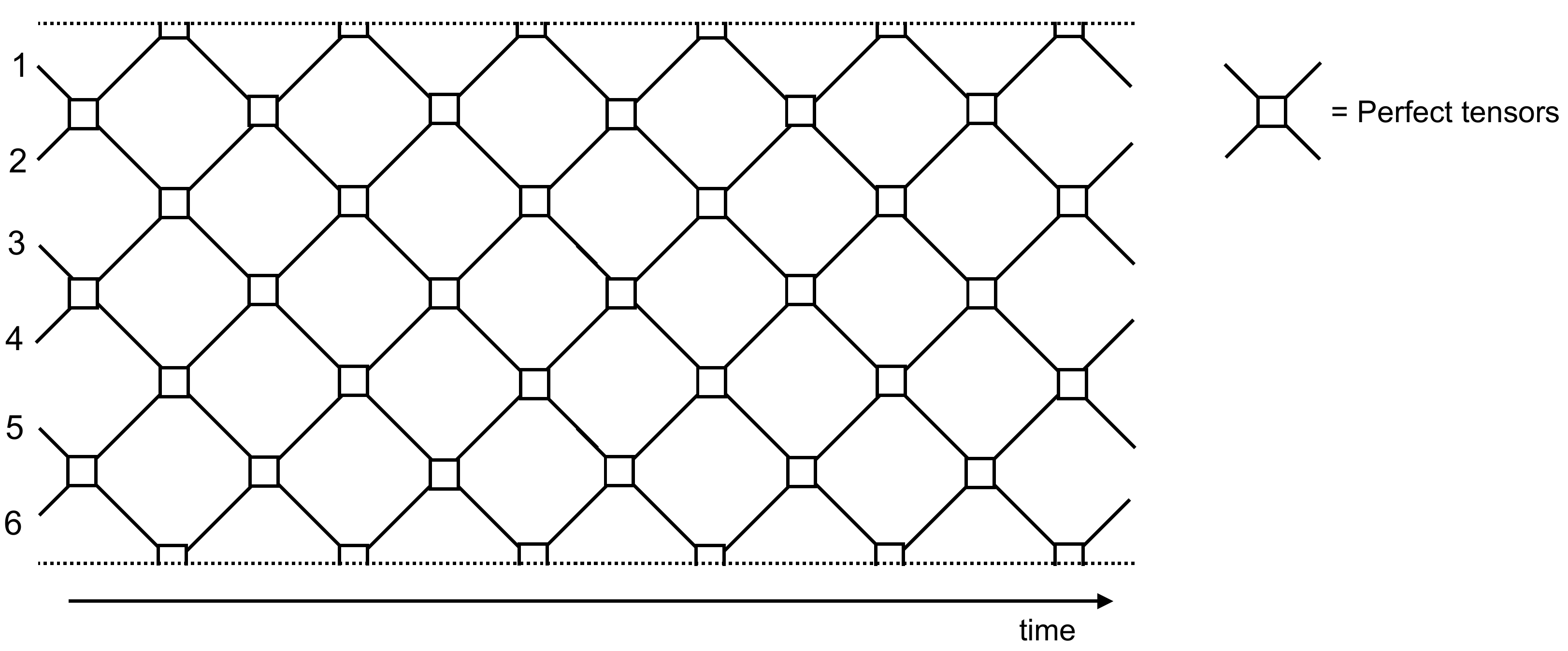}
\caption{A $6$-qubit tensor network model for the geometry of the interior of a black hole. Via holography, the growth of the interior is expected to correspond to chaotic time evolution of a strongly coupled quantum theory. Here, each node corresponds to a perfect tensor and the numbers label the qubit. 
} 
\label{tn-erb}
\end{figure}

    This toy model captures some important features of the complexity growth of the black hole state. The number of tensors in the network grows linearly in time, by construction. Operators will grow ballistically, exhibiting the butterfly effect, and the network scrambles in linear time. Thus, this network captures the aspects black hole chaos related to local scrambling and ballistic operator growth discussed in \cite{Roberts:2014isa} as well as aspects of complexity growth discussed in \cite{Hartman13,Susskind:2014ira,Stanford:2014jda,Brown:2015bva}.

    However, since in this model the perfect tensor is a repeated element of the Clifford group, the complexity can never actually grow to be very big. In fact, the quantum recurrence time of the model was investigated in \cite{Hosur:2015ylk} and was found to be exponential in the entropy $\sim e^{n}$ rather than doubly exponential $\sim e^{e^n}$ as expected in a fully chaotic model. This is related to our oft stated fact that the Clifford group generally does not form a higher-than-2-design. In fact, this model can actually be mapped to a classical problem, and 
    by the Gottesman-Knill theorem its complexity can be no greater than $O(n^2)$ gates \cite{Nielsen_Chuang}.\footnote{Since the Clifford group is a group, the complexity of any particular circuit in our model is an element of the group. Thus, at most we should be able to reach it by a polynomial number of applications of $2$-qubit gates.}

   These observations were the inspiration for this current work, since in this  toy model $4$-point OTO correlators behave chaotically, but higher-point OTO correlators do not. Nevertheless, this model can be ``improved'' by using random $2$-qubit tensors rather than a repeated perfect tensor.\footnote{Note: the case of a averaging over a single random tensor is like time evolution with a time independent Hamiltonian. On the other hand, the case of averaging separately over all tensors has a continuum limit of time evolution with a time-dependent Hamiltonian with couplings that evolve at each time step. This model is known as the Brownian circuit \cite{Lashkari13}.} In \cite{Brandao12}, it was shown that this local random quantum circuit approaches a unitary $k$-design in a circuit depth that scales at most as $O(k^{10})$. Our complexity lower bound for a $k$-design Eq.~\eqref{eq:complexity-k-design}  suggests that the time to become a $k$-design is lower bounded by $k$, and we suspect that this can be saturated.\footnote{It is believed that this is actually saturated by the local random circuit of \cite{Brandao12}. For example, there is some numerical evidence for this claim in \cite{mozrzymas2013local}.} 

It is in this sense that we speculate that the complexity growth of the chaotic black hole is pseudorandom. That is, we suspect that as the complexity of the black hole state increases linearly with time evolution $t$, the dynamics evolve to become pseudo-$k$-designs, with the value $k$ roughly scaling with $t$
   \be
   \mathcal{C}(e^{-iHt}) \sim t \sim k,
   \ee
 and that this may be quantified by either representative $2k$-point OTO correlators or by an appropriate generalization of unitary design. With this in mind, it would be interesting to see whether one could use the tools of unitary design to prove a version of the conjectures of \cite{Brown:2015bva,Brown:2015lvg} suggesting that complexity (is greater than or) equal to action.

\section*{Acknowledgments}

We are grateful to Fernando Brandao, Adam Brown, Jordan Cotler, Guy Gur-Ari, Patrick Hayden, Alexei Kitaev, M\'ark Mezei,  Xiao-Liang Qi, Steve Shenker, Lenny Susskind, Douglas Stanford, and Michael Walter for discussions. 

DR and BY both acknowledge support from the Simons Foundation through the ``It from Qubit'' collaboration.
DR is supported by the Fannie and John Hertz Foundation, the National Science Foundation grant
number PHY-1606531 and the Paul Dirac Fund, and is also very thankful for the hospitality of the Stanford Institute for Theoretical Physics and the Perimeter Institute of Theoretical Physics during the completion of this work. 
Research at Center for Theoretical Physics at MIT is supported by the U.S. Department of Energy under cooperative research agreement Contract Number DESC0012567.
Research at Perimeter Institute is supported by the Government of Canada through Industry Canada and by the Province of Ontario through the Ministry of Research and Innovation. This paper was brought to you unitarily by the Haar measure.\footnote{
    Finally, we would like to thank one of our anonymous JHEP referees for pointing out the prior work \cite{cbd:cd} and suggesting an acknowledgment. While we were unaware of \cite{cbd:cd} at the time of submission, we are nevertheless quite happy to include a citation in our revision.
}

\appendix

\section{More proofs}\label{app:proof}

Here, we collect some proofs.

\subsection{Proof of Theorem~\ref{theorem:OTO_chaos}}\label{sec:proof:oto-channnel}

\begin{proof}
The LHS of Eq.~(\ref{eq:tensor_M}) can be written explicitly as
\begin{align}
\sum_{C_{1},\cdots,C_{k}}\tr \, \{ C_{k}^{\dagger}A_{k}'^{\dagger}\ldots C_{1}^{\dagger}A_{1}'^{\dagger} \} \, \tr \, \{  A_{1}C_{1}\cdots A_{k}C_{k}\}.
\end{align}
Both $C_{k}^{\dagger}A_{k}'^{\dagger}\cdots C_{1}^{\dagger}A_{1}'^{\dagger}$ and $A_{1}C_{1}\cdots A_{k}C_{k}$ are Pauli operators with complex phases. This means that the traces will give a nonzero contribution only if both $A_{1}\cdots A_{k}C_{1}\cdots C_{k}\propto I$ and $A_{1}'\cdots A_{k}' C_{1}\cdots C_{k}\propto I$. Thus, the sum will vanish unless $A_{1}\cdots A_{k}\propto A_{1}'\cdots A_{k}'$. So, we consider the case where $A_{1}\ldots A_{k}\propto A_{1}'\ldots A_{k}' \propto P$ for some Pauli operator $P$. Then, the summation can be written as 
\begin{equation}
\begin{split}
\sum_{C_{1}\cdots C_{k}\propto P^{\dagger}}\tr \, \{C_{k}^{\dagger}A_{k}'^{\dagger}\cdots C_{1}^{\dagger}A_{1}'^{\dagger} \} \, \tr \{ A_{1}C_{1}\cdots A_{k}C_{k} \} \\
=
d\cdot\sum_{C_{1}\cdots C_{k}\propto P^{\dagger}}\tr \, \{ C_{k}^{\dagger}A_{k}'^{\dagger}\cdots C_{1}^{\dagger}A_{1}'^{\dagger}  A_{1}C_{1}\cdots A_{k}C_{k} \} .
\end{split}
\end{equation}
Here we expanded the trace because $C_{k}^{\dagger}A_{k}'^{\dagger}\cdots C_{1}^{\dagger}A_{1}'^{\dagger}$ and $A_{1}C_{1}\cdots A_{k}C_{k}$ are proportional to identity operators when $C_{1}\cdots C_{k}\propto P^{\dagger}$, and the $d$ comes from the trace of an identity operator. By using the cyclic property of the trace, we can eliminate $C_{k}^{\dagger}$ and $C_{k}$
\begin{align}
d\cdot\sum_{C_{1},\ldots,C_{k-1}} \tr\, \{ A_{k}'^{\dagger}\cdots C_{1}^{\dagger}A_{1}'^{\dagger}  A_{1}C_{1}\cdots A_{k} \}.
\end{align}
Here we used the fact that fixing $C_{1},\ldots,C_{k-1}$ uniquely determines $C_{k}$ when $C_{1}\cdots C_{k}\propto P$. Next, recall the relationship for summing over Pauli operators Eq.~\eqref{eq:Pauli_twirl}
\begin{align}
\sum_{C_{1}}(\cdots)C_{1}^{\dagger}A_{1}'^{\dagger}  A_{1}C_{1}(\cdots) = d^2 \cdot \delta^{A'_{1}}_{A_{1}}(\cdots)(\cdots)\label{eq:1-design}
\end{align}
where $(\cdots)$ represents arbitrary operators. This gives
\begin{align}
d^3\cdot \sum_{C_{2},\ldots,C_{k-1}}\tr \, \{ A_{k}'^{\dagger}\ldots C_{2}^{\dagger}A_{2}'^{\dagger}  A_{2}C_{2}\ldots A_{k} \} \delta^{A'_{1}}_{A_{1}}.
\end{align}
By repeated action of Eq.~\eqref{eq:1-design}, we can establish the desired result.
\end{proof}

\subsection{Proof of Eq.~(\ref{eq:some_result})}\label{sec:proof:pauli}

We can expand operators $A,B,C,D$ by using Pauli operators as a basis of operators
\begin{align}
A = \sum_{i}a_{i}Q_{i},\quad B = \sum_{i}b_{i}Q_{i},\quad C = \sum_{i}c_{i}Q_{i},\quad D = \sum_{i}d_{i}Q_{i},
\end{align}
where $Q_{i}$ are Pauli operators
\begin{equation}
\begin{split}
\langle A \tilde{B}C \tilde{D} \rangle_{\text{Pauli}} = 
\frac{1}{d^2} \sum_{P\in \text{Pauli}}\sum_{i,j,k,\ell} a_{i}b_{j}c_{k}d_{\ell} \langle Q_{i} P Q_{j}P^{\dagger} Q_{k} PQ_{\ell}P^{\dagger} \rangle.
\end{split}
\end{equation}
Notice that
\begin{align}
\langle Q_{i} P Q_{j}P^{\dagger} Q_{k} PQ _{\ell}P^{\dagger} \rangle
=  \frac{1}{d}\tr \big\{ K(P,Q_{j}Q_{\ell}) \big\} \cdot  \langle Q_{i} Q_{j} Q_{j} Q _{\ell}\rangle ,
\end{align}
and 
\begin{equation}
\begin{split}
\frac{1}{d^2}\sum_{P} \frac{1}{d}\tr \{ K(P,Q_{j}Q_{\ell}) \} &= 0, \qquad j\not=\ell, \\
&= 1, \qquad   j=\ell,
\end{split}
\end{equation}
where again we
have defined $K(P,Q)$ by $Q^{\dagger}PQ = K(P,Q)P$.
Thus, the desired average becomes
\begin{align}
\langle A \tilde{B}C \tilde{D} \rangle_{\text{Pauli}} = 
\sum_{i,j,k} a_{i}b_{j}c_{k}d_{j} \langle Q_{i} Q_{j} Q_{k} Q_{j} \rangle.
\end{align}
Observe that $\langle Q_{i} Q_{j} Q_{k} Q_{j} \rangle\not=0$ only when $Q_{i}=Q_{k}$, which gives
\begin{align}
\langle A \tilde{B}C \tilde{D} \rangle_{\text{Pauli}} = 
\sum_{i,j} a_{i}b_{j}c_{i}d_{j} \langle Q_{i} Q_{j} Q_{i} Q_{j} \rangle.
\end{align}
Next, we assume that $A$ and $B$ have no overlap. This implies that $[Q_{i},Q_{j}]=0$ if $a_{i},b_{j}\not=0$, and we obtain
\begin{align}
\langle A \tilde{B}C \tilde{D}            \rangle_{\text{Pauli}} = 
\sum_{i,j} a_{i}b_{j}c_{i}d_{j} .
\end{align}
Since $\langle A C \rangle = \sum_{i} a_{i}c_{i}$ and $\langle BD\rangle = \sum_{j} b_{i}d_{i}$, we see that
\be
\langle A \tilde{B}C \tilde{D}  \rangle_{\text{Pauli}} = \langle A C \rangle\langle BD\rangle,
\ee
as desired.

\section{More orthogonal states}\label{sec:appendix:orthogonal}

In the Hilbert space $\mathcal{H}=\mathbb{C}^d$, there are $d$ orthogonal states which form the basis of the Hilbert space
\begin{align}
\langle i |j \rangle = \delta_{ij}, \qquad (i,j=1,\ldots, d).
\end{align}
Yet, there are $\sim 2^d$ states which are nearly orthogonal to each other
\begin{align}
\langle \tilde{i} |\tilde{j} \rangle \ll 1, \qquad (i,j=1,\ldots, O(2^d)),
\end{align}
for $i\not=j$. 
A quantum state in this Hilbert space can be written as
\begin{align}
|\psi\rangle = a_{1}|1\rangle + a_{2}|2\rangle + \cdots + a_{d}|d\rangle,
\end{align}
where $|a_{1}|^2 + |a_{2}|^2 +\cdots + |a_{d}|^2=1$. Thus, a quantum state can be associated with a point on a $2d$-dimensional unit sphere. By computing the volume of states whose inner product with a given reference state is larger than $\epsilon$, we can find the number of nearly orthogonal states. Such a volume was explicitly computed in~\cite{Knill95}, but the point is that
\be
\frac{\text{Vol}(S_{2d})}{\text{Vol}(\epsilon\text{-ball})} \sim \bigg( \frac{1}{\epsilon} \bigg)^{d},
\ee
meaning that there are a doubly exponential number of nearly orthogonal states in terms of the number of degrees of freedom $n$, with $d=2^n$.

Here we provide another simple argument for this scaling.\footnote{Part of this presentation follows a nice talk given by Adam Brown \cite{Brownstates}.} Define $\{a\} = (a_{1},\ldots,a_{d})$ with $a_{j}=\pm1$. We consider the following state
\begin{align}
|\{a\}\rangle = \frac{1}{\sqrt{d}} \sum_{j=1}^d a_{j} |j\rangle, \qquad a_{j}=\pm 1.
\end{align}
There are $2^d=2^{2^n}$ states which can be represented in this way. Let us pick $\{a\}$ and $\{b\}$ randomly, and consider an inner product between two states $|\{a\}\rangle$ and $|\{b\}\rangle$
\begin{align}
\langle \{a\} |\{b\} \rangle = \frac{1}{d}\sum_{j=1}^{d}a_{j}b_{j}.
\end{align}
Since we chose $\{a\}$ and $\{b\}$ randomly, the product $a_{j}b_{j}=\pm 1$ is random, and the inner product will be close to zero 
\be
|\langle \{a\} |\{b\} \rangle| = O(d^{-1/2}),
\ee
implying that there are $O(2^d) = O(2^{2^n})$ nearly orthogonal states.

\section{More complexity bounds}\label{sec:complexity-appendix}
In this appendix we consider complexity lower bounds for when gates can be applied in parallel \S\ref{sec:complexity:depth} and for disordered Hamiltonian systems at early times \S\ref{sec:complexity:early}.

\subsection{Circuit depth}\label{sec:complexity:depth}

In Theorem~\ref{theorem-bound}, the complexity $\mathcal{C}(\mathcal{E})$ was defined as the number of steps required to generate the ensemble $\mathcal{E}$ when only a single two-qubit gate from the gate set can be applied at each step. Yet,  quantum gates may be applied in parallel if they do not act on the same qubits. By allowing simultaneous applications of quantum gates, we can instead consider the minimum quantum circuit depth $\mathcal{D}(\mathcal{E})$ to generate $\mathcal{E}$ and obtain a lower bound with a very similar argument. At each step, we can pair up the qubits in $\binom{n}{q}\binom{n-q}{q}\binom{n-2q}{q} \cdots \binom{q}{q} = n!/(q!)^{n/q}$ different ways (assuming $n$ is divisible by $q$ for simplicity). So, in a depth $1$ circuit there are roughly $g n!/(q!)^{n/q}$ choices.\footnote{We thank Fernando Brandao for this suggestion.} Thus, the lower bound on the depth is 
\begin{align}
\mathcal{D}(\mathcal{E}) \geq \frac{2kn \log(2) - \log F_{\mathcal{E}}^{(k)} }{\log(g) + n \log(n) - (n/q) \log(q!)}.
\end{align}
For $q \ll n$, the denominator is $\log g + n \log (n)$.
For large $n, q$, the denominator goes like $\log g + n \log (n/q)$. Essentially, the effect of large $q$ is only important if there's larger than an exponential number of gates $g \gtrsim 2^n$.

\subsection{Early times}\label{sec:complexity:early}
Here, we consider complexity lower bounds for disordered Hamiltonian systems at early times. This discussion is intended to be completely general, but for concreteness you may imagine that we are referring to the ensemble implied by time evolving with the Sachdev-Ye-Kitaev Hamiltonian \cite{Sachdev:1992fk,Kitaev:2014t2,Maldacena:2016hyu}
\be
H = (i)^{q/2} \sum_{1\le i_1 < \dots < i_q \le N} j_{i_1 \dots i_q} \psi_{i_1} \dots \psi_{i_q}, \qquad \overline{j^2_{i_1 \dots i_q}} = \frac{\mathcal{J}^2}{q} \frac{2^{q-1}}{\binom{N-1}{q-1}}  ,\label{syk-hamiltonian}
\ee
where the overline denotes an ensemble average. 
The model consists of $N$ Majorana fermions $\psi_{i}$, and each term in the Hamiltonian involves $q$ of the fermions interacting with a random coupling independently picked from a Gaussian with variance  $\overline{j^2_{i_1 \dots i_q}}$. 

We want to understand the initial growth of complexity of an ensemble of time evolution operators $\mathcal{E}(t) = \{e^{-iHt}\}$, where $H$ is disordered e.g. defined by Eq.~\eqref{syk-hamiltonian}. 
In general, we assume that $ \tr \, \{ \overline{ H} \}  =0$, and so the first nontrivial moment is the second moment. Start by expanding $U_i(t) = e^{-iH_it}$ and $V_j(t) = e^{-iG_jt}$ for early times
\begin{align}
U_i(t)= 1-iH_it - \frac{1}{2}H_i^2 t^2+ \dots, \qquad V_j(t)=1-iH_jt - \frac{1}{2}H_j^2 t^2 + \dots.
\end{align}
Plugging in to the definition of the frame potential, we get
\be
\Fk(t) = \sum_{i,j} p_i \, p_j \, d^{2k} \Big| 1 + \frac{t^2}{d} \Big( \tr \, \{ H_i H_j\} - \frac{1}{2} \tr \, \{ H^2_i\} - \frac{1}{2} \tr \, \{ H^2_j\}  \Big) +\dots \Big|^{2k}.
\ee
Next, we use the fact that $\tr \, \{ \overline{ H_i} \overline{ H_i} \}  = 0$ by independence, and expand assuming $t^2 \tr \, \{ H^2 \} \ll d$ to get an expression for the frame potential
\be
\Fk(t) = d^{2k}\Big( 1 - \frac{2k t^2}{d} \tr \, \{ \overline{H^2} \} + \dots \Big).
\ee 
This implies a lower bound on the initial growth of complexity of~\footnote{We thank Jordan Cotler for discussions.}
\be
\mathcal{C}(t) > \frac{2k t^2}{d} \tr \, \{ \overline{H^2} \}, 
\ee
valid for early times such that $\tr \, \{ \overline{H^2} \} \, t^2 / d \ll 1$. For the SYK model Eq.~\eqref{syk-hamiltonian}, we have $\tr \, \{ \overline{H^2} \} =  \mathcal{J}^2 d \,  (N/2q^2)$, and $\mathcal{C}(t) > k \,  (\mathcal{J}t)^2 (N/q^2) $ for times $t \ll \sqrt{2 q^2/ \mathcal{J}^2 N}$. 

Thus, while complexity is expected to eventually grow linearly in time \cite{Knill95,Susskind:2015toa}, for early times our bound predicts a quadratic phase of growth.

\section{More Haar random averages}\label{sec:appendix-averages}

In this appendix, we present the details of the Haar random averages of $8$-point OTO correlators as well as an argument for the scaling with $d$ of higher-point functions.

\subsection{8-point functions}\label{sec:8-pt-functions}

A general formula to compute the Haar average of $2k$-point OTO correlators is
\begin{align}
\langle A_{1}\tilde{B_{1}}\ldots A_{k} \tilde{B_{k}} \rangle_{\text{Haar}}
= \frac{1}{d}\sum_{\pi,\sigma}\text{Wg}(\pi\sigma) \, \tr \big\{ W_{\pi}W_{\text{cyc}}(A_{1}\otimes \ldots \otimes A_{k}) \big\} \, \tr \big\{ W_{\sigma}(B_{1}\otimes \ldots \otimes B_{k})\big\},
\end{align}
where $\text{Wg}(\pi)$ is the Weingarten function. Because of the Weingarten function, this is difficult to evaluate.

The unitary Weingarten function on $S_{n}$ is given by~\cite{Collins03}
\begin{align}
\text{Wg} = \frac{1}{n!} \sum_{\lambda } \frac{f^{\lambda}}{s_{\lambda}} \chi^{\lambda},
\end{align} 
where summation is over all partitions $\lambda$ of $n$. $s_{\lambda}$ is a polynomial in $d$ given by
\begin{align}
s_{\lambda} = \prod_{i=1}^{l(\lambda)} \prod_{j=1}^{\lambda_{i}}(d+j-i),
\end{align}
where $l(\lambda)$ is the length of $\lambda$. $f^{\lambda}$ is the dimension of the irreducible representation associated with $\lambda$. $\chi^{\lambda}$ is the irreducible characters.

For $S_{2}$, the character table is given by
\begin{center}
\begin{tabular}{ccc}
&$(1,1)$ & $(2)$  \\
\hline 
$(2)$     & 1 & -1  \\
$(1,1)$   & 1 & 1 
\end{tabular}.
\end{center}
We have 
\begin{equation}
\begin{split}
s_{(1,1)} = d(d-1), \qquad
s_{(2)} = d(d+1).
\end{split}
\end{equation}
Thus, 
\begin{equation}
\begin{split}
\text{Wg}(1,1)=&\frac{1}{2} \left( \frac{1}{d(d-1)} + \frac{1}{d(d+1)} \right) = \frac{1}{d^2-1},\\
\text{Wg}(2)=&\frac{1}{2} \left( \frac{-1}{d(d-1)} + \frac{1}{d(d+1)} \right) =\frac{-1}{d(d^2-1)}.
\end{split}
\end{equation}
For $S_{3}$, the character table is given by
\begin{center}
\begin{tabular}{cccc}
&$(1,1,1)$ & $(2,1)$ & $(3)$ \\
\hline 
$(3)$     & 1 & 1 & 1 \\
$(2,1)$   & 2 & 0 & -1\\
$(1,1,1)$ & 1 & -1 & 1\\
\end{tabular}.
\end{center}
We have 
\begin{equation}
\begin{split}
s_{(1,1,1)} = d(d-1)(d-2), \\
s_{(2,1)} = d(d+1)(d-1),\\
s_{(3)} = d(d+1)(d+2).
\end{split}
\end{equation}
Thus
\begin{equation}
\begin{split}
\text{Wg}(1,1,1)=&\frac{1}{6} \left( \frac{1}{d(d-1)(d-2)} + \frac{4}{d(d+1)(d-1)} + \frac{1}{d(d+1)(d+2)} \right), \\
&=\frac{d^{2}-2}{d(d^{2}-2)(d^{2}-4)},\\
\text{Wg}(2,1)=&\frac{1}{6} \left( \frac{-1}{d(d-1)(d-2)} + \frac{1}{d(d+1)(d+2)} \right), \\
&=\frac{-1}{(d^{2}-1)(d^{2}-4)},\\
\text{Wg}(3)=&\frac{1}{6} \left( \frac{1}{d(d-1)(d-2)} + \frac{-2}{d(d+1)(d-1)} + \frac{1}{d(d+1)(d+2)} \right), \\
&=\frac{2}{d(d^{2}-1)(d^{2}-4).}
\end{split}
\end{equation}
For $S_{4}$, the character table is given by 
\begin{center}
\begin{tabular}{cccccc}
&$(1,1,1,1)$ & $(2,1,1)$ & $(2,2)$& $(3,1)$ & $(4)$  \\
\hline 
$(4)$     & 1 & 1 & 1 & 1 & 1 \\
$(3,1)$   & 3 & 1 & -1  & 0 & -1\\
$(2,2)$   & 2 & 0 & 2  & -1 & 0 \\
$(2,1,1)$     & 3 & -1 & -1  & 0 & 1 \\
$(1,1,1,1)$   & 1 & -1 & 1  & 1 & -1 
\end{tabular}.
\end{center}
We have 
\begin{equation}
\begin{split}
s_{(4)}       &= d(d+1)(d+2)(d+3),\\
s_{(3,1)}     &= d(d+1)(d+2)(d-1),\\
s_{(2,2)}     &= d(d+1)(d-1)d,\\
s_{(2,1,1)}   &= d(d+1)(d-1)(d-2),\\
s_{(1,1,1,1)} &= d(d-1)(d-2)(d-3).
\end{split}
\end{equation}
Thus
\begin{equation}
\begin{split}
\text{Wg}(1,1,1,1) &= \frac{d^{4} - 8d^{2} + 6}{d^{2}(d^{2}-1)(d^{2}-4)(d^{2}-9)}, \\
\text{Wg}(2,1,1) &= \frac{-d^{3} + 4d}{d^{2}(d^{2}-1)(d^{2}-4)(d^{2}-9)}, \\
\text{Wg}(2,2) &= \frac{d^{2} + 6}{d^{2}(d^{2}-1)(d^{2}-4)(d^{2}-9)}, \\
\text{Wg}(3,1) &= \frac{2d^{2} - 3}{d^{2}(d^{2}-1)(d^{2}-4)(d^{2}-9)}, \\
\text{Wg}(4) &= \frac{-5d}{d^{2}(d^{2}-1)(d^{2}-4)(d^{2}-9)}.
\end{split}
\end{equation}
We also note here the large $d$ asymptotic behavior
\begin{align}
\text{Wg}(\lambda) \sim \frac{1}{d^{2n-\text{\# of cycles in $\lambda$}}}.
\end{align}

We now will evaluate the following $8$-point OTO correlator average of the commutator-type correlators
\begin{align}
|\text{$8$-point}|_{\text{Haar}}= \left|\frac{1}{d}\tr \{ A\tilde{B}C\tilde{D}C^{\dagger}\tilde{B}^{\dagger}A^{\dagger} \tilde{D}^{\dagger}\}\right|_{\text{Haar}}, \qquad \text{(commutator type)}
\end{align}
where $\tilde{B}=UBU^{\dagger}$ and $\tilde{D}=UDU^{\dagger}$. We assume that $A,B,C,D$ are Pauli operators, $A,B,C,D\not=I$ and $AC,BD\not=I$. We write the correlator as
\begin{align}
\frac{1}{d}\tr \big\{ (A\otimes C \otimes C^{\dagger} \otimes A^{\dagger}) \Phi_{4}(B\otimes D \otimes B^{\dagger} \otimes D^{\dagger}) W_{2341}\big\},
\end{align}
where $W_{2341}$ is a cyclic permutation. 
We have
\begin{equation}
\begin{split}
\frac{1}{d} \sum_{\pi_{1},\pi_{2}\in S_{4}} \text{Wg}(\pi_{1}\pi_{2}) \, \tr \big\{ (A\otimes C \otimes C^{\dagger} \otimes A^{\dagger})W_{\pi_{1}}W_{2341}\big\} \, \tr \big\{ (B \otimes D \otimes B^{\dagger} \otimes D^{\dagger})W_{\pi_{2}} \big\}.
\end{split}
\end{equation}
Let us assume that $[A,C]=[B,D]=0$. This gives
\begin{align}
\text{Wg}(1,1,1,1)d^{2} +  
\text{Wg}(2,1,1)(d^{3} + 16d) +
\text{Wg}(3,1)7d^2  + \text{Wg}(2,2)3d^2  +
\text{Wg}(4)(d^2 + 20d).
\end{align}
For large $d$, the coefficients $\text{Wg}(\pi)$ scales as follows
\begin{align}
\text{Wg}(\pi) \sim O\left(\frac{1}{d^{8-\text{\#of cycles in $\pi$}}}\right).
\end{align}
At first sight, the dominant contribution might seem to be $O\left(\frac{1}{d^{2}}\right)$
\begin{align}
(\text{Wg}(1,1,1,1)+\text{Wg}(2,1,1)d)d^{2},
\end{align}
but, due to nice cancellation, the above expression becomes
\begin{equation}
\begin{split}
&\left(\frac{d^{4} - 8d^{2} + 6}{d^{2}(d^{2}-1)(d^{2}-4)(d^{2}-9)}+\frac{-d^{3} + 4d}{d^{2}(d^{2}-1)(d^{2}-4)(d^{2}-9)}d\right)d^{2},\\
&= \frac{- 4d^{2} + 6}{(d^{2}-1)(d^{2}-4)(d^{2}-9)},
\end{split}
\end{equation}
which is $O\left(d^{-4}\right)$. Finally, we can write down the complete answer
\begin{align}
|\text{$8$-point}|_{\text{Haar}}=
\frac{- 3d^{2} - 5 d - 33 }{(d^{2}-1)(d^{2}-4)(d^{2}-9)}, \qquad \text{(commutator type)},
\end{align}
which scales as $O(d^{-4})$ for large $d$.
We have also checked other commutation relations between $A,B,C,D$, and although the exact form of the answer is different, the scaling is always $O(d^{-4})$.

Next, let us evaluate the Haar average of the non-commutator type $8$-point OTO correlators
\begin{align}
|\text{$8$-point}|_{\text{Haar}}= \left|\frac{1}{d}\tr \{ A\tilde{B}C\tilde{D}A^{\dagger}\tilde{B}^{\dagger}C^{\dagger} \tilde{D}^{\dagger} \} \right|_{\text{Haar}}, \qquad \text{(non-commutator type)},
\end{align}
which we rewrite as
\begin{equation}
\begin{split}
\frac{1}{d} \sum_{\pi_{1},\pi_{2}\in S_{4}} \text{Wg}(\pi_{1}\pi_{2}) \, \tr \big\{ (A\otimes C \otimes A^{\dagger} \otimes C^{\dagger})W_{\pi_{1}}W_{2341}\big\} \, \tr \big\{ (B \otimes D \otimes B^{\dagger} \otimes D^{\dagger})W_{\pi_{2}} \big\}.
\end{split}
\end{equation}
The difference from the commutator type is that the first trace contains the operator $A\otimes C \otimes A^{\dagger} \otimes C^{\dagger}$ instead of the operator $A\otimes C \otimes C^{\dagger} \otimes A^{\dagger}$. Let us assume the commutation relations $[A,C]=[B,D]=0$. Then we find
\begin{align}
\text{Wg}(1,1,1,1)2d^{2} +  
\text{Wg}(2,1,1)16d +
\text{Wg}(3,1)8d^2  + \text{Wg}(2,2)(2d^2)  +
\text{Wg}(4)(d^3 +20d) 
\end{align}
which scales as $O\left(d^{-2}\right)$. For other choices of commutation relations, it's easy to check that the scaling is always $O\left(d^{-2}\right)$.

\subsection{Higher-point functions}\label{sec:k-point-averages}

Finally, we will make a conjecture regarding the behaviors of higher-point OTO correlation functions. We speculate that $4m$-point OTO correlators of the commutator-type form
\begin{align}
\langle A_{1} (\tilde{B_{1}} \cdots A_{m} \tilde{B_{m}}) A_{1}^{\dagger} (\tilde{B_{1}} \cdots A_{m} \tilde{B_{m}})^{\dagger}  \rangle, \qquad \text{(commutator type)},
\end{align}
will have the following asymptotic scaling.

\begin{conjecture}\emph{
The $4m$-point OTO correlation functions asymptotically scale as 
\begin{align}
\langle A_{1} (\tilde{B_{1}} \cdots A_{m} \tilde{B_{m}}) A_{1}^{\dagger} (\tilde{B_{1}} \cdots A_{m} \tilde{B_{m}})^{\dagger}  \rangle_{\text{Haar}}
\sim \frac{1}{d^{2m}},
\end{align}
when averaged over Haar random unitary operators.
}
\end{conjecture}

For $m=1,2$, we recover the analytical results we have already obtained. Below, we provide a supporting argument for this conjecture.

First, let's consider $8$-point functions using the method from \S\ref{sec:haar-averages}. We are interested in computing the following quantity
\begin{align}
\text{OTO}^{(8)}(A_{1},B_{1},A_{2},B_{2}) =\langle A_{1} \tilde{B_{1}}
A_{2} \tilde{B_{2}} A_{2}^{\dagger} {\tilde{B_{1}}}^{\dagger}A_{1}^{\dagger} {\tilde{B_{2}}}^{\dagger}
 \rangle_{\text{Haar}}, \qquad B_{1},B_{2}\not=I, \quad A_{1}A_{2}\not= I.
\end{align}
From a simple calculation, we find
\begin{align}
\sum_{B_{1}}\text{OTO}^{(8)}(A_{1},B_{1},A_{2},B_{2})=0, \qquad B_{2}\not=I.
\end{align}
If $B_{1}=I$, we have 
\begin{align}
\text{OTO}^{(8)}(A_{1},I,A_{2},B_{2}) =\langle A_{1}A_{2} \tilde{B_{2}} A_{2}^{\dagger} A_{1}^{\dagger} {\tilde{B_{2}}}^{\dagger}
 \rangle_{\text{Haar}} 
= \text{OTO}^{(4)}(A_{1}A_{2},B_{2}) = -\frac{1}{d^2 -1},
\end{align}
since $A_{1}A_{2}\not=I$ and $B_{2}\not=I$. Thus, we have 
\begin{align}
|\text{OTO}^{(8)}(A_{1},B_{1},A_{2},B_{2})|_{B_{1},B_{2}\not=I, \ A_{1}A_{2}\not= I}
= \frac{1}{(d^2-1)^2}.
\end{align}
This method shows what we already know, that the Haar average of commutator-type $8$-point OTO correlators equals $\sim d^{-4}$. 

Now, let us consider the Haar average of commutator-type $4m$-point OTO correlation functions\footnote{The ordering here is inspired by the correlators considered in \cite{Shenker:2013yza} and is related to a $2$-point correlation function in an AdS black hole geometry perturbed by $(2m-1)$-shock waves.}
\begin{align}
\text{OTO}^{(4m)}(A_{1},B_{1},\ldots, A_{m},B_{m}) =\langle 
A_{1} \tilde{B_{1}} \cdots A_{m} \tilde{B_{m}} 
A_{m}^{\dagger} {\tilde{B_{m-1}}}^{\dagger}A_{m-1}^{\dagger}\cdots
{\tilde{B_{1}}}^{\dagger}A_{1}^{\dagger} {\tilde{B_{m}}}^{\dagger}
 \rangle_{\text{Haar}} .
\end{align}
 By a recursive argument, we can show the following
\begin{align}
\text{OTO}^{(4m)}(A_{1},B_{1},\ldots, A_{m},B_{m})|_{B_{1},B_{2},\ldots,B_{k}\not=I, \ A_{1}A_{2}\cdots A_{k}\not= I}
= \left(\frac{-1}{d^2-1}\right)^k. \label{eq:decay}
\end{align}
The  average of $\text{OTO}(A_{1},B_{1},\ldots, A_{m},B_{m})$ over $B_{1}$ is zero as long as $B_{m}\not=I$. If $B_{1}=I$, then we have
\begin{align}
\text{OTO}^{(4m)}(A_{1},I,\ldots, A_{m},B_{m})
= \text{OTO}^{(4m-4)}(A_{1}A_{2},B_{2},\ldots, A_{m},B_{m}).
\end{align}
Thus, we see
\begin{align}
&\text{OTO}^{(4m)}(A_{1},B_{1},\ldots, A_{m},B_{m})|_{B_{1},\ldots,B_{k}\not=I, \ A_{1}\cdots A_{k}\not= I}\\
&= \left(\frac{-1}{d^2-1}\right)\text{OTO}^{(4m-4)}(A_{1}A_{2},B_{2},\ldots, A_{m-1},B_{m-1})|_{B_{2},\ldots,B_{k}\not=I, \ A_{1}\cdots A_{k}\not= I}.
\end{align}

This argument suggests that the Haar average of $4m$-point OTO correlators have a scaling $\sim d^{-2m}$. One would need to evaluate the Weingarten function for $S_{2m}$ in order to check this exactly, and we will not attempt to do that. In particular, the Haar average of $\text{OTO}^{(4)}(A,B)$ does not depend on $A,B$ as long as $A,B\not = I$. For $8$-point or higher-point OTO correlators, $\text{OTO}^{(4m)}(A_{1},B_{1},\ldots, A_{m},B_{m})$ depends on the details of $A_{1},B_{1},\ldots, A_{m},B_{m}$, even when we impose $B_{1},\ldots,B_{k}\not=I$ and $A_{1}\cdots A_{k}\not= I$. Namely, this suggests that the exact result will depend on the commutation relations of the $A_{j}$ and $B_{j}$, and we do not know an exact form of dependence. Our argument can only tell us that Eq.~\eqref{eq:decay} has a value that scales as $O\left(d^{-2m}\right)$.

Finally, for these commutator-type higher-point OTO correlators, we can compute Clifford averages easily
\begin{align}
\langle A_{1} (\tilde{B_{1}} \cdots A_{m} \tilde{B_{m}}) A_{1}^{\dagger} (\tilde{B_{1}} \cdots A_{m} \tilde{B_{m}})^{\dagger}  \rangle_{\text{Clifford}},
\sim \frac{1}{d^{2}}, \qquad \text{(commutator type)},
\end{align}
by using the fact that $A_{j}$ and $\tilde{B_{j}}$ are Pauli operators. This supports our belief that higher-point correlators might be useful probes of whether an ensemble forms a $k$-design but does not form a higher-design.

\section{More (general) frame potentials}\label{sec:sub-space-randomization}

In this appendix, we generalize the notion of unitary design for states described by an arbitrary density matrix $\rho$. We will find two possible generalizations of the frame potential:
\begin{align}
\mathcal{F}^{(k)}_{\mathcal{E}}(\rho)&= \iint dU dV \ [\tr \{ \rho^{1/k} UV^{\dagger}\} \, \tr \{ \rho^{1/k} VU^{\dagger} \}]^{k}\label{eq:generalization1}, \\
\mathcal{G}_{\mathcal{E}}^{(k)}(\rho)&= \iint dU dV \ [\tr \{ \rho^{1/k} UV^{\dagger} \} \, \tr \{ \rho^{1/k} U^{\dagger}V \}]^{k}.\label{eq:generalization2}
\end{align}
The difference is the ordering of $V$ and $U^{\dagger}$ in the second trace. For the maximally mixed state $\rho= \frac{I}{d}$, both expressions are reduced to the original frame potential $F_{\mathcal{E}}^{(k)}$ up to factor of proportionality:
\begin{align}
\mathcal{F}_{\mathcal{E}}^{(k)}\Big(\frac{I}{d}\Big)=\mathcal{G}_{\mathcal{E}}^{(k)}\Big(\frac{I}{d}\Big)= \frac{1}{d^2}F^{(k)}_{\mathcal{E}}.
\end{align}
However, the two expressions behave quite differently when $\rho$ is not the maximally mixed state. Below, we explain how we arrived at these expressions and study their basic properties such as their Haar average values and the fact that they are minimized when the ensemble is a $k$-design. Our conclusion is that the first expression, $\mathcal{F}_{\mathcal{E}}^{(k)}$, seems to be a more appropriate generalization of the frame potential.

\subsection*{Properties}

The original $k$th frame potential, defined for $\rho\propto I$, is minimized if and only if the unitary ensemble $\mathcal{E}$ is a $k$-design. We expect any sensible generalization of the frame potential should have similar minimization property. If so, then e.g. if $\mathcal{F}_{\mathcal{E}}^{(k)}(\rho)=\mathcal{F}_{\text{Haar}}^{(k)}(\rho)$, we can think of the ensemble $\mathcal{E}$ as forming $k$-design with respect to the state $\rho$.  We prove the following lemma.

\begin{lemma}
For generalized frame potentials $\mathcal{F}_{\mathcal{E}}^{(k)}(\rho)$ and $\mathcal{G}_{\mathcal{E}}^{(k)}(\rho)$, we have
\begin{align}
\mathcal{F}^{(k)}_{\mathcal{E}}(\rho)\geq \mathcal{F}^{(k)}_{\text{Haar}}(\rho), \qquad \mathcal{G}^{(k)}_{\mathcal{E}}(\rho)\geq \mathcal{G}^{(k)}_{\text{Haar}}(\rho),
\end{align}
with equality if $\mathcal{E}$ is a $k$-design.
\end{lemma}

\begin{proof}
We begin with $\mathcal{G}^{(k)}$. Consider the following operator 
\begin{align}
S \equiv \int_{\mathcal{E}} dU \ (U)^{\otimes k} \otimes (U^{\dagger})^{\otimes k} - \int_{\text{Haar}} dU \ (U)^{\otimes k} \otimes (U^{\dagger})^{\otimes k}.
\end{align} 
Observe
\begin{align}
\tr \{ (\rho^{1/k}\otimes \rho^{1/k} \otimes \cdots) SS^{\dagger} \}  = \mathcal{G}^{(k)}_{\mathcal{E}}(\rho) - \mathcal{G}^{(k)}_{\text{Haar}}(\rho)\geq 0.
\end{align}

For $\mathcal{F}^{(k)}$, consider the following quantity
\begin{align}
\tr \big\{ \sigma_{L} S\sigma_{R}S^{\dagger} \big\} = \mathcal{F}^{(k)}_{\mathcal{E}}(\rho) - \mathcal{F}^{(k)}_{\text{Haar}}(\rho)\label{eq:frame_subspace}
\end{align}
where $\sigma_{L} =  (\rho^{1/k})^{\otimes k} \otimes (I)^{\otimes k}$ and $\sigma_{R}= (I)^{\otimes k}\otimes (\rho^{1/k})^{\otimes k}$. Define
\begin{align}
X = \sigma_{L}^{1/2} S \sigma_{R}^{1/2},
\qquad X^{\dagger} = \sigma_{R}^{1/2}  S^{\dagger} \sigma_{L}^{1/2}.
\end{align}
We then have 
\begin{align}
\tr \big\{ \sigma_{L} S\sigma_{R}S^{\dagger} \big\} = \tr \{XX^{\dagger} \}\geq 0.
\end{align}
\end{proof}

In \S\ref{sec:OTO_channel}, we showed that the $k$th frame potential is proportional to the average of certain $2k$-point OTO correlators, where the correlators are evaluated for maximally mixed states. Here, we compute an average of OTO correlators for general $\rho$ to derive a candidate expression of generalized frame potential. Inspired by~\cite{Maldacena:2015waa}, we consider regulated OTO correlators of the form
\begin{align}
\langle A_{1} \tilde{B}_{1}\cdots  A_{k}\tilde{B}_{k}\rangle_{\rho}=
\tr \, \{ \rho^{1/2k} A_{1}\rho^{1/2k}\tilde{B}_{1}\cdots \rho^{1/2k} A_{k}\rho^{1/2k}\tilde{B}_{k} \},\label{eq:regulated-def}
\end{align}
instead of unregulated correlators 
\be
\langle A_{1} \tilde{B}_{1}\cdots  A_{k}\tilde{B}_{k}\rangle_{\rho, \text{unregulated}} = \tr \,\{ \rho \, A_{1}\,\tilde{B}_{1}\cdots  A_{k}\,\tilde{B}_{k} \}.\label{eq:unregulated-def}
\ee 
When $\rho$ is the thermal density matrix, Eq.~\eqref{eq:regulated-def} uniformly distributes the operators around the thermal circle. Using Eq.~\eqref{eq:regulated-def}, it is not difficult to prove the following lemma.

\begin{lemma}
For regulated $2k$-point OTO correlators, we have
\begin{align}
\sum_{A_{1},B_{1},\ldots}|\langle A_{1} \tilde{B}_{1}\cdots  A_{k}\tilde{B}_{k}\rangle_{\rho}|^2 \propto \mathcal{F}^{(k)}(\rho). 
\end{align}\label{lemma-generalized-F}
\end{lemma}

This is one of our motivations for considering $\mathcal{F}^{(k)}(\rho)$ as the proper generalization of the frame potential.\footnote{From an operational viewpoint, it might be more sensible to define $\mathcal{F}^{(k)}_{\mathcal{E}}(\rho)$ with $\rho^{1/k} \to \rho$
\begin{align}
\mathcal{F}^{(k)}_{\mathcal{E}}(\rho)&= \iint dU dV \ [\tr \{ \rho UV^{\dagger}\} \, \tr \{ \rho VU^{\dagger} \}]^{k},
\end{align}
so that there is one $\rho$ per trace. However, this means that the correlators Eq.~\eqref{eq:regulated-def} will have many copies of~$\rho$. This is similar to what happens when using correlators to compute R\'enyi entropies.
}
Also, we note that due to the fact that the maximally mixed state is normalized with a factor of $1/d$, for these generalized frame potentials the scaling with respect to $d$ is a bit different from the original definition of the frame potential $F^{(k)}$. For a maximally mixed state $(\rho=I/d)$, we have
\begin{align}
&\text{ \bf generalized} &&\text{ \bf original} &&\text{\bf ensemble} \nonumber \\
&\mathcal{F}^{(k)} = d^{2(k-1)}, &&F^{(k)} = d^{2k}, &&\mathcal{E}=\{I\},\\
&\mathcal{F}^{(k)} = \frac{k!}{d^{2}}, &&F^{(k)} = k!,  &&\mathcal{E}=\text{Haar}.
\end{align}

Next, we will analyze the properties of $\mathcal{F}_{\mathcal{E}}^{(1)}(\rho)$ for arbitrary $\rho$. First, we compute the Haar value:
\begin{align}
\mathcal{F}_{\text{Haar}}^{(1)}(\rho) = \frac{\tr \{ \rho^{2} \}}{d}.
\end{align}
Thus, it behaves differently on pure states vs. mixed states, e.g.
\begin{align}
 &\mathcal{F}^{(1)}_{\text{Haar}} = \frac{1}{d},  &&\text{(pure state)}, \\
 &\mathcal{F}^{(1)}_{\text{Haar}} = \frac{1}{d^2}, &&\text{(maximally mixed state)}.
\end{align}
For $\rho = |\psi\rangle \langle \psi |$ with $|\psi\rangle = |0\rangle^{\otimes n}$, we find that random Pauli-$X$ operators, operators that can either act as $X$ or $I$ on each qubit, achieve the Haar value. To see this, we can evaluate $\tr \, \big\{|\psi\rangle\langle\psi|UV^{\dagger} \big\}$. Since we choose $U$ and $V$ randomly from a set of Pauli-$X$ operators, then $UV^{\dagger}|\psi\rangle$ will be given by a random product state $|a_{1},\ldots,a_{n}\rangle$ with the $a_{j}=0,1$. Since $\tr \, \big\{|\psi\rangle\langle\psi|UV^{\dagger} \big\} =\langle 0,\ldots,0 |a_{1},\ldots,a_{n}\rangle$, the non-zero contribution has probability $1/d$. This means that $\mathcal{F}_{\text{Pauli-$X$}}^{(1)}=1/d$, which is the Haar value. Finally, as before we can use the generalized frame potential to bound the size of the ensemble
\begin{align}
\mathcal{F}^{(1)}_{\mathcal{E}}(\rho) = \frac{1}{|\mathcal{E}|^{2}} \sum_{U,V\in \mathcal{E}} |\tr \{ \rho \, UV^{\dagger} \}|^{2} \geq \frac{1}{|\mathcal{E}|} |\tr \{ \rho \}|^{2} = \frac{1}{|\mathcal{E}|} \quad \rightarrow \quad |\mathcal{E}| \geq \frac{1}{\mathcal{F}^{(1)}_{\mathcal{E}}(\rho)}.
\end{align}
Note that random Pauli-$X$ operators for $\rho=|0\rangle \langle 0 |$ are tight in terms of this lower bound since there are exactly $d$ different Pauli $X$ operators. 

Next, we will analyze the properties of $\mathcal{G}_{\mathcal{E}}^{(1)}(\rho)$. As before, first we compute the Haar value of $\mathcal{G}_{\mathcal{E}}^{(1)}(\rho)$:
\begin{align}
\mathcal{G}_{\text{Haar}}^{(1)}(\rho)=\frac{\tr \{ \rho \}^2}{d^2} = \frac{1}{d^2}.
\end{align}
Note that the result does not depend on $\rho$. The cardinality bound gives
\begin{align}
|\mathcal{E}| \geq \frac{1}{\mathcal{G}^{(1)}}.
\end{align}
Thus, to become a $k$-design requires $|\mathcal{E}|\geq d^2$. Note that this is a tight lower bound since a set of all the Pauli operators consists of $d^2$ different operators. Therefore, the minimization of $\mathcal{G}^{(1)}$ requires a $1$-design, regardless of $\rho$. 

The fact that all these properties of $\mathcal{G}^{(1)}$ are independent of the state $\rho$, in conjugation with Lemma~\ref{lemma-generalized-F}, suggests that $\mathcal{F}^{(k)}$ is the more interesting generalization of the frame potential.
As a result, we will continue by focusing only on $\mathcal{F}^{(k)}$.

We continue by considering $\mathcal{F}_{\text{Haar}}^{(2)}$. The Haar value is
\begin{align}
\mathcal{F}^{(2)}_{\text{Haar}}(\rho) = \frac{2\, \tr \, \{ \rho \}^2}{d^2-1} - \frac{2\, \tr \, \{ \rho^{2}\}}{d(d^2-1)}.
\end{align}
We can evaluate this for a maximally mixed state $\rho = I/d$ and for a pure state $\rho = |\psi\rangle \langle \psi|$, finding
\begin{align}
&\mathcal{F}^{(2)}_{\text{Haar}}=\frac{2}{d^2},  &&\text{(maximally mixed state)}. \\
&\mathcal{F}^{(2)}_{\text{Haar}}=\frac{2}{d(d+1)},  &&\text{(pure state)}.
\end{align} 
On the other hand, the maximum is given by evaluating $\mathcal{F}^{(2)}$ for a trivial ensemble, i.e. $\mathcal{E}=\{I\}$
\begin{align}
&\mathcal{F}_{\{I\}}^{(2)} =  \tr \, \{ \rho^{1/2} \}^4 = d^2,&&\text{(maximally mixed state)}.\\
                     &\mathcal{F}_{\{I\}}^{(2)}   = 1,  &&\text{(pure state)}.
\end{align}
While the minimal values of $\mathcal{F}^{(2)}_{\mathcal{E}}(\rho)$ are of the same order regardless of the state $\rho$, we see that the maximum values are significantly different. 

Finally, there is an easy way to compute the minimal value of the generalized frame potential $\mathcal{F}^{(k)}$ for any pure state e.g. $\rho=|0\rangle \langle 0|$. Let us write the frame potential as follows
\begin{align}
\mathcal{F}_{\mathcal{E}}^{(k)}(|0\rangle \langle 0|) = \frac{1}{|\mathcal{E}|^{2}}\sum_{U,V} \big[\langle 0|UV^{\dagger} |0 \rangle \langle 0|VU^{\dagger} |0 \rangle\big]^{k}.
\end{align}
Next, we define an ensemble of wavefunctions
\begin{align}
\mathcal{E}_{|0\rangle} = \{ U|0\rangle:U \in \mathcal{E} \}.
\end{align}
The frame potential is given by expression
\begin{align}
\mathcal{F}_{\mathcal{E}}^{(k)}(|0\rangle \langle 0|) = \frac{1}{|\mathcal{E}|^{2}}\sum_{|\psi\rangle,|\phi\rangle \in \mathcal{E}_{|0\rangle}}
|\langle \psi|\phi \rangle|^{2k}.
\end{align}
This is minimized when $\mathcal{E}_{|0\rangle}$ forms a $k$-design as an ensemble of quantum states. To obtain the Haar value, recall Eq.~\eqref{haar-random-states-equation} for the $k$-fold average of a Haar random state
\begin{align}
\int_{\text{Haar}} (|\psi\rangle \langle \psi|)^{\otimes k} d\psi
= \frac{\Pi_{\text{sym}}}{\left(
\begin{array}{c}
k+d-1\\
k
\end{array}
\right)},
\end{align}
where $
\Pi_{\text{sym}} = \frac{1}{k!} \sum_{\pi\in S_{k}} W_{\pi}$.
From this, we see that the Haar value of $\mathcal{F}^{(k)}$ evaluated on a pure state will be
\begin{align}
\mathcal{F}_{\text{Haar}}^{(k)}(|0\rangle \langle 0|) = \left(
\begin{array}{c}
k+d-1\\
k
\end{array}
\right)^{-1}.
\end{align}
In summary, we have 
\begin{align}
\begin{array}{cc|cc}
\mathcal{F}_{\mathcal{E}}^{(k)}(\rho) & & \rho=|0\rangle\langle 0| \quad &  \rho=\frac{I}{d}  \\ 
 &&&\\
\hline
&&& \\ 
\mathcal{E}=\{ I \} & & 1 & d^{2(k-1)}\\ 
&&& \\
\mathcal{E}=\text{Haar}&  & \left(
\begin{array}{c}
k+d-1\\
k
\end{array}
\right)^{-1} & 
\begin{array}{c}
k!\\ \hline
d^2
\end{array}  
\end{array}
\end{align}
At large $d$, the Haar value for pure states averaged over the Haar ensemble scales like $\sim k! / d^k$, which for $k>2$ is much smaller than for the Haar average of a mixed state. It's also interesting to point out that if we consider the time average of the generalized frame potential on a pure state $\ket{+}$, this is exactly equal to $k! / d^k$. Thus, in comparison to \S\ref{sec:F-time-average}, the time average of this generalized frame potential can actually obtain an almost Haar value with respect to $\ket{+}$. 

\subsection*{Thermal frame potentials}

The above generalization is incorrect for a system in a thermal state that evolves in time by an ensemble of Hamiltonians. This is because the thermal state $\rho_{\beta} = e^{-\beta H} / \tr \, \{ e^{-\beta H} \}$ depends on the Hamiltonian and thus the ensemble itself. Instead, we will define a \emph{thermal} frame potential $\mathcal{W}^{(k)}_{\beta}(t)$ by taking an average over the squares of $2k$-point OTO correlators as the definition. Define a collection of thermally regulated $2k$-point OTO correlators as 
\begin{align}
\langle A_{1}(0)B_{1}(t)\cdots A_{k}(0)B_{k}(t)\rangle =   \ \frac{\tr \, \big\{ e^{-\frac{\beta H}{2k}} A_{1} \,  e^{-\frac{\beta H}{2k}}  B_{1}(t) \dots  e^{-\frac{\beta H}{2k}} \, A_{k} \,  e^{-\frac{\beta H}{2k}} B_{k}(t) \big\}}{\tr \, \{e^{-\beta H} \}} .
\end{align}
By first averaging over Hamiltonians and then taking an average over all choices of operators of the norm squared of these correlators, we find 
\begin{align}
\mathcal{W}^{(k)}_{\beta}(t)=  \iint dG dH \ \frac{ 
\big|\tr \, \{e^{-(\beta/2k-it )G } e^{-(\beta/2k+it )H }\}\big|^{2k}
}{\tr \, \{e^{-\beta G} \}\tr \, \{e^{-\beta H} \}}.
\end{align}
For $\beta=0$ (infinite temperature) the thermal frame potential reduces to the original frame potential $F^{(k)}$ up to a factor of proportionality
\begin{align}
\mathcal{W}^{(k)}_{\beta=0}(t)= \frac{1}{d^2}\iint dG dH \ 
\big| \tr \, \big\{ e^{iGt}e^{-iHt} \big\} \big|^{2k} = \frac{1}{d^2} F^{(k)} = \mathcal{F}^{(k)}\Big(\frac{I}{d}\Big),
\end{align}
where $F^{(k)}$ is the original frame potential, $\mathcal{F}^{(k)}(I/d)$ is the generalized frame potential from the previous section, and both are evaluated on the ensemble $\mathcal{E}(t)=\{e^{-iHt}\}$.

We can also use $\mathcal{W}^{(k)}_{\beta}(t)$ to make a bound on the cardinality of the ensemble. For simplicity, let's assume that the ensemble $\mathcal{E}(t)$ is a discrete set of $e^{-iHt}$ taken with a uniform weight. This let's us make a bound
\begin{align}
\mathcal{W}^{(k)}_{\beta}(t) \geq \frac{1}{|\mathcal{E}(t)|^2} \sum_{H}  \frac{ 
\tr \, \big\{e^{-\beta H /k }  \big\}^{2k}
}{\tr \, \{e^{-\beta H} \}^2},
\end{align}
or in terms of the cardinality, 
\begin{align}
|\mathcal{E}(t)|^2 \geq
 \frac{1}{\mathcal{W}^{(k)}_{\beta}(t) }\sum_{H} \ \frac{ 
\tr \, \big\{e^{-\beta H /k }  \big\}^{2k}
}{\tr \, \{e^{-\beta H} \}^2}.
\end{align}
For $k=1$, the bound takes a particularly simple form
\begin{align}
|\mathcal{E}(t)| \geq
 \frac{1}{\mathcal{W}^{(k=1)}_{\beta}(t) }.
\end{align}

One nice property of the thermal frame potential is that even when $t=0$ we may obtain a non-trivial bound. For example, with $k=1$ we can express the thermal frame potential as
\begin{align}
\mathcal{W}^{(1)}_{\beta}(0)= \iint dG dH \ 
\frac{ 
\tr \, \big\{e^{-\frac{\beta}{2} G } e^{-\frac{\beta}{2} H }\big\}^{2}
}{\tr \, \{e^{-\beta G} \}\tr \, \{e^{-\beta H} \}}.\label{eq:thermal-t0-bound}
\end{align}
If $\beta=0$, the bound is trivial $|\mathcal{E}(0)|\geq 1$. However, if $\beta >0$ the thermal frame potential can be smaller than unity, which gives a non-trivial lower bound on $|\mathcal{E}(0)|$! We can see this straightforwardly by applying the Cauchy-Schwartz inequality to the integrand of Eq.~\eqref{eq:thermal-t0-bound}
\begin{align} 
\frac{ 
\tr \, \big\{e^{-\frac{\beta}{2} G } e^{-\frac{\beta}{2} H }\big\}^{2}
}{\tr \, \{e^{-\beta G} \}\tr \, \{e^{-\beta H} \}} \leq 1.
\end{align}
The intuition for such a lower bound at $t=0$ comes from the fact that $e^{-\beta G},e^{-\beta H}$ contain imaginary time evolution. We suspect that this will also let us bound the ``complexity of formation'' of such states $\rho_{\beta}$ with respect to an ensemble of Hamiltonians.

\section{More chaos in quantum channels}\label{sec:more-chaos}

Finally, in this appendix we expand on some of the results of \cite{Hosur:2015ylk} that are otherwise somewhat outside the main focus of the current paper: 
\begin{itemize}
\item In \S\ref{sec:more:HP}, we revisit the black hole thought experiment of Hayden and Preskill \cite{Hayden07} and discuss its relationship to the decay of OTO four-point functions.
\item Next, in \S\ref{sec:more:op} we provide an operational meaning to taking averages over OTO correlation functions as introduced in \cite{Hosur:2015ylk}. The argument is partly inspired by a course taught by Kitaev~\cite{Kitaev-Haar}.
\item Finally, in \S\ref{sec:more:renyi} we generalize the relationship between OTO correlators and R\'enyi entropy obtained in \cite{Hosur:2015ylk} relating an average over $2k$-point functions and to an expression involving the $k$th R\'enyi entropy.
\end{itemize}

\subsection{OTO correlators and black holes as mirrors}\label{sec:more:HP}

In the conventional black hole information thought experiment, Alice throws some quantum state (or herself!) into a black hole, the system evolves by some chaotic unitary operator $U$, and then Bob attempts to reconstruct Alice's quantum state (or Alice!) by collecting the outgoing Hawking radiation emitted by the black hole.

Hayden and Preskill \cite{Hayden07} added an interesting twist to this classic setup by assuming that Bob knows both the initial state and the dynamics $U$ of the black hole. In this scenario, they showed that if the dynamics $U$ are sampled from a $2$-design, Alice's $m$-qubit quantum state can be immediately reconstructed by collecting $m+\epsilon$ qubits of the Hawking radiation.
The black hole acts as if it is a quantum information mirror, ``reflecting'' Alice's state almost immediately.

To be more precise, let's split the input state into subsystems $A$ and $B$, and let's split the output state into different subsystems $C$ and $D$. We let $A$ represent Alice or her input quantum state, $B$ represents initial black hole state, $C$ represents the remaining black hole state after some evolution and evaporation, and $D$ represents the emitted Hawking radiation. The black hole dynamics $U$ take $AB$ to $CD$. This setup is shown in Fig~\ref{fig_unitary_state}. In the conventional setting, Bob  has an access to only $D$. However, in the Hayden-Preskill modification, Bob has access to both $B$ and $D$.

\begin{figure}[htb!]
\centering
\includegraphics[width=0.4\linewidth]{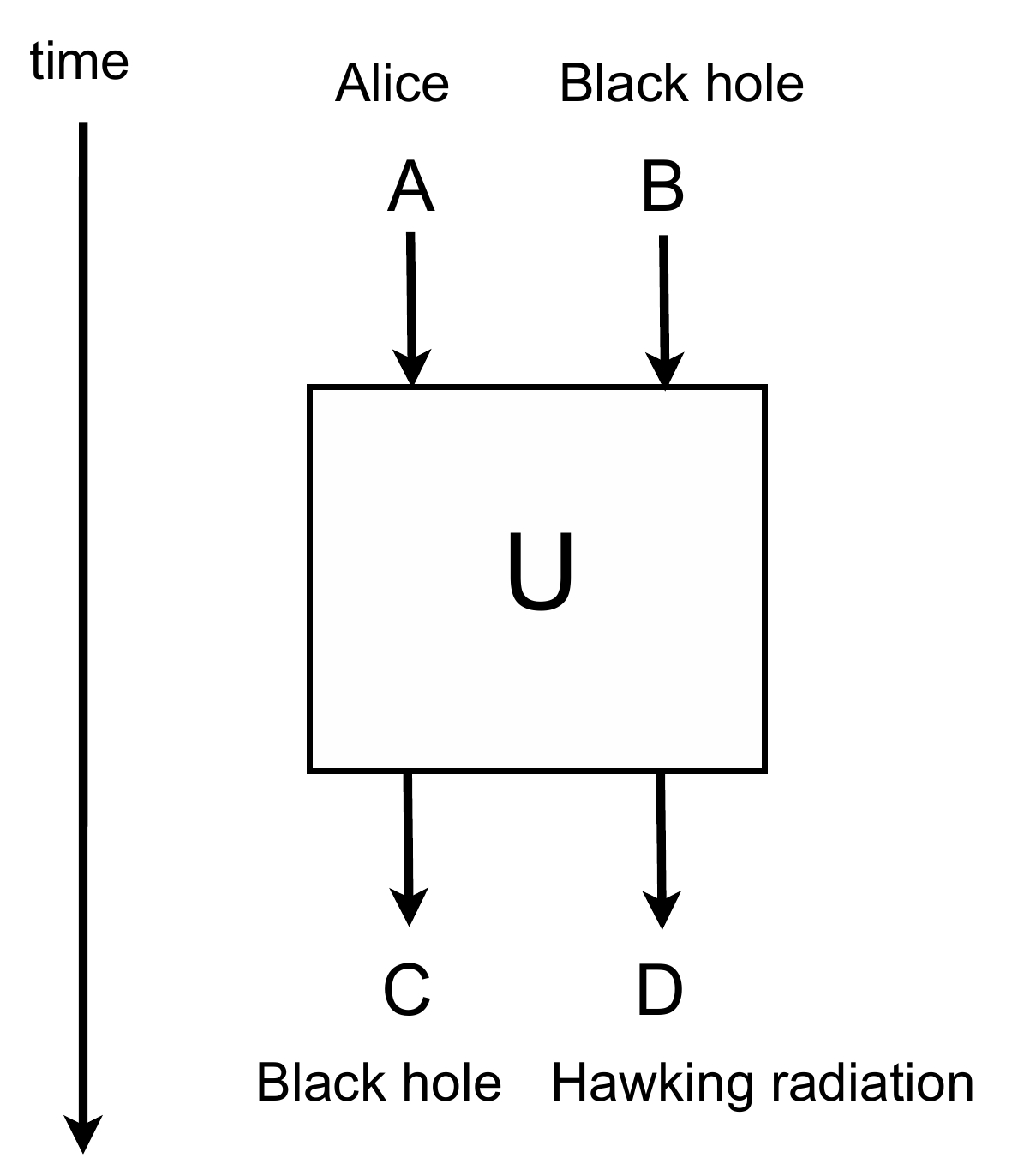}
\caption{A unitary operator $U$ representing the dynamics in the Hayden-Preskill thought experiment.
} 
\label{fig_unitary_state}
\end{figure}

The mirror phenomenon can be better understood by interpreting the unitary operator $U$ as a state via the Choi-Jamilkowski isomorphism~\cite{Hosur:2015ylk}
\begin{align}
U = \sum_{i,j}u_{ij}|i\rangle \langle j | \quad \rightarrow \quad  |U\rangle =  \frac{1}{2^{n/2}} \sum_{i,j}u_{ij}|j\rangle \otimes |i \rangle.\label{eq:channel-state-iso}
\end{align}
The original operator $U$ acts unitarily on $n$-qubit states, while the state $|U\rangle$ is defined on a $2n$-qubit Hilbert space. This interpretation lets us compute entropies and informations between the input and the output. In particular, whether Bob is able to reconstruct Alice's quantum state $A$ can be quantified by the mutual information
\begin{align}
I(A:BD) = S_{A} + S_{BD} - S_{ABD},
\end{align}
where e.g. $S_A = -\tr \, \{ \rho_A \log \rho_A \}$ is the entanglement entropy evaluated on the density matrix $\rho_A =\tr_{BCD}\, \{\ket{U}\bra{U} \}$.
If $I(A:BD)$ is near its maximal value $2S_A$, then Bob can reconstruct Alice's unknown quantum state.\footnote{Here we assumed $S_A\leq S_B,S_D$. For a more precise criteria, see Ref.~\cite{Hayden07}.} It's easy to show that for a $2$-design unitary operator $U$ that the mutual information is close to its maximum $I(A:BD) \approx 2 S_A$, thus enabling reconstruction of Alice's quantum state \cite{Hayden07}.

The information reconstruction problem is closely related to four-point OTO correlators. It was shown in~\cite{Hosur:2015ylk} that four-point OTO correlation functions averaged over all Pauli operators in $A$ and $D$ is related to the second R\'enyi entropy by the formula
\begin{align}
\int dA dD \, \langle A(0)D(t)A^{\dagger}(0)D^{\dagger}(t) \rangle = \frac{d}{d_{A}d_{D}} 2^{- S_{AC}^{(2)}}, \label{eq:answer}
\end{align}
where $\int dA dD$ means an average over all Pauli operators in $A$ and $D$, $S_{AC}^{(2)}$ is the R\'{e}nyi-$2$ entropy for the joint region $AC$, and $d_{A}=2^{S_A}$ is the Hilbert space dimension of $A$, etc. From a simple calculation, we obtain
\begin{align}
- \log_{2} \bigg\{ \int dA dD \, \langle A(0)D(t)A^{\dagger}(0)D^{\dagger}(t) \rangle \bigg\} = I^{(2)}(A:BD)  \leq I^{(1)}(A:BD),
\end{align}
where $I^{(2)}(A:BD)= S_{A}^{(2)} + S_{BD}^{(2)} - S_{ABD}^{(2)}$ is the R\'{e}nyi-$2$ mutual information. This inequality holds since $S_{A}=S_{A}^{(2)}$, $S_{ABD}=S_{ABD}^{(2)}=S_C$, and $S^{(2)}_{BD}\leq S_{BD}$. The decay of all the four-point OTO correlators between Alice $A$ and the Hawking radiation $D$ implies a strong correlation between Alice $A$ and $BD$ (the subsystem that Bob has access to), thus establishing a direct link between four-point OTO correlators and the information reconstruction problem.

\subsection{Operational interpretation}\label{sec:more:op}

In this subsection, we attempt to provide an operational interpretation to sum over out-of-time-order correlation functions discussed in \cite{Hosur:2015ylk} and in this work. 

Consider the following classical analogy. Alice and Bob are playing a game of catch, and Alice is about to throw the ball to Bob. However, Alice is known to cheat sometimes by perturbing the ball (she throws a mean spitball!), and so Bob would like to determine whether Alice has cheated. One way for Bob to check is to ask someone that he trusts, Charlie, to throw an identical ball to Bob. Then, Bob can compare the two balls and determine whether Alice has modified the ball. This setup is shown in Fig.~\ref{fig_catch_ball}. As we will explain, the average of all OTO correlators over the operators being correlated is closely related to the probability of Bob detecting Alice's perturbation.

\begin{figure}[htb!]
\centering
\includegraphics[width=0.75\linewidth]{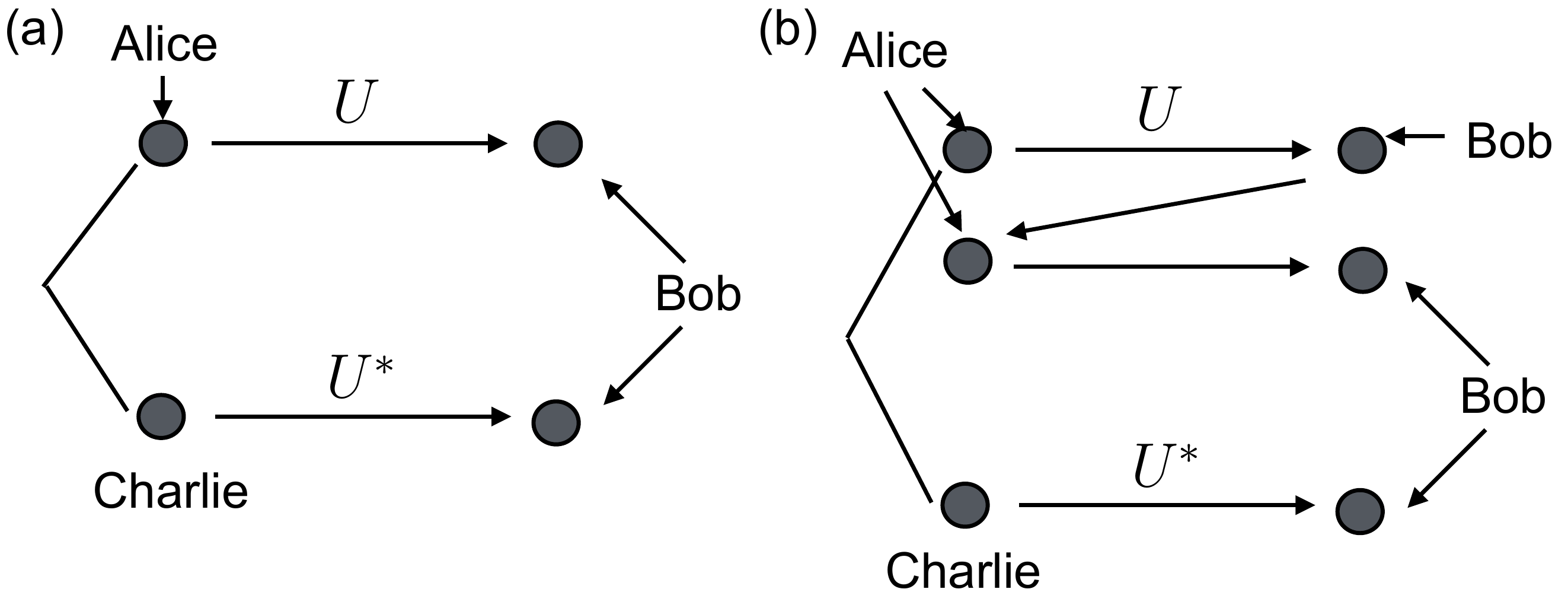}
\caption{
Alice and Bob play a game of catch, and Alice modifies the ball with some nonzero probability. {\bf (a)} Charlie throws an identical but unmodified ball to Bob. This represents an average of $4$-point OTO correlation functions. 
{\bf (b)} Bob throws the ball back to Alice who then throws it back to Bob. This represents an average of $8$-point OTO correlaton functions.  
}
\label{fig_catch_ball}
\end{figure}

In a quantum setting, we prepare the state $\ket{EPR}$ of $n$ EPR pairs and give Alice one half of all the pairs so that Alice's density matrix $\rho_{A}$ is a maximally mixed state. We give the other half of the pairs to Charlie. Thus, Charlie has a perfectly correlated copy of Alice's initial state for Bob to later compare against. Now, Alice applies a perturbation, the Pauli operators $A_{i}$ with probability $p_i$, which
 corresponds to the following superoperator
\begin{align}
\sum_{i} p_i \, A_{i}(\cdot)A_{i}^{\dagger}.
\end{align}
Of course, included in this set is the identity operator $I$ and an associated probability $p_I$ that Alice does not modify her state.\footnote{If $p_{I}=0$, then Bob would already know that Alice always applies a perturbation.}

Next, Alice throws her ball to Bob by applying $U$, and Charlie throws his ball to Bob by applying $U^{*}$.\footnote{The conjugation appears because Charlie's copy of Alice's original state is actually a CPT conjugate. The application of $U^{*}$ can be also interpreted as Bob's \emph{uncomputation}.} The overall state is described by the statistical ensemble
\begin{align}
\rho(p)=\sum_{j}p_j \,(U \otimes U^{*})(A_{i}\otimes I)|EPR \rangle \langle EPR |(A_{i}^{\dagger}\otimes I)(U^{\dagger} \otimes U^{T}).\label{eq:bobs-stat-ensemble}
\end{align}
Bob would like to compare the two quantum states from Alice and Charlie to determine whether Alice applied a perturbation. Note that if Alice did not apply a perturbation $A_{i}=I$, then the final state is equal to the initial state $|EPR\rangle$ because 
\begin{align}
U \otimes U^{*} |EPR\rangle  = |EPR\rangle, \qquad \forall U. \label{eq-U-on-epr-relation}
\end{align}
Therefore, Bob's strategy is to perform a projective measurement $\Pi = |EPR\rangle \langle EPR|$ on his state.
The projection operator $|EPR\rangle \langle EPR|$ can be represented as an average over all the Pauli operators on Bob's state:
\begin{align}
\frac{1}{2^{2n}}\sum_{B_{j} \in \text{Pauli}} B_{j} \otimes B_{j}^{*}  = |EPR\rangle \langle EPR|.
\end{align}
Thus, the probability of Bob's measuring of $|EPR\rangle$ is given by
\begin{align}
\frac{1}{2^{2n}}\sum_{i,j} p_i \, \langle EPR | (A_{i}^{\dagger}\otimes I) (U^{\dagger} \otimes U^{T}) (B_{j}\otimes B_{j}^{*})  (U \otimes U^{*})(A_{i}\otimes I)| EPR \rangle = p_{I},\label{eq-bob-measurement-outcome}
\end{align}
which is easy to check by making use of Eq.~\eqref{eq-U-on-epr-relation}. 
Eq.~\eqref{eq-bob-measurement-outcome} is depicted graphically in Fig.~\ref{fig_full_diagram}, which makes it clear that we can interpret it as an average over four-point OTO correlators of the form\footnote{Note that the result does not depend on the time evolution operator $U$ since Bob performed measurements on the whole system.}
\begin{align}
\frac{1}{2^{2n}}\sum_{i,j} p_i \,  \langle  \tilde{B}^{\dagger}_{j} A_{i}^{\dagger}  \tilde{B}_{j} A_{i}  \rangle = p_{I}.
\end{align}
In the Hayden-Preskill setup, Bob performs a projective measurement $|\text{EPR}\rangle \langle \text{EPR}|$ only on qubits on $D$. Then, assuming that Alice applies perturbations $A_{i}$ to qubits on $A$ with equal probabilities, the probability of Bob's measuring $|\text{EPR}\rangle$ on $D$ is exactly given by Eq.~\eqref{eq:answer}.

\begin{figure}[htb!]
\centering
\includegraphics[width=0.45\linewidth]{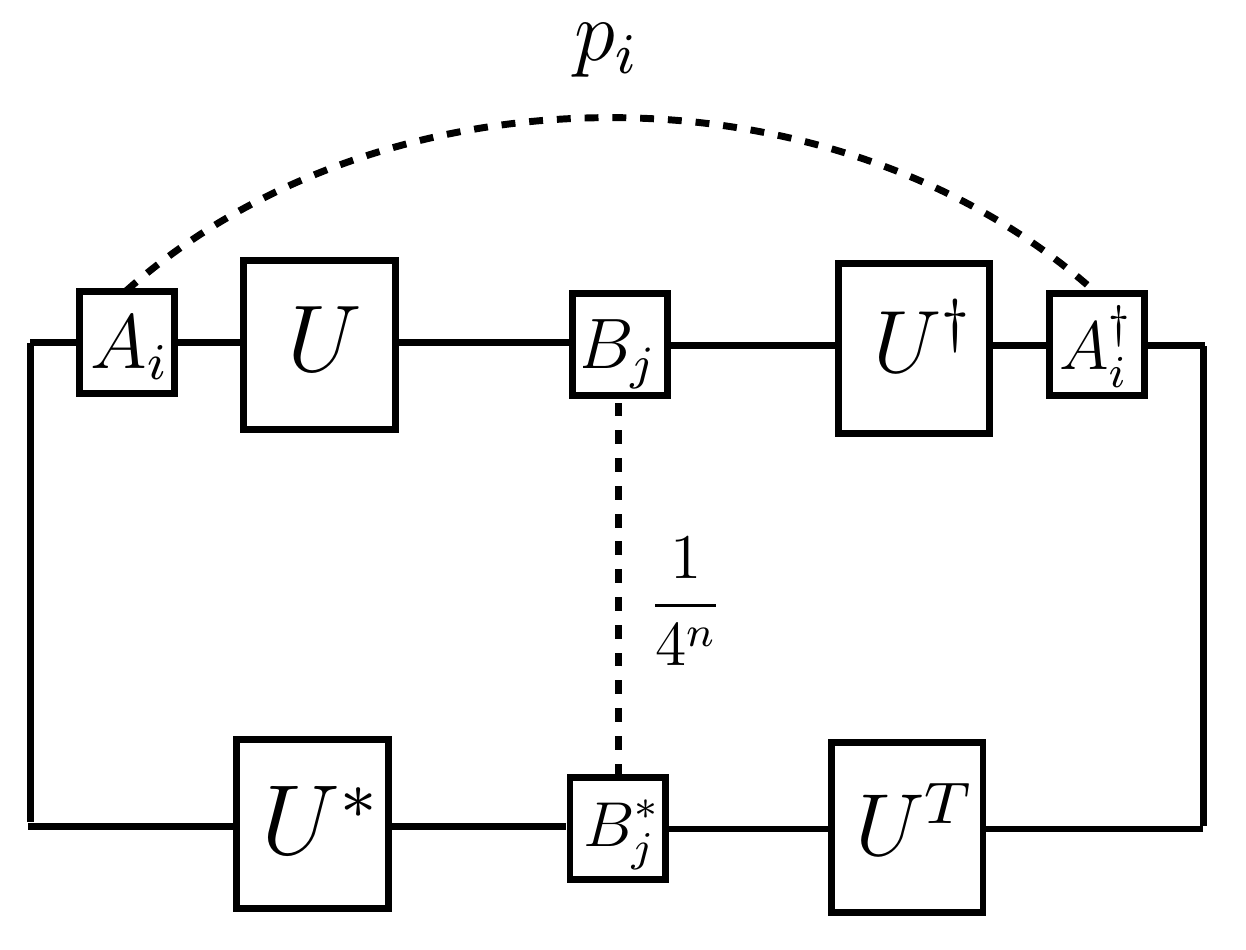}
\caption{A graphical depiction of Eq.~\eqref{eq-bob-measurement-outcome} equal to the probability that Alice did not apply a perturbation, $p_i = p_I$ for $A_i = I$. Starting from the top left and reading clockwise, we see that this corresponds to an average over all choices of $A_i, B_j$ of $4$-point OTO correlators of the form $\langle A_{i}\tilde{B}_{j}A_{i}^{\dagger}\tilde{B}_{j}^{\dagger}\rangle$. Note that the average over $A_i$ is with probabilities $p_i$, but the average over $B_j$ is uniform with probabilities $p_j = 2^{-2n}$.
} 
\label{fig_full_diagram}
\end{figure}

    This game of ``catch'' can be ``generalized'' by considering multiple rounds of throwing the ball back and forth between Alice and Bob. In this case, both Alice and Bob have the option of applying perturbations at any round of the game (except for the final time Bob receives the ball), but Charlie always faithfully throws a copy of the original ball to Bob for comparison at the end. This process turns out to be equal to an average over $4m$-point OTO correlators with the same ordering as those studied in \S\ref{sec:k-point-averages}.
    This setup, with $m=2$ rounds of catch, is shown in Fig~\ref{fig_catch_ball}(b).

\subsection{R\'{e}nyi-$k$ entropy}\label{sec:more:renyi}

Finally, we generalize the relationship between the second R\'enyi entropy and an average over four-point OTO correlators that was obtained in~\cite{Hosur:2015ylk}.\footnote{We thank Xiao-Liang Qi for providing useful insights leading to this generalization. Additionally, the possibility of such a generalization is briefly mentioned in~\cite{fan2016out}.}  We will consider $2k$-point OTO correlation functions of the form
\begin{align}
\left\langle A_{1}\tilde{D}_{1}A_{2}\tilde{D}_{2} \cdots A_{k}\tilde{D}_{k} \right\rangle = \frac{1}{d}\tr\left\{ A_{1}\tilde{D}_{1}A_{2}\tilde{D}_{2} \cdots A_{k}\tilde{D}_{k}\right\}. \label{eq:renyi-form-oto-to-average}
\end{align}
We assume that $A_{j}$ and $D_{j}$ only act on subsets of qubits, and we will denote these subregions by $A$ and $D$, respectively.

We are interested in taking an average of correlators of the form Eq.~\eqref{eq:renyi-form-oto-to-average}, where we will let one of the operators from each subregion $A_k, D_k$ depend on the other operators such that the average gives a permutation operator (see Eq.~\eqref{eq:pauli-decomposition-of-cyc} and the discussion in \S\ref{sec:review}).
In particular, to get a cyclic permutation $\pi_{\text{cyc}}=\pi_{23\ldots k 1}$, we will take an average in the following way
\begin{align}
\frac{1}{{d_{A}}^{2(k-1)}}
\sum_{\substack{A_{1},\ldots,A_{k-1}\in \mathcal{P}, \\ A_{k} = A_{1}^{\dagger}A_{2}^{\dagger}\cdots A_{k-1}^{\dagger} } }, \label{eq:form-of-the-average}
\end{align}
where ${d_{A}}^{2(k-1)}$ is the total number of Pauli operators $A_{1},\ldots,A_{k-1}$ that we average over. We will denote such an average as follows
\begin{align}
|\text{OTO}|_{\pi_{A},\pi_{D}}=\left| \frac{1}{d}\tr \left\{ A_{1}\tilde{D}_{1}A_{2}\tilde{D}_{2} \cdots A_{k}\tilde{D}_{k} \right\}  \right|_{\pi_{A},\pi_{D}},
\end{align}
where the subscript notation $\pi_{A},\pi_{D}$ mean that we took the average such that $A_{1}\otimes A_{2} \otimes \cdots \otimes A_{k}$ and $D_{1}\otimes D_{2} \otimes \cdots \otimes D_{k}$ form the permutation operators associated with $\pi_{A},\pi_{D}$. 

Now, let us represent $U$ as a state $\ket{U}$ by using the channel-state isomorphism Eq.~\eqref{eq:channel-state-iso} discussed in \S\ref{sec:more:HP}. We divide the input into subsystems $A$ and $B$, and we divide the output into subsystems $C$ and $D$. Our key result is a formula relating an average over $2k$-point OTO correlators with operators in $A$ and $D$ to the exponential of the $k$-th R\'enyi entropy of the subsystem $AC$ 
\begin{align}
|\text{OTO}|_{{\pi_{\text{cyc}}}^{-1},\pi_{\text{cyc}}}= \left(\frac{d}{d_{A}d_{D}}\right)^{k-1} 2^{-(k-1)S_{AC}^{(k)}},
\end{align}
with $S_{AC}^{(k)}$ the R\'enyi $k$-entropy of $AC$.  To derive this, we also use the following relation
\begin{align}
\tr\left\{A_{1}\tilde{D}_{1} \cdots A_{k}\tilde{D}_{k} \right\}
= \tr\left\{ (A_{1}\otimes \cdots \otimes A_{k})\cdot (\tilde{D}_{1}\otimes \cdots \otimes \tilde{D}_{k}) \cdot W_{\pi_{\text{cyc}}} \right\},
\end{align}
discussed in the beginning of \S\ref{sec:OTO_channel}, or graphically
\begin{align}
\includegraphics[width=0.6\linewidth]{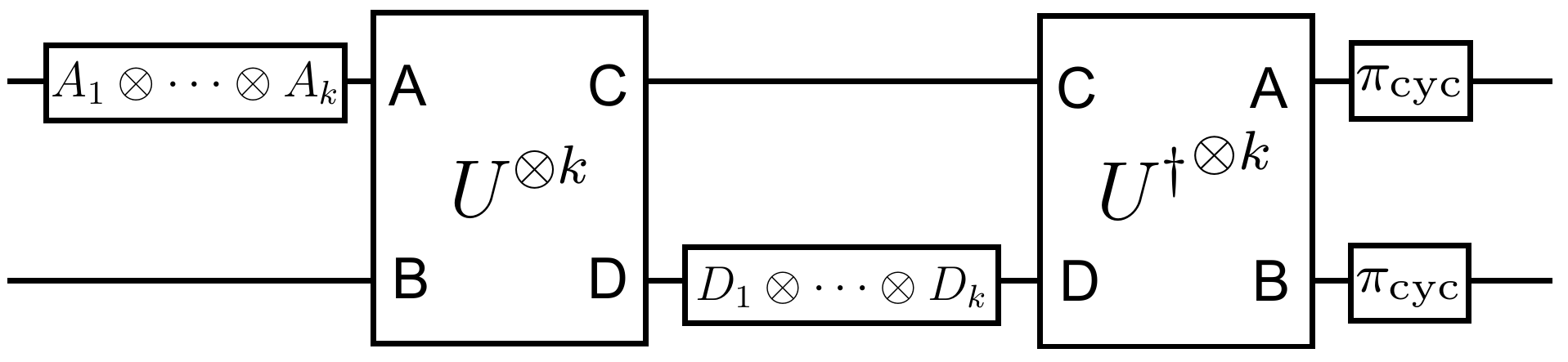}.
\end{align}
From Eq.~\eqref{eq:pauli-decomposition-of-cyc} we know that taking an average of the form Eq.~\eqref{eq:form-of-the-average} replaces $(A_{1}\otimes \cdots \otimes A_{k})$ and $(\tilde{D}_{1}\otimes \cdots \otimes \tilde{D}_{k})$  by permutation operators $W_{\pi_{\text{cyc}^{-1}}}$ and $W_{\pi_{\text{cyc}}}$, which act on $A$ and $D$, respectively. Graphically, we have 
\begin{align}
\includegraphics[width=0.85\linewidth]{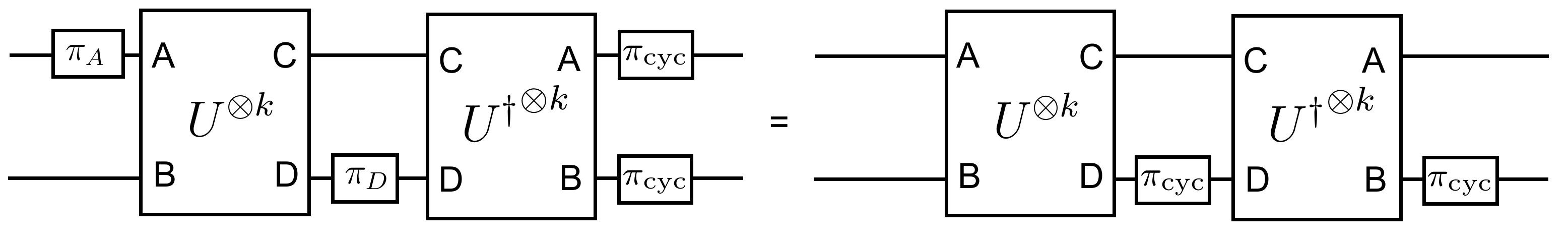}
\end{align}
where $\pi_{A}=\pi_{\text{cyc}}^{-1}$ and $\pi_{D}=\pi_{\text{cyc}}$. Careful consideration of this picture elucidates that it is proportional to $\tr \{ {\rho_{AC}}^{k}\}$.

\mciteSetMidEndSepPunct{}{\ifmciteBstWouldAddEndPunct.\else\fi}{\relax}
\bibliographystyle{utphys}
\bibliography{design}{}
\end{document}